\documentclass{article}
  \usepackage[main,final,preprint]{neurips_2026}

\usepackage[utf8]{inputenc} 
\usepackage[T1]{fontenc}    
\usepackage{hyperref}       
\usepackage{url}            
\usepackage{booktabs}       
\usepackage{amsfonts}       
\usepackage{nicefrac}       
\usepackage{microtype}      
\usepackage{xcolor}         

\usepackage{multirow}

\usepackage{graphicx} 
\usepackage{geometry}
\usepackage{amssymb,amsmath}
\usepackage{mathtools,amsthm}
\usepackage{thm-restate}
\usepackage{color,todonotes,xspace}
\usepackage{colonequals}
\usepackage{stmaryrd}
\usepackage{needspace}

\newcommand{\llangle}{\langle\!\langle}
\newcommand{\rrangle}{\rangle\!\rangle}
\newcommand{\llbrace}{\lbrace \!\! \lbrace}
\newcommand{\rrbrace}{\rbrace \!\!\rbrace}
\newcommand{\bigllbrace}{\big\lbrace \!\! \big\lbrace}
\newcommand{\bigrrbrace}{\big\rbrace \!\!\big\rbrace}

\newtheorem{theorem}{Theorem}
\newtheorem{lemma}[theorem]{Lemma}
\newtheorem{corollary}[theorem]{Corollary}
 
\newtheorem{definition}[theorem]{Definition}
\newtheorem{example}[theorem]{Example}

\newtheorem{remark}[theorem]{Remark}

\newcommand{\cls}{\ensuremath{\mathsf{cls}}\xspace}
\newcommand{\agg}{\ensuremath{\mathsf{agg}}\xspace}
\newcommand{\sumagg}{\ensuremath{\mathsf{sum}}\xspace}
\newcommand{\meanagg}{\ensuremath{\mathsf{mean}}\xspace}
\newcommand{\maxagg}{\ensuremath{\mathsf{max}}\xspace}
\newcommand{\com}{\ensuremath{\mathsf{com}}\xspace}
\newcommand{\Hom}{\ensuremath{\mathsf{Hom}}\xspace}
\newcommand{\Inj}{\ensuremath{\mathsf{Inj}}\xspace}
\newcommand{\res}{\ensuremath{\mathsf{eval}}\xspace}

\newcommand{\Emb}{\ensuremath{\mathsf{Emb}}\xspace}

\DeclareMathOperator{\ind}{ind}
\DeclareMathOperator{\pred}{atom}

\renewcommand{\max}{\ensuremath{\mathsf{max}}\xspace}
\renewcommand{\mod}{\ensuremath{\mathop{\mathsf{mod}}}\xspace}
\newcommand{\mean}{\ensuremath{\mathsf{mean}}\xspace}

\newcommand{\relu}{\textup{ReLU}\xspace}
\newcommand{\trrelu}{\ensuremath{\textup{ReLU}^*}\xspace}

\newcommand{\calC}{\mathcal{C}}
\newcommand{\calF}{\mathcal{F}}

\newcommand{\calK}{\mathcal{K}}
\newcommand{\calL}{\mathcal{L}}
\newcommand{\calN}{\mathcal{N}}
\newcommand{\calP}{\mathcal{P}}

\newcommand{\mn}[1]{\ensuremath{\mathsf{#1}}}

\usetikzlibrary{decorations.pathmorphing, calligraphy,calc}

\usepackage{tikz}
\newcommand{\sun}[1][0.4]{%
  \mathord{%
    \vcenter{\hbox{%
      \begin{tikzpicture}[
        scale=#1,
        baseline=-0.5ex,
        line width=0.6pt
      ]
        \draw (0,0) circle (0.12);
        \foreach \a in {90,150,210,270,330,30} {
          \draw (\a:0.18) -- (\a:0.34);
        }
      \end{tikzpicture}%
    }}%
  }%
}

\def\email#1{\texttt{\small #1}}

\title{Expressive Power of Deep Homomorphism Networks over Relational Databases}
\author{
Balder ten Cate\\
University of Amsterdam\\
\email{b.d.tencate@uva.nl}
\And
Maurice Funk \\
Leipzig University\\
\email{maurice.funk@uni-leipzig.de}
\And
Benny Kimelfeld \\
Technion \& RelationalAI\\
\email{bennyk@cs.technion.ac.il}
\AND 
Carsten Lutz\\
Leipzig University\\
\email{clu@informatik.uni-leipzig.de}
\And
Moritz Schönherr\\
Leipzig University\\
\email{schoenherr@informatik.uni-leipzig.de}
\And
Arie Soeteman\\
University of Amsterdam\\
\email{a.w.soeteman@uva.nl}}

\begin{document}
\maketitle

\begin{abstract}
  The  expressive limitations of message-passing Graph Neural Networks (GNNs) have motivated a wide range of  more powerful graph learning architectures. We advocate Deep Homomorphism Networks (DHNs) as a  model  particularly well-suited for learning over relational databases,
  due to their close connection to important fragments of SQL such as conjunctive queries.
  We study the precise expressive power of DHNs
  by relating them to various natural fragments and extensions of first-order logic (FO). For DHNs with $\maxagg$, $\sumagg$, and $\meanagg$ aggregation, we establish connections to the unary negation fragment (UNFO)
  and to the extensions of UNFO with counting quantifiers and  with ratio quantifiers.
We further relate sum-aggregation DHNs to the unary quantifier alternation fragment of FO and to a fragment of FO extended with expressive counting.
  Through the classical
  correspondence between FO and SQL, these results also illuminate the relation between DHNs and SQL. They further enable us to study the decidability of fundamental static analysis problems for DHNs,  the emptiness problem and the subsumption problem.
  We confirm  through experiments that the established differences in expressive power 
are reflected in the performance on suitable prediction tasks.

%
%
\end{abstract}

\section{Introduction}

Message-passing Graph Neural Networks (GNNs) are a prominent class of graph learning models that have proven highly effective across a wide range of domains, including social network analysis~\cite{DBLP:journals/tii/GuoW21}, molecular property prediction~\cite{DBLP:conf/nips/SunDY22}, recommender systems~\cite{DBLP:journals/csur/SharmaLNSSKK24}, and data analytics for relational databases~\cite{DBLP:conf/icml/FeyHHLR0YYL24}. 
However, it is well known that the expressive power of the standard GNN model is bounded from above by the one-dimensional Weisfeiler--Leman test.
This is a serious limitation:
in essence, GNNs can only detect tree-shaped graph properties whereas even 
simple and common  structures such as triangles, cliques, and
fixed-length cycles lie beyond their expressive power.
Consequently, a growing body of work has sought to design more expressive graph neural architectures, including higher-order GNNs \cite{morris2019weisfeiler},  subgraph GNNs \cite{Qian,You2021,Zeng2023}, and  hierarchical ego  GNNs \cite{soeteman2025logical}.
These have exhibited improved performance on 
many tasks such as predicting constrained solubility of molecules (as in the ZINC benchmark \cite{irwin2012zinc}),  where substructure features such as triangles and rings are highly relevant.

One important and practical way to increase the expressive power of GNNs, breaking free from their tree-shaped limitations, is to provide the ability to count homomorphisms from selected graph patterns. This is often done by feature augmentation: 
every vertex is adorned with homomorphism counts that capture local graph structure, such as the number of triangles it lies on~\cite{DBLP:conf/iclr/BaoJBCL25,barcelo2021graph,DBLP:conf/icml/JinBCL24,DBLP:conf/icml/NguyenM20,DBLP:journals/corr/abs-2308-15283}.  
However, homomorphism counts can also be integrated into the GNN architecture more deeply, replacing the standard edge-based message passing mechanism of GNNs by a 
more general form of message passing along complex patterns. 
The latter, more powerful approach is taken by Deep Homomorphism Networks (DHNs), recently introduced in \cite{maehara2024deep}. As shown in this article,  DHNs offer strictly higher expressive power than standard GNNs even when the latter are augmented with
homomorphism-count based features. Moreover, we argue that DHNs have another important advantage: they provide a natural connection to relational databases. Indeed, a vast amount of  data is stored in relational systems and analyzing such data is an important application of GNNs, usually approached by first translating the database into a graph~\cite{DBLP:conf/icml/FeyHHLR0YYL24,
DBLP:conf/icml/ChenKL25,DBLP:conf/icde/TonshoffFGK23,
DBLP:conf/nips/WangGWZFZZCLMST24}. For DHNs, such a translation is unnecessary as their definition  extends gracefully from graphs to databases. What is more, pattern homomorphisms are essentially  an alternative presentation of conjunctive queries (CQs), arguably the most important query language for relational database systems. We recall that CQs correspond to the select-project-join fragment of SQL and that in practice a large fraction of  SQL queries fall into this fragment~\cite{DBLP:conf/pods/Chaudhuri98}. Thus, DHNs provide a natural form of expressive power for database properties while also opening up the possibility to benefit from decades of work on the efficient evaluation of CQs.

{\bf Our contributions.} The aim of this paper is to provide a careful analysis of the exact expressive power of DHNs used as vertex classifiers, with a particular focus on their relation to natural fragments and extensions of first-order logic (FO). Since an important core of SQL famously corresponds to FO, this also helps clarify which SQL queries are expressible by DHNs and which are not. In addition, it enables us to obtain results about the static analysis of DHNs. We study the \emph{emptiness problem}, which asks whether a given DHN defines the empty set of vertices, thereby indicating a potential flaw in the model, and the related \emph{subsumption problem}, 
which asks whether the vertex classifier realized by one given DHN is subsumed by that realized by another. Our logical characterizations 
also allow us to prove that DHNs are strictly more expressive than GNNs augmented with initial homomorphism count features, as proposed by Barceló et al.~\cite{barcelo2021graph}.

We study DHNs with $\maxagg$, $\sumagg$, and $\meanagg$ aggregation. For $\maxagg$-DHNs, we obtain a complete and uniform logical characterization:  the database properties expressible by $\maxagg$-DHNs are exactly those definable in the unary negation fragment of FO (UNFO) \cite{DBLP:journals/corr/SegoufinC13}. By known results, this implies that the emptiness and subsumption problems for
$\maxagg$-DHNs are decidable. 

For $\sumagg$-DHNs, our results offer a broad yet incomplete picture.
We prove that, for connected patterns, the expressive power on databases of bounded degree is exactly that of UNFOC, the extension of UNFO with counting quantifiers. This also yields a characterization in the non-uniform setting often studied in the GNN literature, where a different GNN may be used for each graph size \cite{DBLP:conf/lics/Grohe21}. 
The connectedness restriction is 
essential to bypass
an open problem concerning the expressive power of GNNs with global readout~\cite{Barcelo2020,DBLP:conf/aaai/HaukeW26}.
We then prove that the emptiness and subsumption problems for $\sumagg$-DHNs are undecidable, both on unrestricted-degree graphs and on bounded-degree graphs. On bounded-degree graphs, these problems become decidable when all patterns are connected. In contrast, under mild assumptions the same problems are decidable for  $\sumagg$-GNNs~\cite{DBLP:conf/icalp/BenediktLMT24}.

There is, however, much more to say about $\sumagg$-DHNs. While we do not obtain a precise logical characterization of their expressive power on databases of unrestricted degree, we are able to sandwich them between two logical formalisms that provide meaningful lower and upper bounds. For the lower bound, we define the
 new fragment  UQAFO of FO that is based on unary quantifier alternation and admits inequalities, and hence counting. We show that UQAFO is strictly more expressive than UNFOC and, at the same time, no more expressive than $\sumagg$-DHNs. This perspective also motivates a natural variation of DHNs, 
 Deep Embedding Networks (DENs), in which homomorphisms are replaced by embeddings. Using Lovász’ homomorphism counting theorem, we prove that $\sumagg$-DENs have the same expressive power as $\sumagg$-DHNs, and we also show that they are strictly more expressive than UQAFO. Making a short digression to $\maxagg$ aggregation, we further observe that $\maxagg$-DENs are strictly more expressive than $\maxagg$-DHNs. For the upper bound, we propose a homomorphism-based logic inspired by first-order logic with counting, as studied in descriptive complexity~\cite{immerman2012descriptive}. Building on recent developments~\cite{DBLP:journals/theoretics/Grohe24}, we show that this logic is at least as expressive as $\sumagg$-DHNs.

We only briefly address the case of $\meanagg$ aggregation. Our main result is that, on bounded-degree databases and for connected patterns, $\meanagg$-DHNs have the same expressive power as a variant of UNFOC in which counting quantifiers are replaced by ratio quantifiers, in the spirit of a recent study on the expressiveness of GNNs with mean aggregation~\cite{AAAIMean}.

We point out that all results stated above hold both on relational databases and on graphs, subject to mild additional assumptions for the undecidability results, namely the presence of vertex colors.  We also provide experiments that corroborate our theoretical findings on expressive power in practice. In particular, they confirm our two main separation results.
Fully detailed proofs 
are provided in the appendix.

\section{Preliminaries}
\label{sect:prelims}

\newcommand{\Sbf}{\ensuremath{\mathbf{S}}\xspace}
\newcommand{\Vbf}{\ensuremath{\mathbf{V}}\xspace}

\paragraph{Databases, Homomorphisms,  Embeddings.}
For a vector $\bar{x}\in\mathbb{R}^{k}$ we use $x_1,\ldots,x_{k}$ 
to refer to its
components. 
Fix a countably infinite set \Vbf of \emph{values}.
A {\em schema} \Sbf
is a finite set of relation symbols~$R$ with associated arity \mbox{$\mn{ar}(R) \geq 0$}. 
An {\em \Sbf-fact} is an expression of
the form $R(\bar v)$, where $R \in \Sbf$ and $\bar v$ is an
$\mn{ar}(R)$-tuple of values from $\Vbf$.
An {\em $\Sbf$-database} is a 
finite set of \Sbf-facts. We write
$\mn{adom}(D)$ to denote the \emph{active domain} of database $D$, that is, the set of values used in  $D$.
%
%
The \emph{degree} of $D$ is the
maximum number of facts that any value in $D$ occurs in.

Our DHNs run on databases,
whereas in the original paper~\cite{maehara2024deep}, DHNs are applied to undirected graphs. By choosing the \emph{graph schema} $\Sbf = \{E\}$ with $\mn{ar}(E)=2$, we can capture the setting
of directed graphs. Undirected graphs require the additional constraint that all inputs are 
irreflexive and symmetric. 
%
A \emph{pointed \Sbf-database} is a pair $(D, v)$, where $D$ is an 
\Sbf-database and $v \in \mn{adom}(D)$. We typically write
$(D,v)$ as $D^v$ or as $D^\bullet$ if the exact
identity of $v$ is not important. 
%
  For  $d \geq 1$, an \emph{$\mathbb{R}^d$-embedded \Sbf-database}
 is a pair $(D,\lambda)$ where $D$ is an \Sbf-database and $\lambda \colon \mn{adom}(D) \to \mathbb{R}^d$ is an embedding function.
Embedded databases also
exist in a pointed version, defined as expected.

A homomorphism from a pointed database $F^v$ to a pointed database $D^u$ is a
function $h \colon \mn{adom}(F)\rightarrow \mn{adom}(D)$ such that 
$h(v)=u$ and $R(\bar v) \in F$
implies $R(h(\bar v)) \in D$, where $h(\bar v)$ denotes the component-wise application of $h$.
We use $\Hom(F^\bullet,D^\bullet)$ to denote the set of all homomorphisms from $F^\bullet$ to~$D^\bullet$. A homomorphism  from $F^\bullet$ to~$D^\bullet$ is an \emph{embedding} if it is injective and, additionally, 
$R(h(\bar v)) \in D$ implies
$R(\bar v) \in F$.
Embeddings are also known as partial isomorphisms or induced substructure isomorphisms.
The following lemma is 
widely known, the graph version can be found, e.g., in \cite[Chapter~5.2.3]{Lovasz2012}. A proof for databases is in the appendix.
\begin{restatable}{theorem}{lovaszhomcount}[Lov\'asz' homomorphism count theorem]\label{lovasz homomorphism count implies embedding count}
  Let \Sbf be a schema.
  For pointed \Sbf-databases \(F^\bullet,D^\bullet\), the number $k$ of embeddings  from \(F^\bullet\) to \(D^\bullet\) is uniquely determined by the number of homomorphisms to \(D^\bullet\), from all \Sbf-databases with at most $|\mn{adom}(F)|$ values. 
  Moreover, $k$ can be written as a linear combination of the homomorphism counts. The same is true when the roles of homomorphisms and embeddings are swapped.
\end{restatable}
%

\paragraph{Deep Homomorphism Networks.} 

A finite multiset over a set $X$ is a function $M \colon X \rightarrow \mathbb{N}$ such that $M(x) > 0$ for only finitely many $x \in X$. An \emph{aggregation function} is a function from finite
multisets over $\mathbb{R}$ to $\mathbb{R}$, such as sum, arithmetic mean, and max. As a convention, we set $\mean(\emptyset) = \max(\emptyset) = 0$. 
We shall often
apply aggregation functions to vectors from $\mathbb{R}^d$.
Such applications are always component-wise.
%
\begin{definition}[Homomorphism Query]
\label{def:homquery}
Let \Sbf be a schema.
A \emph{homomorphism query} over
\Sbf with input dimension $d \geq 0$ and output dimension $d' > 0$ is a triple
   $(F^\bullet,\mu,\agg)$ where $F^\bullet$
   is a pointed \Sbf-database, $\mu$ labels each value in $F$
   with a function $\mu_v$ from $\mathbb{R}^d$ to $\mathbb{R}^{d'}$ referred to as a \emph{transformation function}, and $\agg$ is an aggregation function. The \emph{result} of evaluating 
 $(F^\bullet, \mu,\agg)$ on an $\mathbb{R}^d$-embedded pointed \Sbf-database $(D^\bullet,
\lambda)$ is the $\mathbb{R}^{d'}$-vector
\begin{align*}
    \res((F^\bullet,\mu,\agg),(D^\bullet,\lambda)) := \agg \bigllbrace \prod_{v \in \mn{adom}(F)} \mu_v(\lambda(h(v))) \mid h \in \Hom(F^\bullet, D^\bullet) \bigrrbrace
\end{align*}
where the product is applied component-wise. We say that  $(F^\bullet,\mu,\agg)$ is \emph{connected} if the Gaifman graph of $F$ is
connected (defined in the appendix).
\end{definition}
The database $F$ in the definition of homomorphism queries can
naturally be viewed as a conjunctive query.
In \cite{maehara2024deep}, the vector $\res((F^\bullet,\mu,\agg),(D^\bullet,\lambda))$ is called a \emph{generalized homomorphism
number}. This is because, for $\agg
= \sumagg$, it is identical to that
of weighted homomorphism numbers~\cite{Lovasz2012}. 

%
\begin{definition}[DHN Layer]
Let \Sbf be a schema.
An \emph{\Sbf-DHN layer with input dimension $d \geq 0$ and output dimension $d' \geq 0$} is a pair $\mathcal{L} =(\mathcal{F},
\com)$ such that 
%
 $\mathcal{F}$ is a sequence of \Sbf-homomorphism queries $(F^\bullet_1,
\mu_{1},\agg_1), \dots, (F^\bullet_m, \mu_{m},\agg_m)$
  of input dimension $d$ and
   $\com \colon \mathbb{R}^{d_1} \times \cdots \times \mathbb{R}^{d_m} \rightarrow \mathbb{R}^{d'}$ is the \emph{combination function} with $d_i$ the
     output dimension of $(F^\bullet_i,\mu_{i},\agg_i)$.
%
 Given an $\mathbb{R}^d$-embedded \Sbf-database $(D, \lambda)$,  $\mathcal{L}$ produces an 
 $\mathbb{R}^{d'}$-embedded \Sbf-database $(D, \lambda_\calL)$ defined by setting, for  each value $u \in \mn{adom}(D)$:
\begin{align*}
    \lambda_\calL(u)  = \com(\res((F_1^\bullet, \mu_{1},\agg_1),(D^u,\lambda)), \ldots, \res((F_m^\bullet, \mu_{m},\agg_m),(D^u,\lambda))).
\end{align*}
\end{definition}
%
%
%
In \cite{maehara2024deep}, a deep homomorphism network is a finite sequence of DHN layers with matching input and output dimensions that transforms an input embedding into an output embedding. Here, we want to use DHNs as classifiers and thus add a classification function.
\begin{definition}[Deep Homomorphism Network]

Let \Sbf be a schema.
A \emph{deep homomorphism network (DHN)
over \Sbf} 
 is a sequence  $\calN=(\mathcal{L}_1, \ldots ,\mathcal{L}_k,\cls)$ where each $\calL_i$
is an \Sbf-DHN layer and $\cls \colon \mathbb{R}^{d'} \rightarrow \{0,1\}$ is a \emph{classification function}, with $d' \geq 0$.
The input dimension of $\calL_1$ must be~$0$, the output dimension of $\calL_i$ must be identical to the input dimension of $\calL_{i+1}$ for $1 \leq i < k$, and the output dimension of $\calL_k$ must be $d'$. We say that $\calN$ is
an \emph{\agg-DHN}, with  $\agg \in \{ \sumagg, \meanagg,\maxagg \}$, if \agg is the only aggregation function
used in $\calN$, and $\calN$ is \emph{connected} if every homomorphism query in it is.

For an \Sbf-database $D$, the DHN
$\calN$ produces a 
sequence $(D,\lambda^0_{\calN,D}),\dots,(D,\lambda^k_{\calN,D})$ of embedded \Sbf-databases, where $\lambda^0_{\calN,D}$
assigns to every value the empty vector $()$ and for $i>0$, $\lambda^i_{\calN,D}$ is the
result of applying layer $\calL_i$ to
$(D,\lambda^{i-1}_{\calN,D})$ for $1 \leq i \leq k$.
%
%
%
The DHN $\calN$ \emph{accepts} the pointed \Sbf-database $D^v$
if $\cls(\lambda^{k}_{\calN,D}(v))=1$. It \emph{rejects} this database otherwise.
\end{definition}
%




The combination and classification functions, as well as the transformation functions in homomorphism queries may be represented as feed-forward neural networks (FNNs), defined in the appendix. The parameters of these networks are learned during the training phase of a DHN.
We  put special emphasis on FNNs that use the \emph{truncated ReLU} $\trrelu(x) = \min(\max(0,x),1)$ as the activation function.
However, all our results also apply to the case of
non-truncated  $\relu(x) = \max(0,x)$, since
$\trrelu(x) = \relu(x) - \relu(x-1)$ for all $x \in \mathbb{R}$.
In a \emph{simple} DHN, all
 combination and transformation functions must be represented by FNNs  with $\trrelu$ activation,
 without hidden layers and with rational weights, and the classification function must be a threshold function. This is made precise in the appendix.
It is easy to see and was already observed in \cite{maehara2024deep} that DHNs generalize GNNs. This even applies to GNNs with global readout
when homomorphism queries may be 
disconnected, but not otherwise.

\begin{remark}
There is an apparent mismatch between running a DHN on a database and running a GNN on the graph constructed from a database: the DHN assigns embedding vectors to the \emph{values} in the database, whereas in standard graph encodings, every vertex represents a database \emph{tuple}, and running a GNN therefore assigns embedding vectors to such tuples~\cite{DBLP:conf/icml/FeyHHLR0YYL24,
DBLP:conf/icml/ChenKL25,DBLP:conf/icde/TonshoffFGK23,
DBLP:conf/nips/WangGWZFZZCLMST24}. This mismatch, however, is easily resolved: simply increase the arity of every database relation by one and place a unique value in the fresh component of every tuple, across all relations. The embedding vectors assigned to these unique values then correspond to the embedding vectors assigned to tuples.
\end{remark}



    

\vspace*{-2mm}

\paragraph{Decision Problems.}
We consider two static analysis problems for DHNs. 
In the \emph{emptiness problem}, the question is whether a given DHN accepts
no pointed databases whatsoever. This corresponds to the classifier being unsatisfiable. In the \emph{subsumption problem}, the input consists of two DHNs $\calN_1,\calN_2$ and 
the question is whether for every pointed database $G^\bullet$, acceptance by 
$\calN_1$ implies acceptance by~$\calN_2$. In other words, the subsumption problem asks whether the vertex classifier represented by $\calN_2$ is
at least as general as the vertex classifier represented by $\calN_1$. 
It is easy to see that both problems can be reduced to one another in polynomial time.
Note that there is no obvious way to finitely represent an arbitrary DHN
in which no restrictions
are placed on the form of the
combination, transformation, and classification functions.
In static analysis, we thus assume that all of these functions are given by FNNs in which all coefficients are rational numbers. 




%

\paragraph{UNFO and (G)HML.}

The
formulas over a schema \Sbf of the 
\emph{unary negation fragment of first-order logic with counting (UNFOC)} are defined according to the following syntax rule:
\[
  \varphi(\bar x) ::= R(\bar x) \mid \varphi(\bar x) \wedge \varphi(\bar x) \mid \varphi(\bar x) \vee \varphi(\bar x) \mid \exists^{\geq n} y \, \varphi(\bar x, y) \mid 
  \neg \varphi(x_i)
\]
where $R$ ranges over the relation symbols in \Sbf, 
$n \in \mathbb{N}$, and, in the last clause, $\varphi$ has no free variables besides (possibly) $x_i$.
The fragment of UNFOC in which $n$ is restricted to $1$ is the \emph{unary negation fragment of first-order logic (UNFO)} \cite{DBLP:journals/corr/SegoufinC13}. 
Note that in UNFOC it is in general not possible to express  $\exists^{\leq n} x \, \varphi$ as $\neg \exists^{\geq n+1} x \, \neg \varphi$ 
because $\varphi$ may have more than one free variable. 
This is analogous to the asymmetry 
between existential and universal quantification in UNFO. We use UNFO(C) to classify values and therefore focus, without further notice, on formulas in one free variable.

We next introduce \emph{homomorphism modal logic (HML)}, which may be viewed as a
normal form for UNFO formulas. It was introduced already in \cite{DBLP:journals/corr/SegoufinC13}, though not under this name. While we do not use 
    modal syntax, HML can naturally be viewed as a polyadic modal logic in which the  modal diamond operator is replaced by a generalized operator based on conjunctive queries (CQs), see \cite{DBLP:journals/corr/SegoufinC13}. 
 HML formulas are
defined according to the rule
    \[
        \varphi(x) ::= \neg \varphi(x) \mid \varphi(x) \lor \varphi(x) \mid \exists \bar y \, \psi(x, \bar y)
    \]
    %
    where
    $\psi(x, \bar y)$ is 
    a conjunction
    of relational atoms  
    $R(\bar z)$ and  unary formulas~$\varphi(z)$. We emphasize
    that, up to the latter, formulas $\exists \bar y \, \psi(x, \bar y)$ are nothing but conjunctive queries and thus the expressive power of HML  is closely linked to database query languages. The name HML mentions homomorphisms because the semantics of CQs is defined naturally in terms of homomorphisms.   Note that unary atoms $P(x)$
    can be expressed by formulas
    of the form $\exists \bar y \, \psi(x, \bar y)$ with $\bar y$ empty.
     We say that a HML formula 
    is \emph{connected} if for every subformula
    $\exists \bar y \, \psi(x, \bar y)$,
    (i)~some atom in $\psi$ uses $x$ and
    (ii)~the Gaifman graph of 
    $\psi$ is connected,
     meaning the Gaifman graph of 
    the database whose facts are exactly the relational atoms that occur in $\psi$.
  %
    %
%
It was proved in \cite{DBLP:journals/corr/SegoufinC13}
that   UNFO and HML have the same expressive power.
%

We next define an extension of HML that has the same expressive power as UNFOC.
\emph{Graded homomorphism modal logic (GHML)}  
is defined in the same way as HML, except
that formulas of the form $\exists \bar y \, \psi(x, \bar y)$ are replaced by formulas 
of the form $ \exists^{\geq k_1} y_1 \cdots \exists^{\geq k_m} y_m \, \bigvee_{1 \leq i \leq \ell} \psi_i(x,y_1,\dots,y_m)$ with $\psi_i$ of the same shape as the formulas $\psi$ in HML. A GHML formula is \emph{connected}
if in every existentially
quantified subformula, all formulas $\psi_i$ satisfy Conditions~(i) and~(ii).
%
%
\begin{restatable}{lemma}{unfocghmlequiv}
\label{lem:twocombined}
  UNFOC and GHML have the same expressive power. 
\end{restatable}

\section{Logical Characterizations of DHNs}
\label{sect:DHNs}

We relate DHNs
with different aggregation functions to variants of UNFO and HML. The following theorem summarizes our results on the expressive power of 
$\maxagg$-DHNs and
$\sumagg$-DHNs.
\begin{theorem}
\label{thm:combi}~\\[-4mm]
\begin{enumerate}
    \item    $\maxagg$-DHNs, simple $\maxagg$-DHNs, UNFO, and HML all have the same expressive power.

     \item Simple 
     $\sumagg$-DHNs are at least as expressive as %
     GHML and UNFOC.
     
  \item  
    On databases of bounded degree, the following have the same expressive power: connected $\sumagg$-DHNs, 
     simple connected $\sumagg$-DHNs, and connected GHML.
     
\end{enumerate}
\end{theorem}
%
%
The connectedness assumption in Point~3 of 
the theorem
is related to an open problem concerning the expressive power of GNNs. Namely, it can  be proved
along the lines of
 \cite{Barcelo2020} that 
 on bounded degree graphs,
 the expressive power of 
$\sumagg$-GNNs without global readout is exactly that of  graded modal logic. 
With global readout, by contrast, there is no known logic that exactly matches the expressive power of $\sumagg$-GNNs on bounded degree graphs; see Hauke and Walega~\cite{DBLP:conf/aaai/HaukeW26} for the case without a bounded degree assumption. To avoid this open problem, we impose connectedness.
If the constant degree bound in Point~3 is strengthened to a constant bound on the size of the database, which amounts to considering non-uniform expressive power, then the statement remains true even without the  `connected' qualifier, our proofs still apply. UNFOC is not explicitly mentioned in Point~3 because it is awkward to define a connectedness condition for it. However, the resulting logic would have exactly the same expressive power as connected GHML.

Point~3 ceases to hold when we
drop the bounded degree assumption. 
In a sense, this is well-known: 
even GNNs with $\sumagg$ aggregation can express properties that do not fall within FO such as ``a graph vertex has more red than blue neighbors,'' and thus they also do not fall within UNFOC. What is more surprising is that even when restricted to properties expressible in FO, the expressive power of connected $\sumagg$-DHNs exceeds that of UNFOC.
\begin{restatable}{proposition}{sumdhnnotinunfocuniform}
    \label{thm:sum-DHN-not-in-UNFOC-uniform}
    There is a connected $\sumagg$-DHN that expresses an FO property and is not equivalent to a UNFOC formula.
\end{restatable}
A concrete such property is \emph{local transitivity} of vertices $v$ in a directed graph $G$, meaning that
for all vertices $u_1$, $u_2$:
$(v,u_1),(u_1,u_2) \in E(G)$
implies $(v,u_2) \in E(G)$.

By virtue of Theorem~\ref{thm:combi}, we may  use logics in place of DHNs to analyze static analysis problems for DHNs. Especially in the 
case of undecidability results and lower complexity bounds, this is much more convenient than working directly with DHNs.
The following theorem summarizes our results on the static analysis of 
$\maxagg$-DHNs and
$\sumagg$-DHNs. 
\begin{theorem}
\label{thm:staticanalysiscombi}
\phantom{}
\begin{enumerate}
\item  The emptiness problem and the subsumption problem for $\maxagg$-DHNs are decidable.

\item  The emptiness problem and the subsumption problem for  connected $\sumagg$-DHNs on databases of bounded degree are decidable and {\sc coNExpTime}-complete.

\item   The emptiness problem and the subsumption problem for   $\sumagg$-DHNs are undecidable.
This already holds (a)~for connected simple $\sumagg$-DHNs on graphs and (b)~for
 (not necessarily connected)  simple $\sumagg$-DHNs on graphs of bounded degree.

\end{enumerate}
\end{theorem}
The precise computational complexity of the problems in Point~1 of
Theorem~\ref{thm:staticanalysiscombi}
remains open. We note, however,
that the translation from UNFO
and HML to $\maxagg$-DHNs is polynomial (while the converse translation is not), and thus the  {\sc 2ExpTime}-hardness of (finite) satisfiability in UNFO transfers to the emptiness and subsumption problems of $\maxagg$-DHNs. The reader might want to compare Point~2 
of Theorem~\ref{thm:staticanalysiscombi} with the result from \cite{DBLP:conf/icalp/BenediktLMT24}
 that for simple $\sumagg$-GNNs, defined in exact analogy to our simple $\sumagg$-DHNs, emptiness and subsumption are decidable.

Turning towards mean aggregation, 
    we define an extension RHML of HML with ratio quantifiers, in analogy with the extension of modal logic by such quantifiers considered in
    \cite{AAAIMean}. Unlike the 
    logics considered up to this point,
    RHML is \emph{not} a fragment of first-order logic, and neither of monadic second-order logic. We refrain
     from considering an analogous extension of UNFO because it is not clear how it can be defined.
    %
%
   %
    %
%
    An RHML formula is defined according to the rule
    \[\varphi(x) \coloncolonequals \neg\varphi(x)\mid \varphi(x)\lor \varphi(x)\mid \exists _{\geq t}\, \bar{y} \, (\mu(x,\bar{y}), \nu(x,\bar{y}))\mid \exists_{>t}\, \bar{y}\, (\mu(x,\bar{y}),\nu(x,\bar{y})),\]
    where \(t\in \mathbb{Q}\), \(\mu(x,\bar{y})\) is a conjunction of unary formulas of the form \(\varphi(z)\), and \(\nu(x,\bar{y})\) is a conjunction of relational atoms \(R(\bar{z})\).
    In particular, unary atoms of the form \(P(z)\) may occur in both \(\mu\) and \(\nu\).
    %
  The semantics of the RHML formula
  $
  \xi = \exists_{\geq t}\, \bar{y}\,(\mu(x,\bar{y}),\nu(x,\bar{y}))
  $, with  $|\bar y|=n$,
  is as follows: 
  \[
  D \models \xi(v)
  \text{ if and only if }
  \frac{|\{\overline{w}\in \mn{adom}(D)^n\mid D\models \mu(v,\overline{w})\land \nu(v,\overline{w})\}|}{|\{\overline{w}\in \mn{adom}(D)^n\mid D\models\nu(v,\overline{w})\}|}\geq t.
  \]
As a special case, we  put $D \models \xi(v)$ if the denominator of the above fraction is \(0\).
The semantics of the \(\exists_{>t}\) quantifier is defined in the same way, except that
such formulas are not satisfied when the denominator is \(0\). As an example, we can define in RHML over the graph schema the
set of all vertices $v$ such that the number of triangles rooted at $v$ whose vertices all have a reflexive loop is at least 50\% of the total number of triangles rooted at $v$:
\[
\exists_{\geq 0.5} \, y z \, \big (
(E(x,x) \wedge E(y,y) \wedge E(z,z)),
(E(x,y) \wedge E(y,z) \wedge E(z,x))
\big ).
\]

%
\begin{theorem}\label{thm:meancombinew}~\\[-4mm]
\begin{enumerate}
    \item \meanagg-DHNs are at least as expressive as RHML.

    \item   On databases of bounded degree, the following have the same expressive power: connected \meanagg-DHNs, connected simple \meanagg-DHNs, and connected RHML.
\end{enumerate}
\end{theorem}
Point~1 does not mention simple 
\meanagg-DHNs (only) because we 
use a non-continuous  activation function (Heavyside), cf.\ \cite{AAAIMean}.









\section{Deep Embedding Networks and  Unary Quantifier Alternation}
\label{sect:DENs}


Section~\ref{sect:DHNs} leaves open the question of a logic that  characterizes the expressive power of $\sumagg$-DHNs on databases of
unrestricted degree. 
In this section, we 
establish a lower bound.
It
is given by a novel and significant 
strengthening UQAFO of UNFOC 
that replaces unary negation with unary quantifier alternation. 
In the presentation as a modal logic,
the transition from UNFO to UQAFO
corresponds to the replacement
of homomorphisms by embeddings.
This 
leads us to also introduce \emph{deep embedding networks (DENs)}, 
defined exactly like DHNs, but with homomorphisms replaced by embeddings; see \cite{Luo2026} for a related proposal.
From now on, we assume that it is understood what we mean by an  \emph{embedding query}, a \emph{DEN layer}, and a \emph{DEN}. It is an easy
consequence of Theorem~\ref{lovasz homomorphism count implies embedding count} that $\sumagg$-DHNs and $\sumagg$-DENs have the same expressive power. 
Interestingly, although motivated by $\sumagg$ aggregation, UQAFO turns out to have exactly the same expressive power as $\maxagg$-DENs.

The formulas over schema \Sbf of the \emph{unary quantifier alternation fragment of
first-order logic (UQAFO)} 
are the formulas $\varphi(\bar x)$ formed according to the following syntax rule:
\[
\begin{array}{r@{\;}c@{\;}l}
\varphi(\bar{x}) &::=& \varphi_\exists(\bar{x})\mid \varphi_\forall(\bar{x}) \mid \varphi(\bar{x}) \circ \varphi(\bar{x}) 
\\[1mm]
\varphi_\exists(\bar{x}) &::=& R(\bar{x})\mid \neg R(\bar{x})\mid \varphi_\exists(\bar{x})\circ \varphi_\exists(\bar{x})
\mid x_i=x_j\mid x_i\neq x_j\mid \exists\bar{y}\,\varphi_\exists(\bar{x}, \bar{y})\mid \varphi_\forall(x_i) \\[1mm]
    \varphi_\forall(\bar{x}) &::=& R(\bar{x})\mid \neg R(\bar{x})\mid \varphi_\forall(\bar{x})\circ \varphi_\forall(\bar{x})\mid 
    x_i=x_j\mid x_i\neq x_j\mid \forall\bar{y}\,\varphi_\forall(\bar{x}, \bar{y})\mid \varphi_\exists(x_i)
  \end{array}
  \]
  where ``$\circ$'' stands for either the ``$\vee$'' or the ``$\wedge$'' symbols. 
   It is easy to see that
   quantifier alternation is only possible for unary formulas: if $\exists \bar x \, \psi$ is a UQAFO formula and $\forall \bar y \, \vartheta$ is a subformula of
   $\psi$ that is not nested within another
   quantifier in $\psi$, then $\forall \bar y \, \vartheta$ has at most one free variable, and likewise with $\exists$ and
   $\forall$ swapped.
We also remark
  that  because UQAFO admits inequalities, it does not become more expressive when  counting quantifiers are added. 
  For example,
  the formula  \(\exists x \exists y\, (\forall z\, R(x,z) \land \forall z\,S(y,z))\) is a UQAFO formula while the formula
  \(\exists x \exists y \forall z \, (E(x,z) \wedge E(z,y)) \)
  is not. Like for UNFO, we concentrate on UQAFO formulas in one free variable.
%


We next observe that UQAFO is strictly more expressive than UNFOC. The separating property is local transitivity.
\begin{restatable}{lemma}{lemUQAFOUNFOC}
\label{lem:UQAFOUNFOC}
 UQAFO is strictly more expressive than UNFOC.
\end{restatable}

We next introduce a modal logic-style presentation of UQAFO, in the same spirit
in which HML is an alternative presentation of UNFO. We define \emph{embedding modal logic (EML)} exactly like HML except that
in formulas of the form $\exists \bar y \, \psi(x, \bar y)$, the
subformula
$\psi(x, \bar y)$ may now also contain negated relational atoms
$\neg R(\bar z)$ and inequalities $z_1 \neq z_2$. Just like UQAFO, EML does not become
more expressive when quantifiers $\exists \bar y \, \psi(x, \bar y)$ are replaced
with a counting version. More details are in the appendix.
We remark that EML is closely related to a GNN-related modal logic GML$(\mathcal{T})$ recently introduced in \cite{Luo2026}. 

\noindent
\Needspace*{10\baselineskip}
\begin{theorem}
\label{thm:DENcombi}
~\\[-4mm]
    \begin{enumerate}
        \item UQAFO, EML, and $\maxagg$-DENs have the same expressive power.
        \item $\sumagg$-DENs are strictly more expressive than UQAFO.
        \item $\sumagg$-DENs and $\sumagg$-DHNs have the same expressive power.
        \item $\maxagg$-DENs are more expressive than $\maxagg$-DHNs.
    \end{enumerate}
\end{theorem}
%
The separating property for Point~2  concerns directed graphs with three unary relation
symbols, $P$, $Q$, and $R$. We are interested in all vertices $v$ in a  graph $G$ such that \(v\) and its descendants form a linear order 
$v=v_0,\dots,v_n$ in which each vertex satisfies exactly one of \(P, Q,\) and \(R\); we require that \(n \mod 3= 2\) and that $v_i$ satisfies $P$
if $i \mod 3 = 0$, $Q$ if $i \mod 3 = 1$, and $R$ if $i \mod 3 =2$. We call this property
\emph{PQR-order}.
A separating property for Point~4 is local transitivity.

Barcel\'o et al.\ studied 
GNNs in which the initial vertex features include homomorphism counts for a finite set of pointed pattern graphs
\cite{barcelo2021graph}. While it is obvious that DHNs subsume such GNNs  
\cite{maehara2024deep}, a strict separation had been left open.
We use Points~1 to~3 of Theorem~\ref{thm:DENcombi} to show that, already on undirected unembedded graphs, $\sumagg$-DHNs (and even $\maxagg$-DENs)  can express properties that are not expressible by a $\sumagg$-GNNs with initial homomorphism counts. 
As a separating example, we use the property that a vertex lies on a $6$-cycle on which every vertex has a degree $1$ neighbor. This \emph{sun property} can be expressed as
an EML formula:
\begin{align*}
\varphi_{\sun}(x_1) =
\exists x_2\cdots x_6 \,\Big( \bigwedge_{1 \leq i \leq 6}\!(E(x_i,x_{(i\bmod 6)+1}) \wedge \exists y\,(E(x_i,y) \wedge \exists^{=1} z\,E(y,z)) 
\land\!\!\!\bigwedge_{i < j\leq 6}\!\!\! x_i\neq x_j)
\Big)
\end{align*}
where $\exists^{=1} z\, E(y,z)$ is shorthand for $\exists^{\geq1}z\, E(y,z) \wedge \neg \exists^{\geq2}z\, E(y,z)$. By Points~1 to~3 of Theorem~\ref{thm:DENcombi}, the sun property can thus be expressed by a 
$\sumagg$-DHN. We prove that 
$\varphi_{\sun}$ is not expressed by GNNs with initial homomorphism counts, even when they are 
equipped with global readout,
by
reusing a result by Chen et al.\ which states that homomorphism counts cannot detect isolated vertices \cite{chen2025algorithms}.
\begin{theorem}\label{thm:sun-expressivity}
    On undirected graphs, $\sumagg$-DHNs are strictly more expressive than $\sumagg$-GNNs with global readout and features enriched with homomorphism counts. The
    same holds on directed graphs
    and on graphs with constant embedding $(0)$.
\end{theorem}

\section{An Upper Bound via More Expressive Counting}
\label{sect:HMLC}

To identify an upper bound
on the expressive power of $\sumagg$-DHNs, we propose a  logic HML+C that is in the spirit of first-order logic with counting as studied in descriptive complexity \cite{immerman2012descriptive}.
Our development closely follows  that of Grohe  \cite{DBLP:journals/theoretics/Grohe24} who introduces a certain version of first-order logic with counting called FO+C and then adapts his development to the guarded fragment GFO+C. HML+C
is an adaptation in the same style, based on HML.

HML+C is a two-sorted logic,
with an individual sort that ranges over the active domain of the database and a number sort, ranging over $\mathbb{N}$. We denote number variables with $n,m$ and individual variables with $x,y$. 
The \emph{(numerical) terms}
and \emph{(individual) formulas} of HML+C are defined by the following mutually recursive syntax rules:
\[
\begin{array}{rcl}
  t(x,\bar n) & ::= & c \mid n 
  \mid t + t \mid t \cdot t \mid \#(\bar y, m_1 < t_1,\dots,m_k < t_k) . \psi(x,\bar y,\bar n, m_1,\dots,m_k) \\[1mm]
  \varphi(x,\bar n) & ::= & \neg \varphi(x,\bar n) \mid
  \varphi(x,\bar n) \vee \varphi(x,\bar n) 
  \mid t(x,\bar n) \leq t'(x,\bar n)
\end{array}
\]
where $c$ ranges over $\mathbb{N}$,
each \emph{bounding term} $t_i$  may use $x$, all variables from $\bar y$, as well as the variables $m_1,\dots,m_{i-1}$ and all variables in $\bar n$, and
 $\psi$ is a conjunction
    of relational atoms  
    $R(\bar z)$ where $\bar z$ contains only $x$ and variables from $\bar y$, as well as   formulas with one free individual variable~$\chi(z,\bar n,m_1,\dots,m_k)$ where $z$ is $x$ or a variable from $\bar y$. We call such a formula $\psi$ a  \emph{pattern formula}. Note that HML+C may still be viewed as a modal logic in the sense that all formulas and terms have at most one free individual variable. On the other hand,
    the numerical part of HML+C
    is exactly that of FO+C: every FO+C formula that does not use individual variables is also an HML+C formula.

The semantics of HML+C is just the standard semantics of multi-sorted FO, with the exception of the counting terms. 
Let \(t(x,\bar n)\) be an  HML\(+\)C term, $D$ a database, 
\(u\in\mn{adom}(D)\), and \(\bar c\in\mathbb N^{|\bar n|}\). We
write
\(
\llbracket t\rrbracket^D(u,\bar c)
\)
to denote the value of the term \(t\) in \(D\) under the assignment that sends \(x\) to
\(u\) and \(\bar n\) to~\(\bar c\). The semantics of  counting terms is then defined by setting
$
\llbracket \#(\bar y, m_1 < t_1,\ldots,m_k < t_k) . \psi \rrbracket^{(D,\alpha)}
$
to be the number of tuples
$\bar a \bar b$ with
$\bar a \in  \mn{adom}(D)^{|\bar y|}$ and $\bar b = (b_1,\dots,b_k) \in \mathbb{N}^k
$
such that
\[
b_i < \llbracket t_i \rrbracket^{(D,\alpha[\bar y/\bar a][m_1/b_1]\cdots[m_{i-1}/b_{i-1}])}
\text{ for } 1 \leq i \leq k
\quad
\text{ and }
\quad 
\llbracket \psi \rrbracket^{(D,\alpha[\bar y/\bar a][\bar m/\bar b])} = 1
\]

where $\bar m=(m_1,\dots,m_k)$.
The following example HML+C formula 
expresses local transitivity:
\[
\#(y,z).(E(x,y)\wedge E(y,z))
\;\le\;
\#(y,z).(E(x,y)\wedge E(y,z)\wedge E(x,z)).
\]

Since our technical development mirrors the one in \cite{DBLP:journals/theoretics/Grohe24}, we likewise assume that all activation functions are Lipschitz continuous. A dyadic rational is a   rational number of
the form $\frac{z}{2^s}$ with $z \in \mathbb Z$ and $s \in \mathbb N$. Equivalently,  a rational  is dyadic if its binary expansion is finite.
We say that an FNN \emph{uses dyadic rational piecewise-linear activation} if the activation function is  piecewise-linear  and all slopes, intercepts, and thresholds are dyadic rationals. Then a  DHN \emph{uses dyadic rational piecewise-linear activation} if every transformation, combination, and classification function in it is represented by an FNN that does. Moreover, a DHN \emph{uses dyadic rational coefficients} if every
 FNN in it does.

Our main result in this section is the following.
\begin{theorem}
\label{thm:DHNstoHMLplucS}
   Every \sumagg-DHN that uses  dyadic rational piecewise-linear activation and dyadic rational coefficients is equivalent to an HML+C formula. 
\end{theorem}
By Theorem~\ref{thm:DENcombi}, the same is also true for $\sumagg$-DENs and $\maxagg$-DENs. Also recall that, by Theorem~\ref{thm:combi},
$\maxagg$-DHNs
can even be translated into the fragment HML of HML+C. We believe that the expressive power of HML+C is strictly greater than that of $\sumagg$-DHNs. It was observed in 
\cite{DBLP:journals/theoretics/Grohe24} that the logic GFO+C studied there can define the vertex property of having even degree, and that GNNs cannot do that. It is easy to see
that also HML+C can express this property while there is no obvious way to express it using a DHN in which all functions are represented as FNNs, see Remark~7.7 in \cite{DBLP:journals/theoretics/Grohe24}. 

\section{Experiments}


To experimentally evaluate our theoretical findings, we concentrate on the fact that $\sumagg$-DHNs have higher expressive power than $\maxagg$-DHNs (Theorems~\ref{thm:combi} and~\ref{thm:DENcombi} and Lemma~\ref{lem:UQAFOUNFOC}), and also than GNNs with initial homomorphism counts (Theorem~\ref{thm:sun-expressivity}).
The relevant properties are 
local transitivity and the sun property.
%
Our experiments are based on synthetic datasets.
%
For local transitivity, we construct a dataset with a balanced number of positive and negative examples by starting with the disjoint union of sufficiently many linear-order graphs, and then randomly deleting edges. This results in 1921 out of 4000 vertices to be locally transitive.
For the sun property, we start with a number of 6-cycles and then randomly attach fresh vertices and other 6-cycles in a way that results in 600 out of 1200 vertices  satisfying the property.

For local transitivity, we equip DHNs with 13 homomorphism queries using at most three variables, in order to reduce the manual pattern-selection bias. For the sun property, we use a 6-cycle query and a single-edge query, as there are no other obvious candidate patterns for our dataset. The symmetric nature of the 6-cycle results in an excessive number of homomorphisms, leading to performance issues. We handle these by restricting to injective homomorphisms in the sun property case.
We represent DHN transformation and combination functions
by FNNs with 1 and 3 hidden layers respectively, and with leaky ReLU activation; the leakiness 
 serves to alleviate vanishing gradients.

We compare the performance of DHNs to  the standard GCN~\cite{DBLP:conf/iclr/KipfW17}, GraphSAGE~\cite{DBLP:conf/nips/HamiltonYL17}, and GIN~\cite{xu2019powerful}  implementations provided by PyTorch Geometric. For all models, we use an embedding dimension of 32, three message passing layers,  the Adam optimizer, and  single batch training. 
 DHNs are trained with a learning rate of 0.0003 for 1000 epochs, and
 the other models  with a learning rate of 0.001 for 10000 epochs. We split the datasets using 6:2:2 ratios in training, validation and test set, and report results on the test set. All experiments were performed on a CPU of a MacBook Pro with M1 Pro chip.
 More details on the experimental setup can be found in Appendix~\ref{app:experiments}.
Code and data will be made publicly available.

\begin{table}
    \begin{minipage}{.5\linewidth}
      \caption{Local Transitivity Experiments}
      \label{tab:exp-results-lt}
      \centering
        \begin{tabular}{lrr}\toprule
        Architecture & F1 & AUROC \\\midrule
            \sumagg-DHN & 0.98 & 1.00  \\
            \maxagg-DHN & 0.83 & 0.93 \\
            GCN & 0.76 & 0.86 \\
            GraphSAGE & 0.71 & 0.90 \\
            GIN & 0.94 & 0.98 \\
            \bottomrule
        \end{tabular}
    \end{minipage}%
    \begin{minipage}{.5\linewidth}
      \centering
              \caption{Sun Property Experiments}
                    \label{tab:exp-results-sp}
        \begin{tabular}{lrr}\toprule
        Architecture & F1 & AUROC \\\midrule
            \sumagg-DHN & 1.00 & 1.00 \\
            GIN with hom. counts & 0.88 & 0.94 \\
            GCN & 0.81 & 0.88 \\
            GraphSAGE & 0.67 & 0.50 \\
            GIN & 0.68 & 0.73 \\
            \bottomrule
        \end{tabular}
    \end{minipage} 
    \vspace{-0.3cm}
\end{table}


Results are reported in Tables~\ref{tab:exp-results-lt}
and~\ref{tab:exp-results-sp}.
Both experiments confirm the theoretical findings for which they were designed.  Moreover, DHNs 
perform consistently better than
all other models. Only the 
performance of GINs on the local transitivity dataset is remarkable.
We believe that it is due to the
relative simplicity of the data,
on which counting and comparing the number of successors and counting the number of descendants at distance two should 
enable a good approximation of local transitivity.




\section{Conclusion and Limitations}




DHNs provide a principled neural architecture for relational data, as homomorphism-based message passing connects directly to conjunctive queries. 
We have developed a logical account of the expressive power of DHNs over relational databases. Our results fully characterize \maxagg aggregation, while exact characterizations for \sumagg aggregation and \meanagg aggregation over databases of unrestricted degree remain open. These logical characterizations offer a stepping stone toward linking DHNs to SQL-style query languages, cf.~\cite{DBLP:journals/pvldb/BaranyCO12}. 
Further directions 
include tighter integration of DHNs with database management systems, both for optimized training and inference and for end-to-end training of neuro-relational pipelines combining DHNs with relational query processing. Finally, making DHNs usable out of the box will require principled pattern-selection methods. For this, results by Barcel{\'o} et al.~\cite{barcelo2021graph} can provide a starting point.

{\bf Limitations.} As discussed above, several theoretical questions remain open, as does a more systematic account of the relationship between DHNs and relational databases. Our implementation should be viewed as a proof of concept rather than a highly optimized system. Although some static analysis problems for DHNs turn out to be  decidable, their worst-case complexity is high, and developing practically efficient algorithms for them is a non-trivial challenge.

\bibliographystyle{plain}
\bibliography{DHN}

\newpage 
\appendix

\section{Additional Material for Section~\ref{sect:prelims}}

\subsection{Missing Definitions}

We first define some notions introduced in Section~\ref{sect:prelims}.
\begin{definition}
    The \emph{Gaifman graph} of an \(\Sbf\)-database \(D\) is the undirected graph \((V,E)\), where 
    \begin{itemize}
        \item the set of vertices \(V\) is the active domain of \(D\), and
        \item there is an edge between vertices \(v_1\) and \(v_2\) iff there is a fact \(R(u_1,\ldots,u_n)\) in \(D\) such that \(v_1,v_2\in \{u_1,\ldots,u_n\}\).
    \end{itemize}
\end{definition}

\medskip
\begin{definition}
The \emph{distance} between two 
values $v_1,v_2 \in \mn{adom}(D)$
in a database $D$
is the length of the shortest  path between $v_1$ and $v_2$ in the Gaifman graph of $D$, or $\infty$ if no such path exists. The \emph{diameter} of a pointed database $D^v$
is the smallest $n$ such that 
for every value $u \in \mn{adom}(D)$, 
the distance between $v$ and $u$ is at most $n$.
\end{definition}
\medskip
\begin{definition}
    A \emph{feed-forward neural network (FNN)} with \(n\) hidden layers and input dimension \(d_0\), hidden dimensions \(d_1,\ldots,d_n\) and output dimension \(d_{n+1}\) consists of 
    \begin{itemize}
        \item weight matrices \(A_i \in \mathbb{R}^{d_i\times d_{i+1}}\) for each \(1\leq i\leq n+1\),
        \item biases \(\bar b_i \in \mathbb{R}^{d_i}\) for each \(1\leq i\leq n+1\), and
        \item activation functions \(\sigma_i: \mathbb{R}\to\mathbb{R}\) for each \(1\leq i\leq n+1\).
    \end{itemize}
    Given an input vector \(\bar v_0\in \mathbb{R}^{d_0}\), the FNN computes the following vectors:
    \[\bar v_{i+1} = \sigma_{i+1}(v_{i} \cdot A_{i+1} +\bar b_{i+1}),\]
    where the activation function is applied component-wise.
    It outputs the vector \(\bar v_{n+1}\).
    Common choices for the activation function include \(\relu(x) =\max(0,x)\), \(\trrelu(x) =\min(1, \max(0,x))\) and the sigmoid function \(f(x)=\frac{1}{1+e^{-x}}\).
\end{definition}

\medskip
We make precise what we mean by a \emph{simple} DHN.
They are defined analogously to simple GNNs, which were introduced in \cite{Barcelo2020}.

\begin{definition}
    
 A \emph{simple} DHN is a DHN, where\footnote{All dimensions are chosen to fit the definition of DHNs.}
 \begin{itemize}

  \item each combination function must be of the form
\[\com(\bar{x}_1,\dots,\bar{x}_m) = \trrelu(\bar{x}_1\cdot A_1 + \cdots + \bar x_m \cdot A_m +\bar{b}),\]
where $\trrelu$ is applied component-wise, each
$A_i \in \mathbb{Q}^{d_i\times d'}$ is a matrix, and  $\bar{b} \in \mathbb{Q}^{d'}$ is a bias vector;

\item each transformation function must be of the form 
$t(x)=\trrelu(\bar x \cdot A + \bar b)$ where $\trrelu$ is applied component-wise, $A \in \mathbb{Q}^{d \times d'}$  is a matrix, and $\bar b \in \mathbb{Q}^{d'}$ is a bias vector;

\item for some
$c \in \mathbb{Q}$, $i \in \{1,\dots,d'\}$, and ${\sim} \in \{ {>}, {\geq} \}$, the classification function must be of the form 
    \[\cls(x_1,\dots,x_{d'}) = \begin{cases}
        1 & \text{ if } x_i \sim c\\
        0 & \text{otherwise}.
    \end{cases}\]
 \end{itemize}
 
\end{definition}
\medskip

\paragraph{Expressive Power of Value Classifiers.} Here we make precise what we mean by two formalisms having the same expressive power.
Both DHNs and FO formulas with one free variable may be viewed
as classifiers for values in databases. 
 A \emph{value property}  is a set $P$ of pointed 
databases. Two value classifiers are \emph{equivalent} if they express the same value property. We say that $P$ is \emph{expressible} by a class
of value classifiers $\calC$, such as $\calC = \text{DHN}$,
if there is a $C \in \calC$ such that the
pointed databases accepted by $C$ are exactly
those in~$P$.
For classes of value classifiers $\calC_1,\calC_2$, we
 say $\calC_2$ is \emph{at least as expressive} as $\calC_1$ to mean that every
property expressible by a classifier from $\calC_1$ 
is also expressible by a classifier from $\calC_2$.
We say \(\calC_1\) and \(\calC_2\) have the \emph{same expressive power} if \(\calC_1\) is at least as expressive as \(\calC_2\) and vice versa.

\subsection{Proof of Theorem~\ref{lovasz homomorphism count implies embedding count}}

We now turn our attention to the relation between homomorphism numbers and embedding numbers.
\lovaszhomcount*

 Before we prove this theorem, we introduce a central notion used in the proof.

\begin{definition}
  Let \(F^v\) be an \Sbf-database, \(A = \mn{adom}(F)\) and let \(\mathcal{P}\) be a partition of \(A\), that is a set of pairwise disjoint subsets of \(A\) such that \(\bigcup_{P\in \mathcal{P}} P = A\).

  Then \(F^v/\mathcal{P}\) is the \Sbf-database with \(\mn{adom}(F^v/\mathcal{P}) = \mathcal{P}\) such  that:
  \begin{itemize}
    \item \(R(X_1,\ldots,X_n) \in F^v/\mathcal{P}\) if and only if there exist \(x_1\in X_1,\ldots,x_n\in X_n\) such that \({R(x_1,\ldots,x_n) \in F^v}\),
    \item The distinguished value is the set \(P\in \mathcal{P}\) such that \(v\in P\).
  \end{itemize}
\end{definition}
\begin{proof}
  We use \(\Emb(F^\bullet, D^\bullet)\) and \(\Inj(F^\bullet, D^\bullet)\) to denote the set of embeddings and injective homomorphisms from \(F^\bullet\) to \(D^\bullet\).
  We start with the second point, that is, the number of homomorphisms from \(F^\bullet\) to \(D^\bullet\) is uniquely determined by the number of embeddings to \(D^\bullet\) from all \Sbf-databases with at most \(|\mn{adom}(F)|\) values.
  
  It is easy to see that each homomorphism \(h:\mn{adom}(F)\to \mn{adom}(D)\) can be written as
  \begin{itemize}
    \item the partition \(\mathcal{P}\) of \(\mn{adom}(F)\) induced by the equivalence \(v\simeq u \text{ iff } h(v) = h(u)\) and the homomorphism \({f(v) = \{u\mid h(u) = h(v)\}}\) from \(F^\bullet\) to \(F^\bullet/\mathcal{P}\),
    \item followed by an injective homomorphism \(g\) from \(F^\bullet/\mathcal{P}\) to \(D^\bullet\).
  \end{itemize}
  Thus, \[|\Hom(F^\bullet,D^\bullet)| = \sum_{\mathcal{P}} |\Inj(F^\bullet/\mathcal{P},D^\bullet)|,\]
  where the sum ranges over all partitions of \(\mn{adom}(F)\).

  Let \(F^v\) and \(G^v\) be two \Sbf-databases.
  We write \(G^v\succeq F^v\) if \(\mn{adom}(F) = \mn{adom}(G)\) and \({R(\bar{x}) \in F^v\text{ implies } R(\bar{x})\in G^v}\).

  Each injective homomorphism which maps from \(F^v\) to \(D^u\) can be written as
  \begin{itemize}
    \item a database \(G^v\succeq F^v\)
    \item and an embedding from \({G}^v\) to \(D^u\).
  \end{itemize}
  Therefore, \[|\Inj(F^\bullet, D^\bullet)| = \sum_{{G}^\bullet\succeq F^\bullet} |\Emb({G}^\bullet, D^\bullet)|,\]
  where the sum ranges over non-isomorphic databases that satisfy \(G^\bullet\succeq F^\bullet\).
  We can thus write the homomorphism count as
  \begin{equation}\label{eq:lovasz-emb-to-hom}
    |\Hom(F^\bullet, D^\bullet)| = \sum_{\mathcal{P}}\sum_{{G^\bullet}\succeq F^\bullet/\mathcal{P}} |\Emb({G}^\bullet, D^\bullet)|.
  \end{equation}
  Let \(\mathcal{F} = \{G^\bullet\mid \exists \mathcal{P}.\,G^\bullet\succeq F^\bullet/\mathcal{P}\}\) and let \(\leq\) be an order of \(\mathcal{F}\) such that \(|\mn{adom}(F_1)| < |\mn{adom}(F_2)|\text{ implies } F_1\leq F_2\) and \(F_1\succeq F_2\) implies \(F_1\leq F_2\).
  Let \(\Hom(\mathcal{F}, D^\bullet)\) be the vector of homomorphism counts, where the entries are sorted according to \(\leq\). Define \(\Emb(\mathcal{F}, D^\bullet)\) analogously.
  By Equation \ref{eq:lovasz-emb-to-hom}, there then exists a matrix \(M\) such that for all \(D^\bullet\), \[\Hom(\mathcal{F}, D^\bullet) = M \Emb(\mathcal{F}, D^\bullet).\]
  Moreover, \(M\) is a triangular matrix with all diagonal entries being \(1\).
  Thus, there exists an inverse \(M^{-1}\) and
  \[\Emb(\mathcal{F}, D^\bullet) = M^{-1} \Hom(\mathcal{F}, D^\bullet).\]
  This shows that the number of embeddings of \(F^\bullet\) to \(D^\bullet\) is uniquely determined by the number of homomorphisms to \(D^\bullet\) from all \Sbf-databases with at most \(|\mn{adom}(F)|\) values.
\end{proof}

\subsection{Proof of Lemma~\ref{lem:twocombined}}

We show that UNFOC and GHML have the same  expressive power.
This parallels the result from \cite{DBLP:journals/corr/SegoufinC13} that the corresponding logics without counting, UNFO and HML, can express the same properties.

Before we embark on the proof of Lemma~\ref{lem:twocombined}, we make
a remark on the definition of UNFOC. We have deliberately defined UNFO, and consequently also UNFOC, to not admit equality atoms $x=y$. This is a slight deviation from the definition of UNFO in \cite{DBLP:journals/corr/SegoufinC13}
where such atoms are included. It is easy to see that UNFO formulas in one free variable have the same expressive power with and without these atoms. The same, however, is not clear in the case of UNFOC.
In fact, there is no obvious way
to find a translation to GHML
for  UNFOC formulas in one free variable that may contain equality atoms. To get an idea of the challenges, the reader might consider the formula 
$\exists^{\geq 2} y\exists^{\geq 2} z \, ((x=z \wedge R(x,y,z)) \vee S(x,y,z))$.

%

%
%
Towards a proof of Lemma~\ref{lem:twocombined},
we first describe a transformation that allows us to pull out quantifiers.
\begin{lemma}\label{counting_quantifier_scope_change}
Let \(\varphi,\psi\) be two UNFOC formulas which have no free variable in common and let \(x\) be a free variable in \(\psi\).
Let \(c\) be the maximum counting bound in both formulas and let \(n\in\mathbb{N}\).
Then for all databases \(D\) with at least \(\max(n,c)\) values the following transformations pull out existential quantifiers:
\begin{align*}
    D\models (\varphi\circ \exists^{\geq n}x\, \psi)(d_1,\ldots,d_k) &\iff D\models(\exists^{\geq n}x\, (\varphi\circ \psi))(d_1,\ldots,d_k)\\
    D\models (\exists^{\geq n}x\, \psi\circ \varphi)(d_1,\ldots,d_k) &\iff D\models(\exists^{\geq n}x\, (\psi\circ \varphi))(d_1,\ldots,d_k),
\end{align*}
where \(\circ\in\{\land,\lor\}\) and \(d_1,\ldots,d_k\in\mn{adom}(D)\).
\end{lemma}

\begin{proof}
    We only proof the first statement. The second one follows from the commutativity of \(\land\) and \(\lor\).
    Let \(D\) be a database with at least \(n\) elements and let \(d_1,\ldots,d_k\in\mn{adom}(D)\).
    
    \begin{align*}
        D\models (\varphi\lor\exists^{\geq n}x\, \psi)(d_1,\ldots,d_k) &\iff D\models \varphi(d_1,\ldots,d_k)\text{ or }  D\models (\exists^{\geq n} x\, \psi)(d_1,\ldots,d_k)\\
        \intertext{Since there are at least \(n\) values in \(D\) and \(x\) is not free in \(\varphi\):}
        &\iff D\models(\exists^{\geq n} x\,\varphi)(d_1,\ldots,d_k)\text{ or }  D\models (\exists^{\geq n} x\, \psi)(d_1,\ldots,d_k)\\
        &\iff D\models (\exists^{\geq n}x\, (\varphi\lor\psi))(d_1,\ldots,d_k). 
    \end{align*}
    The statement with \(\land\) instead of \(\lor\) can be proven analogously.
    For this case we do not need that \(D\) has at least \(n\) values, since this follows from \(D\models (\exists^{\geq n} x\, \psi)(d_1,\ldots,d_k)\).
\end{proof}

Databases that do not have enough values trivially cannot satisfy counting quantifiers with large constants.
\begin{lemma}\label{exists_eq_bot_for_small_db}
    Let \(\varphi(x_1,\ldots,x_k)\) be a UNFOC formula, let \(n\in\mathbb{N}\) and let \(D\) be a database with less than \(n\) values and let \(d_1,\ldots,d_k\in\mn{adom}(D)\). Then \(D\not\models\exists^{\geq n} x\, \varphi(d_1,\ldots,d_k)\).
\end{lemma}

We also state De~Morgan laws for UNFOC and some well known equivalences for tautologies and formulas which are not satisfiable.
\begin{lemma}\label{unfoc_de_morgan}
   Let \(\neg(\varphi\lor \psi)\) and \(\neg(\varphi\land\psi)\) be two UNFOC formulas.
   Then \(\neg\varphi\land\neg\psi\) and \(\neg\varphi\lor\neg\psi\) are also UNFOC formulas and
   \begin{align*}
       \neg(\varphi\lor \psi) &\equiv \neg\varphi\land\neg\psi\\
       \neg(\varphi\land \psi) &\equiv \neg\varphi\lor\neg\psi.
   \end{align*}
\end{lemma}
\begin{proof}
    Since  \(\neg(\varphi\lor\psi)\) is a UNFOC formula, \(\varphi\lor\psi\) has at most one free variable, thus \(\varphi\) and \(\psi\) both have at most one free variable.
    Therefore \(\neg\varphi\land\neg\psi\) is still a UNFOC formula.
    The same argument holds for \(\lor\) and \(\land\) swapped.
    The equivalences are well known De~Morgan laws.
\end{proof}
\begin{lemma}\label{top_bot_equivalences}
    Let \(\top\) be formula that is always satisfied, \(\bot\) be a formula that is never satisfied.
    Then the following equivalences hold for all formulas \(\varphi\):
    \begin{alignat*}{3}
        \neg\bot&\equiv\top\qquad&\bot\lor\varphi&\equiv \varphi\qquad&\bot\land\varphi&\equiv\bot\\
        \neg\top&\equiv\bot&\top\lor\varphi&\equiv\top&\top\land\varphi&\equiv\varphi.
    \end{alignat*}
\end{lemma}

\unfocghmlequiv*
\begin{proof}
  By definition, GHML is a fragment
  of UNFOC.
  Conversely, every UNFOC  formula \(\varphi\) in one free variable can be converted into a GHML formula as
  follows:
  We first use variable renaming such that no variable is bounded twice and no variable occurs free and bounded in \(\varphi\).
  Let \(c\) be the largest counting constant in \(\varphi\).
  Then \(\varphi\) is equivalent to 
  \[(\exists^{\geq c}x\, \top\land \varphi)\,\lor\bigvee_{1\leq i <c} (\exists^{\geq i}x\,\top\land \neg\exists^{\geq i+1}x\,\top \land \varphi).\]
  This is a UNFOC formula, since UNFOC is closed under conjunctions and disjunctions.

  The transformation of each occurrence of \(\varphi\) into GHML will be done in two steps.
    In the first step we will group quantifiers such that every subformula \(\exists^{\geq n} x\, \psi\) either has only one free variable or is preceded by another existential quantifier.
    In the second step we will transform these formulas into GHML.
  
  For the first disjunct \(\exists^{\geq c} x\,\top\land \varphi\), it suffices to apply transformations that are correct for all databases with at least \(c\) values but which might be incorrect for databases with less than \(c\) values.
  In the first step, we apply Lemma~\ref{counting_quantifier_scope_change} until every subformula \(\exists^{\geq n} x\, \psi\) either has only one free variable or is preceded by another existential quantifier.
  This can be done, since formulas \(\exists^{\geq n}x\, \psi\) with more than one free variable cannot be preceded by a negation in UNFOC.

  The second step is by induction on the quantifier depth.
  Let \(\psi\) be a subformula of \(\varphi\) that has one free variable and that starts with an existential quantifier.
  It has the form 
  \[\exists^{\geq k_1}x_1\cdots\exists^{k_n}x_n\,\xi(x,x_1,\ldots,x_n),\]
  where \(\xi\) is a Boolean combination of relational atoms and UNFOC formulas in one free variable.
  By induction hypothesis, the UNFOC formulas in one free variable in \(\xi\) can be transformed into equivalent GHML formulas.
  By using De~Morgan laws in Lemma~\ref{unfoc_de_morgan}, the negations can be pushed inwards in front of the GHML formulas.
  Distributivity laws can then be used to transform \(\xi\) into disjunctive normal form, that is, \[\xi\equiv \bigvee_{1\leq j\leq m} \xi_j,\]
  where each \(\xi_j\) is a conjunction of relational atoms, GHML formulas and negations thereof. 
  In particular, only unary GHML formulas are negated.
  Thus the resulting formula \(\varphi_c'\) is in GHML.
  Then \[\exists^{\geq c} x\,\top\land \varphi\equiv \exists^{\geq c}x\,\top \land \varphi_c',\]
  since these formulas are only satisfied by databases with at least \(c\) values.
  The second formula can be expressed in GHML as
  \[\exists^{\geq 1} z \, (\exists^{\geq c} x\, \top \land \varphi'_c),\]
  where \(z\) is a variable that does not occur in \(\varphi_c'\).

  For the other disjuncts, fix some \(1\leq i<c\).
  First find in \(\varphi\) every occurrence of \(\exists^{\geq n} x\,\psi\), where \(n >i\), and replace these subformulas with \(\bot\).
  Denote the resulting formula by \(\varphi_i\).
  By Lemma~\ref{exists_eq_bot_for_small_db} \[\exists^{\geq i}x\,\top\land \neg\exists^{\geq i+1}x\,\top \land \varphi\equiv \exists^{\geq i}x\,\top\land \neg\exists^{\geq i+1}x\,\top \land \varphi_i,\]
  since \(\exists^{\geq i}x\,\top\land \neg\exists^{\geq i+1}x\,\top\) can only be satisfied by databases with exactly \(i\) values.
  Now we can argue as before to transform this formula into GHML.
  In the first step, we can use Lemma~\ref{counting_quantifier_scope_change} to pull out existential quantifiers.
  In the second step we can use De~Morgan laws, distributivity laws and Lemma~\ref{top_bot_equivalences} to transform \(\varphi_i\) into a GHML formula \(\varphi_i'\).
  Then \[\exists^{\geq i}x\,\top\land \neg\exists^{\geq i+1}x\,\top \land \varphi_i\equiv \exists^{\geq i}x\,\top\land \neg\exists^{\geq i+1}x\,\top \land \varphi_i'.\]
  As before, the second formula can be expressed in GHML.

  Since GHML is closed under (unary) disjunction, \(\varphi\) is equivalent to the GHML formula 
  \[(\exists^{\geq c}x\, \top\land \varphi_c')\,\lor\bigvee_{1\leq i <c} (\exists^{\geq i}x\,\top\land \neg\exists^{\geq i+1}x\,\top \land \varphi_i').\]

\end{proof}






\section{Proofs for Section~\ref{sect:DHNs}}

\subsection{Max Aggregation: Expressive Power}

%
%
We prove Point~1 of Theorem~\ref{thm:combi}.
Since UNFO and HML are known to have the same expressive power \cite{DBLP:journals/corr/SegoufinC13}, we can disregard UNFO and work only with HML.
\begin{lemma}\label{lem:UNFO-to-max-DHN-uniform}
    Every  HML formula is equivalent to a simple  $\maxagg$-DHN.
  \end{lemma}
\begin{proof}
  Let $\varphi$ be a  HML formula. The set of \emph{subformulas} of $\varphi$ is defined in the expected way, according to the syntax rule used to define HML formulas. Fix an enumeration  $\varphi_1,\ldots,\varphi_k$ of the subformulas of $\varphi$ such that if $\varphi_\ell$ is a subformula of $\varphi_k$,
then $\ell \leq k$. In particular, this implies that conjunctions of unary atoms that are used in $\varphi$ must occur at the beginning of the enumeration.

We construct a $\maxagg$-DHN $\calN=(\calL_1,\dots,\calL_k,\cls)$ with $k$ layers, all of output dimension $k$.
To encode satisfaction of a subformula $\varphi_i$ 
at a value~$v$, we store a $1$ in the $i$-th component of the feature vector of $v$; likewise, falsification
of $\varphi_i$ is encoded by $0$.
The DHN evaluates one subformula in each layer
so that for every $1\leq i\leq k$, from layer $i$ on all subformulas $\varphi_1,\ldots,\varphi_i$ are  encoded correctly.

We next define the layers $\calL_1,\dots,\calL_k$. Recall
that, to define a layer $\calL_i=(\calF_i,\com_i)$, we 
need to define a sequence of
homomorphism queries 
$\calF_i$ and a simple combination function $\com_i$.
All homomorphism queries will be of output dimension $k$.
We make a case distinction according to 
the form of the formula $\varphi_i$.

First assume that $\varphi_i$
takes the form $\varphi_{i_1} \vee \varphi_{i_2}$. Then $\calF_i$ consists of the homomorphism query $(F^\bullet,\mu,\maxagg)$
where
\begin{itemize}

    \item $F^\bullet$ is the single-value database without any facts.

  \item $\mu$ maps every 
        tuple $\bar x \in \mathbb{R}^k$ to the
        tuple that can be obtained from $\bar x$ by putting 
        \(\max(x_{i_1},x_{i_2})\) into the $i$-th position.

\end{itemize}
The combination function $\com_i$ is the identity.
It is straightforward to verify
that these transformation and combination functions are affine with 
$\trrelu$ activation.
For the function $\mu$ we need to choose an appropriate matrix $A$ and bias vector $\bar b$.
We can choose the identity matrix updated
so that entry $(i,i)$ is 0 rather than~1, and entries $(i_1,i)$ and \((i_2,i)\) are 1 rather than 0.
This yields sum rather than
max, but this is the same since the
component-wise inputs are guaranteed to be $(0,0)$, $(1,0)$, \((0,1)\), or $(1,1)$, and the $\trrelu$ corrects the sum of 2 back to 1,
attaining max as desired.

If \(\varphi_i = \neg\varphi_j\), then \(\mathcal{F}_i\) consists of the homomorphism query \((F^\bullet,\mu, \max)\), where \(F^\bullet\) is the single-value database without any facts and \(\mu\) maps every tuple \(\bar{x}\in \mathbb{R}^k\) to the tuple that can be obtained from \(\bar{x}\) by putting \(1-x_j\) into the \(i\)-th position.
The combination function is then the identity function.
For \(\mu\) we can choose the identity matrix updated such that the entry \((i,i)\) is \(0\), the entry \((j,i)\) is \(-1\), and the \(i\)-th entry in the bias is \(1\).

Finally assume that $\varphi_i$ takes the form 
$\exists \bar y\, \psi(x,\bar y)$.
It is equivalent to the form
\[\exists \bar y\,\Big(\psi_0\land \bigwedge_{z\in \{x\}\cup \bar{y}}\psi_z\Big),\]
where \(\psi_0\) is a conjunction of relational atoms and each \(\psi_z\) is a conjunction of (zero or more) HML formulas with the free variable \(z\).
Then $\calF_i$ consists of two homomorphism queries, one for evaluating all conjuncts in \(\psi\), and one for copying the previous embedding.

The first homomorphism query \((F_1^x, \mu_1,\maxagg)\) is defined as follows:
\begin{itemize}

  \item The active domain of \(F_1\) contains $x$ and the variables in $\bar y$,
  \item \(F_1\) contains a fact \(R(z_1,\ldots,z_k)\) for every atom \(R(z_1,\ldots,z_k)\) in \(\psi_0\),
  \item and \(\mu_1\) assigns to each value \(z\in \mn{adom}(F_1)\) the transformation function that maps every tuple \(\bar{x}\in \mathbb{R}^k\) to the tuple \((1)\) if the \(j\)-th position of \(\bar{x}\) is \(1\) for all conjuncts \(\varphi_j(z)\) in \(\psi_z\), and \((0)\) otherwise.
        It also returns \((1)\) if \(\psi_z\) is the empty conjunction.
\end{itemize}

The second query \((F_2^\bullet, \mu_2,\maxagg)\) is defined as follows:
\begin{itemize}
  \item a single-value database without any facts,
  \item if \(i=1\), \(\mu_2\) returns the \(k\)-dimensional vector with all entries set to \(0\), and
  \item if \(i > 1\), \(\mu_2\) is the identity function.
\end{itemize}
The combination function $\com_i$ then moves the result of the first homomorphism query to position \(i\) and deletes the \(i\)-th entry in the result of the second query.
It is again straightforward to verify that all transformation functions are affine with \(\trrelu\) activation.

As the classification function we use
    \[\cls(\bar{x}) = \begin{cases}
        1 & \text{ if } x_k \geq 1\\
        0 & \text{ otherwise}.
    \end{cases}\]
It is straightforward to show the correctness of the translation, that is:
\\[2mm]
{\bf Claim.} For all databases $D$, values $v \in \mn{adom}(D)$, and $i,j$
with $1 \leq i \leq j \leq k$, the following holds: \[(\lambda^j_{\calN,D}(v))_i = \begin{cases}
    1&\text{ if } D\models \varphi_i(v),\\
    0&\text{ otherwise.}
\end{cases}\]
\end{proof}



%
To show the second direction of Point~1 of Theorem~\ref{thm:combi}, we first prove that, across all \(\Sbf\)-databases, the set of embedding vectors ever computed by a $\maxagg$-DHN is finite.
\begin{lemma}
\label{lem:finitelymanyvectors}
Let \(\Sbf\) be a schema and let $\mathcal{N}=(\calL_1, \dots \calL_k,\cls)$ be a $\maxagg$-DHN over \(\Sbf\).
Then for $1 \leq \ell \leq k$, the set 
    \[
      \chi^\ell_{\calN}= \{ \lambda^\ell_{\calN,D}(v) \mid D \text{ an \(\Sbf\)-database and } v \in \mn{adom}(D) \}
    \]
    is finite.
\end{lemma}
\begin{proof}
 We proceed  by induction on~$\ell$.
  For the induction start,
$\chi^0_{\calN}$ is the set that contains only the empty vector. For the induction step, let $\ell >0$.
Let  $\calL_\ell=(\calF_\ell,\com_\ell)$,
 with $\calF_\ell$ the sequence $(F^\bullet_1,
\mu_{1},\agg_1), \dots, (F^\bullet_m, \mu_{m},\agg_m)$.
Then the set $\chi^\ell_{\calN}$ 
consists only of vectors $\bar x$  that can be
obtained by choosing, for
$1 \leq i \leq m$, a set $S_i$ of mappings $\lambda:\mn{adom}(F_i) \rightarrow \chi_{\calN}^{\ell-1}$
and setting ${\bar x = \com_\ell(\bar y_1,\dots, \bar y_m)}$ where, for $1 \leq i \leq m$,
\begin{align*}
   \bar y_i := \underset{\lambda \in S_i}{\maxagg} \prod_{v \in \mn{adom}(F)} \mu_v(\lambda(v)).
\end{align*}
Clearly,
this implies that $\chi^\ell_{\calN}$ is finite. The crucial difference to sum and mean aggregation  is that we can work with a set  rather than with multisets since max aggregation
is indifferent to multiplicities.\end{proof}

\begin{lemma}\label{lem:max-DHN-to-UNFO-uniform}
    Every  $\maxagg$-DHN $\calN$ is equivalent to a  HML formula.
\end{lemma}
\begin{proof}
    Let $\calN = (\calL_1, \dots \calL_k,\cls)$ be a  $\maxagg$-DHN. By
    Lemma~\ref{lem:finitelymanyvectors},
    we may build an equivalent HML formula
    $\varphi(x)$ by constructing, for every
    $\ell \in \{1,\dots,k\}$ and embedding vector $\bar x \in \chi^\ell_{\calN}$,
    a  HML formula $\psi^\ell_{\bar x}(x)$ such that
    \[
      D \models \psi^\ell_{\bar x}(v) \text{ iff } \lambda^\ell_{\calN,D}(v)=\bar x
      \text{ for all pointed \(\Sbf\)-databases } D^v
    \]
    and then taking the disjunction of all
    formulas $\psi^k_{\bar x}$ such that $\cls(\bar x)=1$.

    The construction of the formulas  $\psi^\ell_{\bar x}$ is by induction on $\ell$. For $\ell=0$, the only relevant vector $\bar x$ is the empty
    vector and we may set $\psi^0_{()}(x)=\text{true}$. Now let $\ell>0$ and $\bar x \in \chi^\ell_{\calN}$. Further let  $\calL_\ell=(\calF_\ell,\com_\ell)$,
 with $\calF_\ell$ the sequence $(F^{u_1}_1,
\mu_{1},\maxagg), \dots, (F^{u_m}_m, \mu_{m},\maxagg)$. We define $\psi^\ell_{\bar x}$ to be a disjunction that contains one disjunct for every possible way to  choose for every
$i \in \{1,\dots,m\}$ a partition $X_i,Y_i$ of the set of mappings $\lambda:\mn{adom}(F_i) \rightarrow \chi_{\calN}^{\ell-1}$ such that 
$\com_\ell(\bar y_1,\dots, \bar y_m) = \bar x$ where, for $1 \leq i \leq m$,
\begin{align*}
   \bar y_i := \underset{\lambda \in X_i}{\maxagg} \prod_{v \in \mn{adom}(F)} \mu_v(\lambda(v)).
\end{align*}
Note that the definition of $\bar y_i$ corresponds to the evaluation of the homomorphism query $(F^{u_i}_i,
\mu_{i},\maxagg)$ under any homomorphism $h$ that is 
`compatible' with  some $\lambda \in X_i$, that is, $h$ maps each $v \in \mn{adom}(F_i)$ to a value
$h(v)$ in the input graph that carries the embedding vector $\lambda(v)$. Clearly, we may view
the  database $F_i$ as a
conjunction $\vartheta_i$ of relational
atoms, that is, $\vartheta_i$ contains
an atom $R(v_1,\ldots,v_k)$ for each fact 
$R(v_1,\ldots,v_k) \in F_i$ where values of
$F_i$ are viewed as variables except that
the distinguished value $u_i$ is renamed to the variable $x$ (and thus shared by all $\vartheta_1,\dots,\vartheta_m$). Let $\bar z_i$ denote the tuple of variables in $\vartheta_i$, ordered arbitrarily but without $x$. Then $\psi^\ell_{\bar x}(x)$
includes as a disjunct
\[
  \bigwedge_{1 \leq i \leq m} \Big (\bigwedge_{\lambda \in X_i}
  \exists \bar z_i \, \big (\vartheta_i
  \wedge \bigwedge_{v \in \mn{adom}(F_i)} \psi^{\ell-1}_{\lambda(v)}(v)
  \big )
  \wedge
 \bigwedge_{\lambda \in Y_i}
  \neg
  \exists \bar z_i \, \big (\vartheta_i
  \wedge \bigwedge_{v \in \mn{adom}(F_i)} \psi^{\ell-1}_{\lambda(v)}(v)
  \big ) \Big ).
\]
It is easy to verify that $\psi^\ell_{\bar x}$ is as required.
\end{proof}

\subsection{Sum Aggregation: Expressive Power}

We do not provide an explicit proof of Point~2 of Theorem~\ref{thm:combi}, because it is a consequence of the (later) Lemma~\ref{lem:UQAFOUNFOC} and Theorem~\ref{thm:DENcombi}.

\medskip


When translating GHML to $\sumagg$-DHNs, it is 
convenient to pass through a weaker variant GHML$^-$ of GHML. GHML$^-$ is defined like GHML, except that  counting formulas now take the form 
$\exists^{\geq k} \bar y \, \psi(x, \bar y)$ with $\psi$ a conjunction of relational atoms and unary GHML$^-$ formulas. Connectedness is defined
as for HML. Note that GHML$^-$ differs from GHML in two ways: (i)~there is now a \emph{single} counting quantifier over a \emph{tuple} of variables rather 
than a sequence of such quantifiers over single variables, and (ii)~$\psi$ is a conjunction rather than a disjunction of conjunctions.
We conjecture that either change decreases the expressive power of GHML. 
%
%
For instance, 
$\exists^{\geq 2} y_1 \exists^{\geq 2} y_2 \, \big ( E(x,y_1) \wedge E(x,y_2) \wedge  E(y_1,y_2)\big )$ is a GHML formula that seems impossible to express even when only restriction~(i) is adopted, and 
$\exists y\exists^{\geq 2}z (S(x,y,z)\lor R(x,y,z))$ seems hard
to express even when only restriction~(ii) is used. The reader might want to contrast this
with HML, for which it is easy to see
that neither of these variations makes a difference in expressive power. On databases of bounded degree, however, GHML and GHML$^-$ have the same expressive power.
 To show this, we first introduce some
preliminaries.

  Let \(F^v, D^u\) be  pointed databases,  \({\tau \subseteq \{x\neq y\mid x,y\in \mn{adom}(F)\}}\)  a set of inequalities,
  and let \({\Phi: \mn{adom}(F)\rightarrow L}\), where \(L\) is a set of unary formulas.
  A homomorphism $h$ from $F^v$ to $D^u$ \emph{satisfies}
  $\tau$ if $h(x)\neq h(y)$  for all $x\neq y\in \tau$
  and it \emph{satisfies} $\Phi$ if
  $D \models \Phi(x)(h(x)) \text{ for all } x\in \mn{adom}(F)$.
  We use $\Hom(F^v, D^u, \tau, \Phi)$ to denote the set of
  homomorphisms $h\in \Hom(F^v, D^u)$ that satisfy both $\tau$ and $\Phi$, and
  $\Inj(F^v, D^u, \Phi)$
  to denote the set of injective
  homomorphisms $h\in \Hom(F^v, D^u)$ that satisfy  $\Phi$.
%
%

We next observe that
we can compute the cardinality
of $\Inj(F^v, D^u,\Phi)$
from the cardinalities of certain
sets $\Hom(H^w, D^u,\emptyset, \Psi)$.
%
%
  Let \(F^v\) be a pointed database and
%
  let \(P\) be a partition of \(\mn{adom}(F)\).
%
For  
  \(\Phi: \mn{adom}(F)\rightarrow L\), we define \(\Phi/P\) as the function that maps \([v]_P\) to 
  $\bigwedge_{u\in [v]_P} \Phi(u).$

\begin{lemma}\label{lemma: compute number of injective hom}
  Let \(F^v\) be a pointed database and  \(\calP\) the set of all partitions of \(\mn{adom}(F)\).
  Then for every \(P\in \calP\) there exists  a coefficient \(\mu_P\) such that for all databases \(D^u\):
  \begin{align*}
    |\Inj(F^v, D^u,\Phi)| = \sum_{P\in\calP} \mu_P |\Hom(F^v/P, D^u,\emptyset, \Phi/P)|.
  \end{align*}
\end{lemma}
\begin{proof}
The proof of
Theorem~\ref{lovasz homomorphism count implies embedding count}
can be extended straightforwardly to yield this result.
\end{proof}

We also state distributivity laws for large conjunctions and disjunctions.
\begin{lemma}\label{lemma:large_distributivity}
    Let \(n,m\in \mathbb{N}\) and let \(\varphi_{i,j}\) be a UNFOC formula for each \(1\leq i\leq n\) and \(1\leq j\leq m\).
    Let \(\calF =\{f: \{1,\ldots,n\}\to \{1,\ldots,m\}\}\).
    Then the following equivalence holds:
    \[\bigwedge_{i=1}^n\bigvee_{j=1}^m \varphi_{i,j} \equiv \bigvee_{f\in \calF}\bigwedge_{i=1}^n \varphi_{i,f(i)}.\]
\end{lemma}
\begin{proof}
To keep this proof simple, we only show this lemma for \(\varphi_{i,j}\) without free variables. It is straightforward to extend this proof to formulas with free variables.

   If a database \(D\) satisfies the left side, then for each \(i\) exists a \(j_i\) such that \(D\) satisfies \(\varphi_{i,j_i}\).
   One can now define \(f(i)=j_i\). \(D\) then satisfies \(\bigwedge_{i=1}^n\varphi_{i,f(i)}\) and hence also the formula on the right side.

   Conversely, if a database \(D\) satisfies the right side, then there exists a \(f\in \cal F\) such that \(D\) satisfies each \(\varphi_{i,f(i)}\).
   Thus, the database satisfies \(\bigvee_{j=1}^m\varphi_{i,j}\) for each \(1\leq i\leq n\).
   Therefore, \(D\) also satisfies the formula on the left side.
\end{proof}

\begin{lemma}
\label{lem:GHMLvsGHMLminus}
   Let $B \geq 0$. Then
   \begin{enumerate}

       \item every  (connected) GHML$^-$ formula is equivalent to a (connected) GHML formula
       and
       
       \item on databases of degree at most $B$, every (connected) GHML formula is equivalent to a (connected) GHML$^-$ formula.
       
   \end{enumerate}
   %
\end{lemma}
\begin{proof}
  We start with Point~1.
  In view of Lemma~\ref{lem:twocombined}, it suffices to show that 
 UNFOC can express tuple counting quantifiers $\varphi(\bar{x}) = \exists^{\geq n} \bar{y}\,\psi(\bar{x}, \bar{y}).$
 The proof is by induction on the length of \(\bar{y}\).
 For \(|\bar{y}| = 1\) there is nothing to do, since UNFOC can quantify single variables.

 For the induction step, let \(\bar{y} = (y_0, \bar{y}')\).
 
 An \emph{integer partition} of an $n \in \mathbb{N}$ is a sequence of positive
 integers $s = s_1, \ldots, s_\ell$ (with $\ell \geq 0$) such that $n = \sum_{1 \leq i \leq \ell} s_i$.
 Let $S_n$ be the set of all integer partitions of $n$.
 For a given integer partition $s = s_1, \ldots, s_\ell$ of $n$ and $k \in
 \mathbb{N}$, let $\#_{\geq k} (s)$ be the number of integers
 $s_i$ such that $s_i \geq k$.
 Given a database \(D\) and an interpretation \(\bar{a}\) for \(\bar{x}\), there exist \(n\) tuples \(\bar{b}_i\) such that \(D\models \psi(\bar{a},\bar{b}_i)\) if and only if there exist \(\ell\) values \(c_1,\ldots,c_\ell\) and an integer partition \(s_1,\ldots,s_\ell\) of \(n\), such that for each \(c_i\) there exist \(s_i\) tuples \(\bar{d}_i^1,\ldots,\bar{d}_i^{s_i}\) such that \(D\models \psi(\bar{a}, c_i, \bar{d}_i^{j})\).
 
In UNFOC we can express the latter property with the formula
\[
  \bigvee_{s \in S_n} \bigwedge_{1 \leq m \leq n} \exists^{\geq (\#_{\geq m}(s))} y_0 \,\exists^{\geq m} \bar{y}'\,\psi(\bar{x},y_0,\bar{y}').
\]
By induction hypothesis, each \(\exists^{\geq m}\bar{y}'\,\psi\) can be replaced with an equivalent UNFOC formula, resulting in a UNFOC formula that expresses the tuple counting quantifier \(\exists^{\geq n} \bar{y}\,\psi(\bar{x},\bar{y})\)
This concludes the proof of Point~1.

For Point~2, the proof is by induction on the structure of GHML formulas, and it suffices to treat the existential quantifiers.
 We start with GHML formulas that satisfy a weaker notion of connectedness.
 Let \(\psi(x) = \exists^{\geq k_1}x_1\cdots\exists^{\geq k_n}x_n \,\bigvee_{1\leq i\leq \ell}(\mu_i(x,x_1,\ldots,x_n)\land \nu_i(x,x_1,\ldots,x_n))\) be a GHML formula where each \(\mu_i\) is a conjunction of relational atoms and \(\nu_i\) is a conjunction of the remaining unary formulas.
 We say that $\psi$ is  \emph{top-level connected} if the Gaifman graph of \(\mu_i\) viewed as a database is connected and \(x\) and every \(x_i\) occurs in \(\mu_i\). 
 Notice that this definition does not require  the formulas in \(\nu_i\) to be  connected.
 For formulas $\psi$ of the same form, but without a free variable, top-level connectedness is defined analogously.

 Let \(\psi(x) = \exists^{\geq k_1}x_1\exists^{\geq k_2}x_2\cdots\exists^{\geq k_n}x_n\,\bigvee_{1\leq i\leq m}\varphi_i(x,x_1,\ldots,x_n)\) be a GHML formula. We first address the case where $\psi$ is top-level connected.
 It is equivalent to the following UNFOC formula with inequalities:
 \[\psi' = \exists \bar x_1 \cdots \exists \bar x_n \, (\tau\land \varphi'),\]
 where $\bar x_1,\dots,\bar x_n$ are tuples of
 variables as follows:
 \[
 \begin{array}{rcl}
   \bar x_1 &=& x^1_1 \cdots x^{k_1}_1 \\[1mm]
   \bar x_2 &=& x_2^{1,1} x_2^{1,2} \cdots x_2^{1,k_2} x_2^{2,1} \cdots x_2^{2,k_2} x_2^{3,1} \cdots x_2^{k_1,k_2} \\[1mm]
   && \cdots \\[1mm]
   \bar x_{n} &=& x_n^{1,\dots,1} \cdots x_n^{k_1,\dots,k_n}
 \end{array}
 \]
 and
 \[
 \begin{array}{rcl}
 \tau &=& \displaystyle\bigwedge_{j=1}^{n}\bigwedge_{\ell_1 = 1}^{k_1}\cdots \bigwedge_{\ell_{j-1}=1}^{k_{j-1}} \bigwedge_{1 \leq i < i' \leq k_j}  x_j^{\ell_1,\dots,\ell_{j-1},i}\neq x_j^{\ell_1,\dots,\ell_{j-1},i'} \\[7mm] 
\varphi' &=& \displaystyle \bigwedge_{\ell_1 = 1}^{k_1}\cdots \bigwedge_{\ell_n=1}^{k_n} \bigvee_{i=1}^m\varphi_i(x,x_1^{\ell_1},x_2^{\ell_1,\ell_2},\ldots,x_n^{\ell_1,\dots,\ell_n}).
 \end{array}
\]
 
 Let \(\calF =\{f:\{1,\ldots,k_1\}\times\cdots\times\{1,\ldots,k_n\}\to \{1,\ldots,m\}\}\).
 By using distributivity laws and Lemma~\ref{lemma:large_distributivity}, \(\tau\land \varphi'\) can be transformed into disjunctive normal form, that is
 \[\bigvee_{f\in \cal F}\Big(\tau\land\bigwedge_{\ell_1 = 1}^{k_1}\cdots \bigwedge_{\ell_n=1}^{k_n} \varphi_{f(\ell_1,\ldots,\ell_n)}(x,x_1^{\ell_1},x_2^{\ell_1,\ell_2},\ldots,x_n^{\ell_1,\dots,\ell_n})\Big).\]
 
 By using standard FO equivalences, \(\psi'\) is equivalent to
\[\bigvee_{f\in \cal F}\exists\bar x_1\cdots\exists\bar x_n\,\Big(\tau\land\bigwedge_{\ell_1 = 1}^{k_1}\cdots \bigwedge_{\ell_n=1}^{k_n} \varphi_{f(\ell_1,\ldots,\ell_n)}(x,x_1^{\ell_1},x_2^{\ell_1,\ell_2},\ldots,x_n^{\ell_1,\dots,\ell_n})\Big).\]

 Since GHML$^-$ is closed under unary disjunction, it suffices to translate each
 \[\exists\bar x_1\cdots\exists\bar x_n\,\Big(\tau\land\bigwedge_{\ell_1 = 1}^{k_1}\cdots \bigwedge_{\ell_n=1}^{k_n} \varphi_{f(\ell_1,\ldots,\ell_n)}(x,x_1^{\ell_1},x_2^{\ell_1,\ell_2},\ldots,x_n^{\ell_1,\dots,\ell_n})\Big)\]
 into GHML$^-$. Thus, we now fix one \(f\in \calF\).

 We may equivalently write the inner conjunction in the form
 \[\tau\land \bigwedge\xi_{0}\land \bigwedge \xi_i^{\ell_1,\ldots,\ell_i},\]
 where \(\xi_{0}\) is a conjunction of relational atoms and each \(\xi_i^{\ell_1,\ldots,\ell_i}\) is a conjunction of unary GHML formulas with the free variable \(x_i^{\ell_1,\ldots,\ell_i}\).
 Let \(\Xi\) be the function that maps \(x_i^{\ell_1,\ldots,\ell_i}\) to \(\xi_i^{\ell_1,\ldots,\ell_i}\).
 
 Let \(F^x\) be the database defined by \(\xi_0\).
 This database is connected: let \(x_i^{\ell_1,\ldots,\ell_i}\) be a variable in \(\bar{x}_i\), then \(x\) and \(x_i^{\ell_1,\ldots,\ell_i}\) are connected in the formula \(\varphi_{f(\ell_1,\ldots,\ell_i,1\ldots,1)}\), because \(\psi\) is top-level connected.
 This connection is also present in \(F^x\).

 Also, each of the unary formulas has a quantifier depth bounded by the quantifier depth of \(\bigvee_{1\leq i\leq m}\varphi_i\).
 Thus, by the induction hypothesis, each \(\varphi_{f(\ell_1,\ldots,\ell_n)}\) can be written as a conjunction of GHML$^-$ formulas.
 Additionally, also by induction hypothesis, if \(\psi\) is connected, then each of these GHML$^-$ formulas is connected.

 Now we have for each pointed database \(D^v\):
 \begin{align*}
   D \models \psi'(v) &\iff |\Hom(F^x, D^v, \tau, \Xi')| \geq 1
 \intertext{Let \(\calP_\tau\) be the set of all partitions of \(\mn{adom}(F)\) that satisfy \(\tau\). By construction of the databases $F^x/P$, we obtain}
   |\Hom(F^x, D^v, \tau, \Xi')| \geq 1 &\iff \sum_{P\in\calP_\tau} |\Inj(F^x/P, D^v, \Xi'/P)|\geq 1\\
     &\iff \bigvee_{P\in\calP_\tau} (|\Inj(F^x/P, D^v, \Xi'/P)|\geq 1).
 \end{align*}

Since GHML$^-$ is closed under unary disjunction, it thus remains to show how to express \(|\Inj(F^x/P, D^v,\Xi'/P)| \geq 1\).
 Notice that since \(F^x\) is connected, \(F^x/P\) is also connected for all \(P\).
 To simplify notation, we show how to express \(|\Inj(H^u, D^v, \xi)| \geq 1\) for all connected \(H^u\) and \(\xi\) that map \(\mn{adom}(H)\) to conjunctions of unary GHML$^-$ formulas.
 In this step we shall exploit the bounded degree of \(D^v\).
 
 Let \(H^u\) be a connected pointed database, \(\xi: \mn{adom}(H)\rightarrow L\) be a mapping with \(L\)  the set of conjunctions of unary GHML$^-$ formulas,
 and \(B\in\mathbb{N}\). We show that there exists a GHML$^-$ formula \(\chi\) such that \[D\models \chi(v) \iff |\Inj(H^u, D^v, \xi)| \geq 1 \] for all pointed databases \(D^v\) of degree at most \(B\).

 Let \(\calP\) be the set of all partitions of \(\mn{adom}(H)\).
 By Lemma \ref{lemma: compute number of injective hom} it suffices to find a GHML$^-$ formula that expresses \[\sum_{P\in\calP}\mu_P |\Hom(H^u/P, D^v, \emptyset, \xi/P)| \geq 1.\tag{\(*\)}\label{goal}\]
 The property \(|\Hom(H^u/P, D^v, \emptyset, \xi/P)|\geq n\), with \(P = \{p_0, p_1,\ldots, p_k\}\) and \(u\in p_0\),
 is expressible by the unary GHML$^-$ formula
 \[\rho (x_0) = \exists^{\geq n} \bar x \, (\bigwedge_{R(v_{i_1},\dots,v_{i_\ell})\in H^u/P} R(x_{i_1}, \dots,x_{i_\ell})\land \bigwedge_{i=1}^k \xi_i(x_i)),\]
 where $\bar x = x_1\cdots x_k$ and \(\xi_i = (\xi/P)(p_i)\).
 If each formula \(\xi(v)\) is connected, then the above formula is also connected, since \(H^u/P\) is a connected database and each \(\xi/P\) is connected.
 
 Since $H^u$ is connected and the degree of all relevant databases \(D^v\) is bounded by \(B\), the number of homomorphisms in \(\mn{Hom}(H^u/P, D^v,\emptyset, \xi/P)\) is also bounded. 
 Therefore, Property~(\ref{goal}) can be expressed as a finite Boolean combination of unary GHML$^-$ formulas, which itself is again a GHML$^-$ formula.
 This also preserves connectedness.
 
 We now consider the case where \(\psi(x) = \exists^{\geq k_1}x_1\cdots\exists^{\geq k_n}x_n \, \varphi(x,x_1,\ldots,x_n)\)  is not  top-level connected.
 Let \(k =\max(k_1,\ldots,k_n)\) then \[\psi(x)\equiv \exists^{\geq k}y\, \top \land \exists^{\geq k_1}x_1\cdots\exists^{\geq k_n}x_n \, \varphi(x,x_1,\ldots,x_n).\]
 It is now safe to remove \(\exists ^{\geq k_i}x_i\) if \(x_i\) does not occur in \(\varphi\).
Fix some \(B\in \mathbb{N}\) and restrict to databases of degree at most \(B\).
Using well-known equivalences of first-order logic, and the previous observation that unused quantifications can be removed, we can rewrite $\psi$ into a conjunction of existentially quantified formulas, each of which is top-level connected; exactly one
of these contains the free variable $x$. More precisely, \(\psi\) can be rewritten into a conjunction of
 \begin{enumerate}
         \item a formula \(\exists^{\geq \ell_1}y_1\cdots\exists^{\geq \ell_m}y_m \, \varphi'(x,y_1,\ldots,y_m)\) that is top-level connected;
         \item formulas \(\exists^{\geq \ell_1}y_1\cdots\exists^{\geq \ell_m}y_m\,\varphi'(y_1,\ldots,y_m)\)
         that are all top-level connected;
         \item the formula \(\exists^{\geq k} y\,\top\)
 \end{enumerate}
such that, for distinct formulas, the sets of variables quantified in the outermost quantifier block are disjoint.  We have already described how to transform the formula in Point~1 into a GHML$^-$ formula that is equivalent over all databases of degree at most \(B\).
 Each formula $\xi$ in Point~2 can be treated as follows. We first translate 
 \(\exists^{\geq \ell_2}y_2\cdots\exists^{\geq \ell_m}y_m\,\varphi'(y_1,\ldots,y_m)\)
 into a GHML$^-$ formula $\psi(y_1)$ that is equivalent over all databases of degree at most \(B\), as above. Then, \(\exists^{\geq \ell_1}y_1\, \psi'\) is equivalent to $\xi$ over all databases of degree at most \(B\) where $\psi'$ is obtained from $\psi$ by adding an atom $\mn{true}(z)$
 for a fresh variable $z$, to ensure that
 there is a free variable.
 The third formula is already in GHML\(^-\).
 Since GHML$^-$ is closed under conjunction, this finishes the proof.
\end{proof}


We now prove Point~3 of Theorem~\ref{thm:combi}, starting with the direction from GHML to $\sumagg$-DHNs.
The translation  does
not depend on  connectedness, but preserves it. By Lemma~\ref{lem:GHMLvsGHMLminus}, we may start from GHML$^-$ in place of GHML.
Note that, while the subsequent translation does not depend on bounded degree databases, Lemma~\ref{lem:GHMLvsGHMLminus} does.
\begin{lemma}\label{thm:UNFOC-to-sum-DHN-uniform}
    Every (connected) GHML$^-$ formula  is equivalent to a (connected) simple $\sumagg$-DHN.
\end{lemma}
\begin{proof}
The proof follows the same strategy as that of Lemma~\ref{lem:UNFO-to-max-DHN-uniform}. The cases where a subformula $\varphi_i$ is a negation or of the form $\varphi_{i_1} \vee \varphi_{i_2}$ can 
be handled in exactly the same way as in the proof of Lemma~\ref{lem:UNFO-to-max-DHN-uniform}. In fact,  all homomorphism queries $(F^v,\mu,\maxagg)$ used in these cases satisfy ${\mn{adom}(F)=\{v\}}$ and therefore
for any $G^\bullet$ the set $\Hom(F^v,G^\bullet)$ has cardinality at most one. As a consequence, 
$\res((F^\bullet,\mu,\sumagg),(G^\bullet,\lambda)) =
\res((F^\bullet,\mu,\maxagg), (G^\bullet, \lambda))$.

We treat the remaining case where $\varphi_i$ takes the form 
$\exists^{\geq n} \bar y \, \psi(x, \bar y)$.
It is equivalent to \[\exists^{\geq n} \,\Big( \psi_0\land\bigwedge_{z\in \{x\}\cup \bar{y}}\psi_z\Big),\]
where \(\psi_0\) is a conjunction of relational atoms and each \(\psi_z\) is a conjunction of GHML$^-$ formulas with the free variable \(z\).
 We again use two homomorphism queries, one for evaluating the conjuncts in \(\psi\) and one for copying the previous embedding.

 The first  homomorphism query \((F_1^x, \mu_1, \sumagg)\) of output dimension $1$ is defined as follows:
 \begin{itemize}
   \item \(F_1\) contains as values the variable \(x\) and all variables in \(\bar{y}\),
   \item \(F_1\) contains the fact \(R(z_1,\ldots,z_k)\) for every atom \(R(z_1,\ldots, z_k)\) in \(\psi_0\),
   \item \(\mu_1\) assigns to each value \(z\in \mn{adom}(F_1^x)\) the transformation function that maps every tuple \(\bar{x}\in \mathbb{R}^k\) to the tuple \((1)\) if the \(j\)-th position of \(\bar{x}\) is \(1\) for all conjuncts \(\varphi_j(z)\) in \(\psi_z\), and the tuple \((0)\) otherwise.
 \end{itemize}
 Notice that \(F_1^x\) is connected if \(\varphi_i\) is connected.

 The second query \((F_2^\bullet, \mu_2,\sumagg)\) of output dimension \(k\) is defined as follows:
 \begin{itemize}
    \item \(F_2^\bullet\) is the graph with exactly one value, without any facts,
   \item if \(i=1\), \(\mu_2\) returns the \(k\)-dimensional vector with all entries set to \(0\), and
   \item if \(i>1\), \(\mu_2\) is the identity function.
 \end{itemize}

 The combination function \(\com_i: \mathbb{R}^{k+1}\to \mathbb{R}^k\)
 \begin{itemize}
   \item copies all entries except the first and the \(i+1\)-st,
   \item discards the \(i+1\)-st entry, and 
   \item subtracts \(n-1\) from the first entry, copies the result to the \(i+1\)-th entry, and applies \trrelu.
 \end{itemize}

 It is easy to verify that \(\mu_1,\mu_2,\) and \(\com_i\) are affine with \relu activation.
 It is also easy to verify that the homomorphism number regarding \((F_1^x,\mu_1, \sumagg)\) returns the number of tuples such that \(\psi(x, \bar{y})\) is satisfied.
 Thus, the combination function returns \(1\) in the \(i\)-th entry if and only if this number is at least \(n\).
\end{proof}

We next translate connected DHNs to connected GHML$^-$, assuming bounded degree databases. Under the assumption of bounded degree, it is possible to prove an analogue of Lemma~\ref{lem:finitelymanyvectors} 
also in the case of sum aggregation.
This is in fact implicit in the proof
of the following result.

%
%
%
\begin{lemma}
\label{lem:DHNtoGHML}
    Let $B \geq 0$.
    On databases of degree at most $B$, every connected DHN is equivalent to a connected GHML$^-$ formula.
\end{lemma}
\begin{proof}
Let $B \geq 0$ and let $\calN=(\calL_1, \dots \calL_k,\cls)$ be a connected DHN. Further let $n$ be the maximum size of databases in homomorphism queries in $\cal N$, measured in the number of values.  Every pointed database $G^v$ induces a pointed database of diameter $kn$, simply by taking the restriction of $G^v$ to the values \(u\) that are in distance of at most \(kn\) from \(v\).
 
  Take two pointed databases $G_1^{v_1}$ and $G_2^{v_2}$.
  It follows directly from the definition of connected DHNs that if $G_1^{v_1}$ and $G_2^{v_2}$ induce the
      same database of diameter $kn$, then
  $\lambda^k_{\calN,G_1}(v_1) = \lambda^k_{\calN,G_2}(v_2)$ and thus
  $\calN$ accepts both of these databases
  or neither of them.
      
      
     
  %

  If a pointed database $G^\bullet$ has diameter at most $kn$ and degree at most $B$, then  the number of values in $G^\bullet$ is at most $(B+1)^{kn+1}-1$, and thus the number of such databases is bounded as well.


   It thus suffices to find a connected GHML$^-$ formula $\varphi(x)$ such that 
   for every pointed database $G^v$ of degree at most $B$, 
   $G \models \varphi(v)$ if and only if
   $G^v$ induces a pointed database $F^v$ of diameter $kn$ (and also of degree at most $B$) such that $\cal N$
   accepts $F^v$. Since there are only finitely many such $F^v$ (up to isomorphism), it suffices to find for each of them 
   a GHML$^-$ formula $\psi_{F^v}$  such that $G \models \psi_{F^v}(v)$ if and only
   if $G^v$ induces exactly $F^v$ and then take the disjunction.

   To construct the formulas $\psi_{F^v}$, in turn, we may use Lov\'asz' homomorphism count theorem.
   We can consider all pointed databases $H^w$ with $|\mn{adom}(H)| \leq |\mn{adom}(F)|$ and count
   the exact number of homomorphism from
   $H^w$ using Boolean combinations of
   formulas of the form $\exists^{\geq k} \bar y \, \psi(w, \bar y)$ where $\psi$
   is the conjunction that has one atom $R(z_1,\ldots,z_k)$ for every fact $R(z_1,\ldots,z_k) \in H$. Based on Theorem~\ref{lovasz homomorphism count implies embedding count}, we can then find a Boolean combination $\psi_{F^v}$ that is as required.
\end{proof}
We remark that in the proof of the previous lemma, the constructed GHML$^-$ formula does not nest the operators $\exists^{\geq k} \bar y \, \psi(x, \bar y)$.
This remark and Lemma~\ref{thm:UNFOC-to-sum-DHN-uniform} show that, on bounded degree databases, every connected GHML$^-$ formula is equivalent to a connected GHML$^{-}$ formula that does not nest these particular operators.

\medskip

We next show that the characterizations in Point~3 of Theorem~\ref{thm:combi} cannot be lifted to databases of unbounded degree, even if we restrict the \(\sumagg\)-DHNs to properties definable in first-order logic.
\sumdhnnotinunfocuniform*
To prove this statement, we will first show that local transitivity can be expressed by \(\sumagg\)-DHNs.
To show that GHML cannot express this property, we will then introduce GHML games.
They are defined analogously to {Ehrenfeucht-Fra\"iss\'e} games, which are a useful tool to prove which properties are (in-)expressible in first-order logic.

\begin{lemma}\label{lem:dhns can express local transitivity}
    There is a connected simple $\sumagg$-DHN over the graph schema that accepts a
    pointed database $D^v$ if and only if $v$ is locally transitive in $D$.
\end{lemma}
\begin{proof}
For this, we use two homomorphism queries $(F^{v_1}, \mu, \sumagg)$ and
$(F^{\prime v_1'}, \mu', \sumagg)$, where ${F = \{ (v_1, v_2), (v_2, v_3) \}}$ and
$\mu_{v_i}$ is the constant $1$ for $i \in \{1, 2, 3\}$, $F' = \{ (v'_1, v'_2),
(v'_2, v'_3), (v'_1, v'_3) \}$, and $\mu'_{v_i}$ is the constant $1$ for $i \in
\{1, 2, 3 \}$.

Observe that by choice of $\mu$ as well as the $\sumagg$,
$\mn{eval}((F^{v_1}, \mu, \sumagg), (D^v, \lambda)) = |\mn{Hom}(F^{v_1}, D^{v})|$ and similarly
$\mn{eval}((F^{\prime v_1'}, \mu', \sumagg), (D^v, \lambda)) = |\mn{Hom}(F^{\prime v_1'}, D^{v})|$.
Clearly, $v$ is locally transitive in $D$ if and only if $|\mn{Hom}(F^{v_1},
D^{v})| = |\mn{Hom}(F^{\prime v_1'}, D^{v})|$.
The desired DHN uses two layers.
The first one computes 
$x_1 = \mn{eval}((F^{v_1}, \mu, \sumagg), (D^v, \lambda))$ and $x_2 =
\mn{eval}((F^{\prime v_1'}, \mu', \sumagg), (D^v, \lambda))$, and then
uses the combination function to compute \(y_1 =\trrelu(x_2-x_1)\) and \(y_2 =\trrelu(x_1-x_2)\).
The second layer computes \(z = \trrelu(1-y_1-y_2)\),
which is \(1\) iff \(x_1 =x_2\).
The final classification function accepts $D^v$ if $z = 1$ and rejects otherwise.
\end{proof}

To show
Proposition~\ref{thm:sum-DHN-not-in-UNFOC-uniform} it remains to prove that local
transitivity cannot be expressed by a GHML formula.

The \emph{depth} of a GHML formula $\varphi$ is defined as the number of times
subformulas of the form 
$ \exists^{\geq k_1} y_1 \cdots \exists^{\geq k_m} y_m \, \bigvee_{1 \leq i \leq \ell} \psi_i(x,y_1,\dots,y_m)$ are nested in $\varphi$.
This means that, there are no GHML formulas of depth $0$, and the formula $ \exists^{\geq 2} y R(x, y)$ has depth $1$.
The \emph{width} of a GHML formula $\varphi$ is the smallest number $n$ such that
$n \geq m$
for all subformulas of the form
$ \exists^{\geq k_1} y_1 \cdots \exists^{\geq k_m} y_m \, \bigvee_{1 \leq i \leq \ell} \psi_i(x,y_1,\dots,y_m)$.
The \emph{counting bound} of a GHML formula $\varphi$ is the smallest number $c$
such that $c \geq k_i$ for all subformulas of the form
$\exists^{\geq k_1} y_1 \cdots \exists^{\geq k_m} y_m \, \bigvee_{1 \leq i \leq \ell} \psi_i(x,y_1,\dots,y_m)$ and
all $i$ with $1 \leq i \leq m$.

A homomorphism from a (non-pointed) database $F$ to a (non-pointed) database $D$
is defined as expected. For a database $D$ and a set $X\subseteq\mn{adom}(D)$, we denote with $D|X$ the restriction of $D$ to the values
in $X$.

\begin{definition}[GHML game] \label{def:counting_unfo_games}
    Let $d, n, c \geq 1$. The $d$-round GHML game of width $n$ with
    counting bound $c$ is played between two players, \emph{Spoiler} and
    \emph{Duplicator}, on two databases $D_1$ and $D_2$. 
    The positions of the game are pairs $(v_1, v_2) \in \mn{adom}(D_1) \times \mn{adom}(D_2)$. 

    In each of the $d$ rounds, \emph{Spoiler} first picks an $s \in \{1, 2\}$.
    The players then repeat the following steps \(n\) times to construct tuples $\bar u_1 \in \mn{adom}(D_1)^n$ and $\bar u_2
    \in \mn{adom}(D_{2})^n$.
    \begin{itemize}
      \item \emph{Spoiler} picks distinct values $w_s^{(1)},\ldots,w_s^{(k)}\in \mn{adom}(D_s)$ where \(k\) is at most \(c\).
      \item \emph{Duplicator} responds with distinct values $w_{3-s}^{(1)},\ldots,w_{3-s}^{(k)}\in \mn{adom}(D_{3 - s})$.
            If there are less than \(k\) values in \(D_{3-s}\), \emph{Spoiler} wins immediately.
      \item \emph{Spoiler} picks an index $j$ and appends \(w_1^{(j)}\) and \(w_2^{(j)}\) to \(\bar u_1\) and \(\bar u_2\) respectively.
    \end{itemize}

        After the tuples are completed, it is checked whether there exists a homomorphism \(h\) from \(D_s|\{v_s,u_s^{(1)},\ldots,u_s^{(n)}\}\) to \(D_{3-s}|\{v_{3-s},u_{3-s}^{(1)},\ldots,u_{3-s}^{(n)}\}\)
    that satisfies \(h(v_s) = v_{3-s}\) and \({h(u_s^{(i)}) = u_{3-s}^{(i)}}\). 
    If this is not the case, Spoiler wins immediately.
    Otherwise,  \emph{Spoiler}  picks an index $i$. The pair $(u_1^{(i)}, u_2^{(i)})$ becomes the position for the next round.

    \emph{Duplicator} wins if \emph{Spoiler} does not win after $\ell$ rounds.
\end{definition}

\begin{lemma}\label{lem:winning strategy implies indistinguishable}
    For all $d \geq 1$ and $n, c \geq 0$ and all pointed databases $D_1^{v_1}$,
    $D_2^{v_2}$ the following holds:

    If \emph{Duplicator} has a winning strategy in the $d$-round GHML
    game of width $n$ with counting bound $c$ on $D_1$ and $D_2$ from initial
    position $(v_1, v_2)$, then for all GHML formulas $\varphi(x)$ of depth
    at most $\ell$, width at most \(n\)
    and counting bound at most $c$,
    \[
        D_1 \models \varphi(v_1)\quad \text{ if and only if }\quad D_2 \models \varphi(v_2).
    \]
\end{lemma}

\begin{proof}
We show the lemma by induction on $d$. In the induction start $d = 1$ and we have to show this result for formulas \(\exists^{\geq k_1} y_1 \cdots \exists^{\geq k_m} y_m \, \bigvee_{1 \leq i \leq \ell} \psi_i(x,y_1,\dots,y_m)\),
where each \(\psi_i\) is  a conjunction of relational atoms.
Since this is only a special case of the induction step, where the \(\psi_i\) can also contain unary formulas,
we refer to the induction step.

In the induction step, let $d > 1$ and assume that the statement holds for
all $d'$ with $d' < d$.
Further, assume that \emph{Duplicator} has a winning strategy in the
$d$-round game of width $n$ with counting bound $c$ from the position $(v_1, v_2)$.
Let $\varphi(x)$ be a GHML formula of depth $d$, width at most $n$
and counting bound $c$.
To show the statement, it suffices to show that for all subformulas of $\varphi$ of the form
$\exists^{\geq k_1} y_1 \cdots \exists^{\geq k_m} y_m \, \bigvee_{1 \leq i \leq \ell} \psi_i(x,y_1,\dots,y_m)$ with \(m\leq n\) and \(k_i\leq c\)
it holds that
\begin{align*}
& D_1 \models \exists^{\geq k_1} y_1 \cdots \exists^{\geq k_m} y_m \, \bigvee_{1 \leq i \leq \ell} \psi_i(v_1,y_1,\dots,y_m) \\
& \qquad \text{if and only if } \\
& D_2 \models \exists^{\geq k_1} y_1 \cdots \exists^{\geq k_m} y_m \, \bigvee_{1 \leq i \leq \ell} \psi_i(v_2,y_1,\dots,y_m).
\end{align*}

Now assume that there is such a subformula with
\[
D_s \models \exists^{\geq k_1} y_1 \cdots \exists^{\geq k_m}y_m\, \bigvee_{1\leq i\leq\ell} \psi_i(v_1, y_1, \ldots, y_m).
\]
If \emph{Spoiler} plays in $D_s$, \emph{Duplicator} must be able to respond in
$D_{3-s}$ appropriately.
In particular, this means that when \emph{Spoiler} picks $k_1$ distinct values
$w_s^{(1)},\ldots,w_s^{(k_1)}$ that satisfy $\exists^{\geq k_2} y_2 \cdots \exists^{\geq k_m} y_m
\,\bigvee_{1\leq i\leq \ell}\psi_i(v_s, w_s^{(i)}, y_2, \ldots, y_m)$, then \emph{Duplicator} must be able to respond
with $k_1$ distinct values $w_{3-s}^{(1)},\ldots,w_{3-s}^{(k_1)}$ of which \emph{Spoiler} can pick any to continue.
After $m$ steps, \emph{Spoiler} can thus arrive at tuples $\bar u_1$ and $\bar u_2$
such that\
\[
  D_s \models \bigvee_{1\leq i\leq \ell}\psi_i(v_s, \bar u_s).
\]
Hence, there exists an \(1\leq i\leq\ell\)
such that
\[D_s\models \psi_i(v_s,\bar u_s).\]
As \emph{Duplicator} has a winning strategy for the $d$-round game, it must
be the case that (i) \(h\) with \(h(v_s) = v_{3-s}\) and \(h(u_s^{(i)}) = u_{3-s}^{(i)}\) is a homomorphism from \(D_s|\{v_s,u_s^{(1)},\ldots,u_s^{(m)}\}\) to \(D_{3-s}|\{v_{3-s},u_{3-s}^{(1)},\ldots,u_{3-s}^{(n)}\}\), and
(ii) that for any $i$ that \emph{Spoiler} picks, \emph{Duplicator} has a
winning strategy for the $(d - 1)$-round game in the pair \((u_1^{(i)}, u_2^{(i)})\).
Point~(i) implies that \(v_{3-s}\) and \(\bar{u}_{3-s}\) satisfy all relational atoms in $\psi_i$.
Point~(ii) implies, via the induction hypothesis, that \(v_{3-s}\) and \(\bar{u}_{3-s}\) satisfy
all unary GHML formulas in $\psi_i$. Together it
follows that
\[
D_{3-s} \models \psi_i(v_{3-s}, \bar u_{3-s}) \quad \text{ and }\quad D_{3-s} \models \bigvee_{1\leq i\leq \ell} \psi_i(v_{3-s},\bar u_{3-s}).\]
As \emph{Duplicator} must be able to respond to all choices of \emph{Spoiler},
it also follows that
\[
  D_{3-s} \models \exists^{\geq k_1} y_1 \cdots \exists^{\geq k_m}y_m\, \bigvee_{1\leq i\leq \ell}\psi_i(v_
  {3-s}, y_1, \ldots, y_m).
\qedhere
\]
\end{proof}

Using GHML games, we now proceed to show that local transitivity is not
expressible in GHML.

\begin{lemma}\label{lem:winning strategy exists}\label{prop:nolocaltransinGHMLstar}
    For all $d, n, c \geq 0$, there exist databases $D_1^{v_1}$, $D_2^{v_2}$
    over the graph schema such that $D_1^{v_1}$ is locally transitive,
    $D_2^{v_2}$ is not locally transitive and
    \emph{Duplicator} has a winning strategy in the $d$-round GHML game
    of width $n$ and counting bound $c$ starting in position $(v_1, v_2)$.
\end{lemma}
\begin{proof}
  To show the lemma, we use the following families $F_1^m$, $F_2^m$ of databases
  over the graph schema.
For $m \geq 0$, set
\[
  F_1^m = \{ E(v_a, v_b) \mid  1 \leq a < b \leq 2m + 3 \}
\]
and $F_2^m = F_1^m \setminus \{E(m+1, m + 3) \}$. Clearly, \(v_{m+1}\) in $F_1^m$ is locally
transitive and \(v_{m+1}\) in $F_2^{m}$ is not locally transitive.
The structure of $F_1^m$ and $F_2^m$ is visualized in Figure~\ref{fig:local-transitive}.
We define the function \(\ind(v_i)= i\).

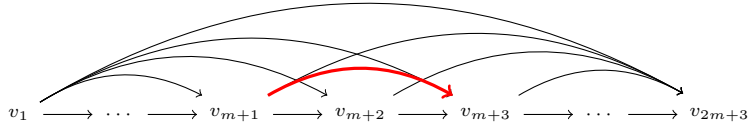
\begin{figure}
  \centering
    \begin{tikzpicture}[node distance = 0.65cm, font = \scriptsize]
        \node[circle] (a1) {$v_1$};
        \node[circle, right = of a1] (a3) {$\cdots$};
        \node[circle, right = of a3] (a4) {$v_{m+1}$};
        \node[circle, right = of a4] (a5) {$v_{m + 2}$};
        \node[circle, right = of a5] (a6) {$v_{m + 3}$};
        \node[circle, right = of a6] (a7) {$\cdots$};
        \node[circle, right = of a7] (a9) {$v_{2m + 3}$};
        \draw[->] (a1) to (a3);
        \draw[->, bend left] (a1) to (a4);
        \draw[->, bend left] (a1) to (a5);
        \draw[->, bend left] (a1) to (a6);
        \draw[->, bend left] (a1) to (a9);
        \draw[->] (a3) to (a4);
        \draw[->] (a4) to (a5);
        \draw[->, bend left] (a4) to (a9);
        \draw[->] (a5) to (a6);
        \draw[->, bend left] (a5) to (a9);
        \draw[->] (a6) to (a7);
        \draw[->, bend left] (a6) to (a9);
        \draw[->] (a7) to (a9);
        \draw[->, bend left, very thick, red] (a4) to (a6);
    \end{tikzpicture}
    \caption{Structure of the databases $F_1^m$ and $F_2^m$ with $F_1^m \setminus F_2^m$ highlighted in red.}
    \label{fig:local-transitive}
  \end{figure}

We now show the following claim which implies the lemma.

\noindent
\textbf{Claim.}
For all $d, n, c\geq 0$,
if $m \geq d (c\, 2^{n + 1}-c+1)$, then \emph{Duplicator} has a winning
strategy in the $d$-round GHML game of width $n$ and counting bound $c$ on $F_1^m$ and $F_2^m$
for any starting position $(x_1,x_2)$ such that $\ind(x_1) = \ind(x_2)$ or
\[d (c\, 2^{n +1} -c + 1)+ 1 \leq \ind(x_1),\ind(x_2)\leq 2m + 3 - d (c\,2^{n+1 }-c+1)\tag{\(*_d\)}\]

We start by describing the idea behind the constant \(c\, 2^{n+1}-c+1\).
Assume that two integers \(x,y \in \mathbb{N}\) satisfy \(y-x \geq c\, 2^{n+1}-c +1\).
Then there exist \(c\) integers \({z_1 = x + c\,2^{n}-c + 1},\ldots, {z_c = x+c\,2^{n}}\) that satisfy
\[z_i - x \geq c\,2^n-c+1\] and \[y-z_i \geq c\, 2^{n+1}-c + 1 + x - z_i \geq c\,2^n-c+1.\]

\noindent
\textbf{Proof of the claim.} 
We prove that \emph{Duplicator} has a winning strategy where they never choose \(v_{m+1}\) and \(v_{m+3}\) in \(F_{3-s}^m\), with the exception that they may choose the starting position.
Within one round, they have to ensure that, at any time, there is enough space between chosen vertices in \(F_{3-s}^m\) and that
they respond with vertices that have the same index if \emph{Spoiler} chooses vertices near one of the ends.
We show that the Condition (\(*_d\)) gives \emph{Duplicator} enough freedom, without allowing \emph{Spoiler} to win.

For readability, set $N_k = c\,2^{n + 1 - k} -c +1 $ and $M = 2m + 3$.
We show the claim by induction on $d$.
In the induction start, $d = 0$ and the statement holds trivially.

Now assume that the claim holds for $d - 1$, and that an $d$-round 
GHML game of width $n$ and counting bound $c$
starts in position $(x_1,x_2)$ for some $x_1,x_2$ that satisfy $\ind(x_1)=\ind(x_2)$ or ($*_d$).

Let \(s\in\{1,2\}\) be the choice of \emph{Spoiler}, that is, they play in \(F_s^m\) and
\emph{Duplicator} responds with vertices in \(F_{3-s}^m\).

\emph{Duplicator} has to deal with the following problems:
If \(\ind(x_1)<\ind(x_2)\), then there are fewer vertices right of \(x_2\) than there are vertices right of \(x_1\),
and additionally, \emph{Duplicator} may have to shift the indices of their response to avoid choosing \(v_{m+1}\) or \(v_{m+3}\).

For this, let \(p_1 =(d-1)N_0+1\) and \(p_2 =M-(d-1)N_0\).
\emph{Duplicator} divides the graphs into five sections. We write them as intervals of indices.
\begin{itemize}
    \item Two sections at the ends of the graphs defined by \({A = [1,p_1-1]\cup[p_2+1,M]}\).
    \item Two sections with the intervals \(B_1 =[p_1, p_1+N_0-1]\) and \(B_2 =[p_2-N_0 +1, p_2]\).
    \item The center defined by the interval \({C =[p_1+N_0 , p_2-N_0]}\).
\end{itemize}
\emph{Duplicator} will play as follows:
If \emph{Spoiler} chooses a vertex in \(A\), \emph{Duplicator} will respond with the same vertex in \(F_{3-s}\).
Otherwise, \emph{Duplicator} will choose a vertex in \(B_1\cup B_2\) which avoids choosing \(v_{m+1}\) and \(v_{m+3}\) which are in section \(C\).
\emph{Duplicator} may choose \(x_{3-s}\) even if it is in \(C\).
The starting positions satisfy \(x_1\in C\) iff \(x_2\in C\), and if \(x_1,x_2 \in A\cup B_1\cup B_2\), then by assumption \(\ind(x_1)=\ind(x_2)\).
This strategy ensures
that each selected pair \((u_1,u_2)\) will satisfy \(\ind(u_1)=\ind(u_2)\) or (\(*_{d-1}\)),
since the bounds of these zones are equal to the bounds of (\(*_{d-1}\)).
By induction hypothesis, \emph{Duplicator} then has a winning strategy in the \(d-1\)-round GHML game with width \(n\) and counting bound \(c\).
%

Thus, we only have to show that \emph{Duplicator} can choose vertices according to this strategy without loosing this round.
Within this round, both players will construct tuples \(\bar u_1\) and \(\bar u_2\) with \(n\) elements.
We show by induction on \(n\) that \emph{Duplicator} can choose vertices such that these tuples satisfy the following conditions.
\begin{enumerate}
\item \emph{Duplicator} responds with vertices according to the sections above.
They don't choose vertices in \(C\) except \(x_{3-s}\).
\item Let \((u_1,u_2), (u_1',u_2')\) be two played pairs which are in the interval \([p_1,p_2]\). Then
        \[\ind(u_s) \sim \ind(u_s') \text{ iff } \ind(u_{3-s})\sim \ind(u_{3-s}') \text{ for all }\sim\,\in\{<,=,>\}.\]
\item 
If \(x_1,x_2\not\in B_i\) for some \(i\in\{1,2\}\):
After selecting \(k\) vertices, let \((u_1,u_2), (u_1',u_2')\) be two played pairs such that \(u_{3-s},u_{3-s}'\in B_i\).
Then they satisfy:
\begin{itemize}
  \item if \(\ind(u_{3-s})\neq \ind(u_{3-s}')\) then \(|\ind(u_{3-s})-\ind(u_{3-s}')| \geq N_k\), and
  \item  \(|\ind(u_{3-s})-p_i|\geq N_k\) and \(|p_i+N_0-1 - \ind(u_{3-s})| \geq N_k\).
\end{itemize}
\item If \(x_1,x_2\in B_i\) for some \(i\in\{1,2\}\):
Let \((u_s,u_{3-s})\) be a played pair such that \(\ind(u_s)\in B_i\). Then \(\ind(u_s) =\ind(u_{3-s})\).
\end{enumerate}
These conditions are obviously satisfied for \(k=0\).

Now let \(\bar u_1\) and \(\bar u_2\) be tuples with \(k\) elements for some \(0\leq k < n\).
The game proceeds with \emph{Spoiler} choosing vertices $w_s^{(1)}, \ldots, w_s^{(c')}$, where \(c'\) is at most \(c\).
\emph{Duplicator} now has to choose \(c'\) distinct elements \(w_{3-s}^{(1)},\ldots,w_{3-s}^{(c')}\) such that \(\bar u_1, w_1^{(i)}, \bar u_2\), and \(w_2^{(i)}\) satisfy the conditions above for each \(i\).

Fix some \(w_s^{(i)}\).
If \(w_s^{(i)}\in A\) or \(x_1,x_2,w_s^{(i)}\in B_j\) then \emph{Duplicator} chooses the vertex with index \(\ind(w_s^{(i)})\) in \(F_{3-s}^m\).
If \(w_s^{(i)} =x_s\) then \emph{Duplicator} chooses \(w_{3-s}^{(i)} =x_{3-s}\).
We next clarify how \emph{Duplicator} chooses their response in the other cases, that is \(w_s^{(i)}\in C\) or \(w_s^{(i)}\in B_j\) and \(x_1,x_2\not\in B_j\).
\emph{Duplicator} then chooses a \(j\in \{1,2\}\) such that their response must be in \(B_j\) to satisfy Condition~2.
If \(x_1,x_2\in B_{3-j}\) then \emph{Duplicator} can choose \(B_j\). 
They now find the interval \([e_1,e_2]\) within \(B_j\) in which their response has to be to satisfy Condition~2,
such that both, \(e_1\) and \(e_2\) are the index of some played vertex or one of the limits of \(B_j\).
If \(e_1 = e_2\), then \(w_s^{(i)}\) was chosen by \emph{Spoiler} before and is thus in \(\bar u_s\) and \emph{Duplicator} can respond with the corresponding vertex in \(\bar u_{3-s}\).
Otherwise, observe that by the Condition~3, \(e_2-e_1 \geq N_k\).
Therefore, there exist at least \(c\) vertices \(v_d\) such that \(d\) satisfies \(e_2-d \geq N_{k+1}\) and \(d-e_1 \geq N_{k+1}\).
\emph{Duplicator} then chooses one of these elements.
Because there are always \(c\) of these elements, \emph{Duplicator} can always choose a fresh vertex for each \(w_s^{(i)}\).
After that, \emph{Spoiler} chooses one \(w_s^{(i)}\) and the corresponding response and appends these vertices to \(\bar u_1\) and \(\bar u_2\), which now contain \(k+1\) vertices.
It is easy to see that the four conditions above are still satisfied.

This strategy ensures that in the final tuple $\bar u_{3-s}$ the order of the values
is the same as in the tuple $\bar u_s$.
And additionally by choice of the intervals,
it is never the case that \emph{Duplicator} plays $v_{m + 1}$ and \(v_{m+3}\) at the same time.
Thus, \(h\) with \(h(x_s)=x_{3-s}\) and \(h(u_s^{(i)}) = u_{3-s}^{(i)}\) is a homomorphism from $F_s^m|\{v_s, u_s^{(1)},\ldots, u_s^{(n)}\}$ to $F_{3-s}^m|\{v_{3-s},u_{3-s}^{(1)},\ldots, u_{3-s}^{(n)}\}$, as required.
Hence, \emph{Duplicator} does not lose this round.
As described before, by induction hypothesis, \emph{Duplicator} thus has a winning strategy in the \(d\)-round GHML game with width \(n\) and counting bound \(c\).
\end{proof}

Lemma~\ref{lem:winning strategy exists} and Lemma~\ref{lem:winning strategy implies indistinguishable}
now together imply that no GHML formula can
express local transitivity. 
This completes the proof of Proposition~\ref{thm:sum-DHN-not-in-UNFOC-uniform}.

\subsection{Max and Sum Aggregation: Static Analysis}
We now use the logical equivalences to find upper and lower bounds for
the complexity of the emptiness and subsumption problems for DHNs.
We show each point of Theorem~\ref{thm:staticanalysiscombi}.

For \(\maxagg\)-DHNs, we can use known complexity results for UNFO.
\begin{theorem}
  The emptiness problem and the subsumption problem for $\maxagg$-DHNs are decidable.
\end{theorem}
\begin{proof}
  Decidability of the emptiness problem is an immediate consequence of the facts that (i)~the translations in the proof of Point~1 of Theorem~\ref{thm:combi} 
are effective, (ii)~the satisfiability problem for UNFO is decidable, and (iii)~UNFO has the finite model property \cite{DBLP:journals/corr/SegoufinC13}.
These facts also imply decidability of the subsumption problem. Given DHNs $\calN_1$ and $\calN_2$, we construct
equivalent UNFO sentences $\varphi_1$
and $\varphi_2$ and decide whether
$\varphi_1 \wedge \neg \varphi_2$
is unsatisfiable (equivalently: 
finitely unsatisfiable).
\end{proof}
On databases of bounded degree and for connected DHNs, decidability follows from a counting argument.
Lower bounds for connected \(\sumagg\)-DHNs can again be obtained from complexity results for the corresponding logics.
\begin{theorem}
\label{cor:sumstaticdecidableconnectedbounded}
  Let $B \geq 4$.
  On databases of degree at most $B$, the emptiness problem and the subsumption problem for  connected $\sumagg$-DHNs are decidable and {\sc coNExpTime}-complete. The lower bounds already hold on vertex-labeled directed graphs.
\end{theorem}
\begin{proof}[Proof (sketch).]

    We first describe an algorithm that decides the non-emptiness problem of \sumagg-DHNs in {\sc NExpTime}.
    This algorithm receives as input
    a \sumagg-DHN \(\calN\) over some schema $\Sbf$. Whether we consider the schema \(\Sbf\) as fixed or part of the input makes no difference for the subsequent proof.
        Let \(\ell\) be the number of layers of \(\calN\) and let \(k\) be the size of the largest database \(F^\bullet\) that occurs in some homomorphism query of \(\calN\).
    For each database \(D^v\), the output of \(\calN\) on \(D^v\) only depends on the \(\ell k\) neighborhood of \(v\).
    Thus, there exists a database that \(\calN\) accepts iff there exists a database of diameter at most \(\ell k\) that \(\calN\) accepts.
    Such databases have at most \(B^{\ell k}\) values and \(B^{\ell k r}\) facts, where $r$ is the sum of the arities of the relation symbols in \(\Sbf\).
    Hence, these databases have a size at most single exponential in \(\ell, k, r\) which are all bounded from above by the input length.
    Thus, a {\sc NExpTime} algorithm can guess a database of this size and then check in polynomial time (in the size of the database) whether it is accepted by \(\calN\).


%
To show that the non-emptiness problem of connected \sumagg-DHNs is {\sc NExpTime}-hard, we use the translation of GHML$^-$ formulas to \sumagg-DHNs provided in the proof of Lemma~\ref{thm:UNFOC-to-sum-DHN-uniform}, which is polynomial. 
We thus have to show that finite satisfiability in connected GHML$^-$
is {\sc NExpTime}-hard.
To start with, note that connected GHML$^-$ contains modal logic (ML) as
a proper fragment and (finite) satisfiability of uni-modal ML formulas is {\sc PSpace}-hard. The standard proof by reduction from quantified Boolean formulas even yields this result on directed graphs of degree at most~3. Thus the non-emptiness problem for \sumagg-DHNs is at least
{\sc PSpace}-hard, even on vertex-labeled directed graphs of degree at most~3. However, it follows from the results in \cite{Lutz-IJCAR08} that finite satisfiability in connected HML (and thus also in connected GHML$^-$) is even {\sc NExpTime}-hard, already on vertex-labeled directed graphs of degree at most~4.
More precisely, it is proved in 
\cite{Lutz-IJCAR08} that for a given uni-modal ML formula $\varphi$ 
and unary connected conjunctive query~$q(x)$,
it is {\sc coNExpTime}-hard to decide whether $\varphi$ implies $q$ in the sense that for all (potentially infinite) pointed vertex-labeled directed graphs $(G,v)$, $(G,v) \models \varphi$ implies $(G,v) \models q(x)$.\footnote{Note that \cite{Lutz-IJCAR08} uses a different language. The stated result is Theorem~1 of that article.}
It is implicit in the proofs of \cite{Lutz-IJCAR08} that this problem coincides with the restricted version where only graphs $G$ are considered that are finite and of degree at most~4. Now, it is easy to see that $\varphi$ (finitely) implies $q$ if and only if
the HML formula $\varphi \wedge \neg q$ is not finitely satisfiable. Note that here we view the unary CQ $q$ as a HML formula of the form $\exists \bar y \, \psi(x, \bar y)$.
Clearly,  {\sc NExpTime}-hardness of finite satisfiability in connected HML follows. 






The subsumption problem can be solved as follows: \(\calN_1\) does not subsume \(\calN_2\) iff there exists a database that is accepted by \(\calN_2\) but not by \(\calN_1\).
The size of such a database can be bounded as in the non-emptiness case.
Thus, the subsumption problem can be solved in {\sc coNExpTime}.
Subsumption for GHML$^-$ is {\sc coNExpTime}-hard since a GHML$^-$ formula \(\varphi\) is satisfiable iff \(\bot\) does not subsume \(\varphi\).
By using the translation in the proof of Lemma~\ref{thm:UNFOC-to-sum-DHN-uniform}, we can reduce subsumption of GHML$^-$ formulas to subsumption of \sumagg-DHNs.
Thus, the subsumption problem of \sumagg-DHNs is {\sc coNExpTime}-complete.
\end{proof}

The following is Point~3 of Theorem~\ref{thm:staticanalysiscombi}.
\begin{theorem}
    The emptiness problem and the subsumption problem for  connected   $\sumagg$-DHNs are undecidable, already for the graph schema.
    The same is true for (not necessarily connected)  $\sumagg$-DHNs on graphs of any bounded degree $B \geq 5$.
\end{theorem}
\begin{proof}
  \newcommand{\freevar}{\alpha}
  We start with  $\sumagg$-DHNs that are not necessarily connected and graphs of bounded or unbounded degree and later sketch the modification necessary for connected  $\sumagg$-DHNs on graphs of unbounded degree.
 By Lemma~\ref{thm:UNFOC-to-sum-DHN-uniform}, it suffices to prove
  that the finite satisfiability of GHML$^-$ formulas is undecidable, both on unrestricted graphs and on graphs of bounded degree. 
 We remark that finite satisfiability in UNFOC is known to be undecidable
 on unrestricted graphs, a proof is implicit in  the proof of Theorem~7.1 of \cite{DBLP:journals/corr/SegoufinC13}. But since GHML$^-$ is less expressive than UNFOC and we are also interested in graphs of bounded
 degree, 
 we sketch a direct proof by reduction of the tiling problem that
 asks for tiling a finite rectangle using restricted sets of tiles at the borders. In addition to the edge relation $E$ of graphs, we use unary relation symbols and later argue how these can be removed. 

A \emph{tiling system} takes the form $S=(T,L,R,U,D,H,V)$
where $T$ is a finite set of tile types, $L,R,U,D \subseteq T$
are sets of tiles to be used at the left, right, upper, and lower borders, 
and $H,V \subseteq T^2$ are the
 horizontal and vertical matching conditions. A \emph{solution} for $S$ is a triple $(n,m,f)$ with $n,m \in \mathbb{N}$ and  $f:\{1,\dots,n\} \times \{1,\dots,m\} \rightarrow T$
 such that the following conditions are satisfied:
 \begin{enumerate}

\item for $1 \leq i < n$ and $1 \leq j \leq m$: $(f(i,j),f(i+1,j)) \in H$;

\item for $1 \leq i \leq n$ and $1 \leq j< m$:  $(f(i,j),f(i,j+1)) \in V$;

\item for $1 \leq j \leq m$: $f(1,j) \in L$ and $f(n,j) \in R$;

\item for $1 \leq i \leq n$: $f(i,1) \in D$ and $f(i,m) \in U$.
         
 \end{enumerate}
 It is shown in \cite{berger1966undecidability} that the problem of satisfiable tiling systems is undecidable.

 We show how to construct
 a GHML$^-$ formula $\varphi_S$ such that 
 \begin{enumerate}

 \item[(a)] if there is a solution for $S$, then $\varphi_S$ is finitely satisfiable in a graph of degree at most~4 and 
 \item[(b)] if $\varphi_S$ is finitely satisfiable, then there is a solution for $S$.

 \end{enumerate}
 We simulate two binary relations
 that form a grid, a horizontal and a vertical one. Each edge of these relations is simulated via two consecutive edges and a propositional letter $H$ or $V$
 satisfied at the middle point. 
 We also introduce a propositional
 letter $P_T$ for every  tile type $t \in T$, a letter $G$ to identify vertices on the grid (as opposed to middle vertices on simulated edges), and
 letters $B^\leftarrow, B^\rightarrow, B^\downarrow, B^\uparrow$ to identify the borders.
 We set \[\varphi_S = \exists \freevar \, (B^\leftarrow(\freevar) \wedge B^\downarrow(\freevar) \wedge G(\freevar)) \wedge \forall x \, \vartheta(x)\]
 where $\vartheta$ is a conjunction that consists of the following conjuncts:
\begin{enumerate}

  \item\label{enum:tiling-border-no-neighbor} The vertices at the borders are exactly the vertices that have no corresponding neighbor
        \begin{align*}
          &(G(x)\land \neg B^{\rightarrow}(x)) \leftrightarrow \exists y_1y_2 \,(E(x,y_1)\land H(y_1)\land E(y_1,y_2)\land G(y_2))\\
          \land\,&(G(x)\land \neg B^{\leftarrow}(x)) \leftrightarrow \exists y_1y_2 \,(E(y_1,x)\land H(y_1)\land E(y_2,y_1)\land G(y_2))\\
          \land\,&(G(x)\land \neg B^{\uparrow}(x)) \leftrightarrow \exists y_1y_2 \,(E(x,y_1)\land V(y_1)\land E(y_1,y_2)\land G(y_2))\\
          \land\,&(G(x)\land \neg B^{\downarrow}(x)) \leftrightarrow \exists y_1y_2 \,(E(y_1,x)\land V(y_1)\land E(y_2,y_1)\land G(y_2))\\
        \end{align*}

    \item\label{enum:tiling-create-cell} Every grid vertex,  that is not on the  upper or right border is the origin of a grid cell:
    
    \begin{alignat*}{2}
        (G(x) \wedge \neg B^\uparrow(x) \wedge \neg B^\rightarrow(x))\qquad\\[1mm] \rightarrow \exists y_1y_2y_3z_1z_2z_3u\, 
      \big (&
        E(x,y_1)\wedge V(y_1) \wedge E(y_1,y_2) \wedge G(y_2)& \\[1mm]
        &\land\, E(x,z_1) \wedge  H(z_1) \wedge E(z_1,z_2) \wedge G(z_2)& \\[1mm]
         &\land\, E(y_2,y_3) \wedge  H(y_3) \wedge E(y_3,u)& \\[1mm]
        &\land\, E(z_2,z_3) \wedge  V(z_3) \wedge E(z_3,u) \wedge G(u)
      &\big ) 
    \end{alignat*}
    
    \item\label{enum:tiling-functional-neighbors} The horizontal and vertical grid relations are functional:
    \begin{align*}      
      G(x) \rightarrow& \big (\neg \exists^{\geq 2} y\, (E(x,y) \wedge V(y)) \wedge \neg \exists^{\geq 2} y\,( E(y,x) \wedge V(y)) \\
      &\phantom{\big(} \land \neg \exists^{\geq 2} y\, (E(x,y) \wedge H(y))\wedge \neg \exists^{\geq 2} y \,(E(y,x) \wedge H(y)) \big) \\[1mm]
        \land\, (V(x) \vee H(x)) \rightarrow& \big(\neg \exists^{\geq 2} y \,E(x,y) \land  \neg\exists^{\geq 2} y\, E(y,x)\big)
      \end{align*}

    \item\label{enum:tiling-consistent-borders} The borders are consistent:
        \begin{align*}
          \forall y_1y_2\, \Big(&\big((E(x,y_1)\lor E(y_1,x ))\land (E(y_1,y_2)\lor E(y_2,y_1))\land V(y_1)\big)\\
          &\rightarrow \big((B^\downarrow(x)\leftrightarrow B^\downarrow(y_2))\land (B^\uparrow(x)\leftrightarrow B^\uparrow(y_2))\big)\Big)\\
          \land\, \forall y_1y_2\, \Big(&\big((E(x,y_1)\lor E(y_1,x ))\land (E(y_1,y_2)\lor E(y_2,y_1))\land H(y_1)\big)\\
          &\rightarrow \big((B^\rightarrow(x)\leftrightarrow B^\rightarrow(y_2))\land (B^\leftarrow(x)\leftrightarrow B^\leftarrow(y_2))\big)\Big)\\
        \end{align*}    
    \item\label{enum:tiling-unique} Every grid vertex is uniquely tiled:
    \[
       G(x) \rightarrow \bigvee_{t \in T} \Big ( P_t(x) \wedge 
       \bigwedge_{t' \in T \setminus \{ t \}} \neg P_{t'}(x) \Big )
    \]
    
    \item\label{enum:tiling-matching} The matching conditions are satisfied:
    %
    \begin{align*}
        (G(x) &\wedge \neg B^\uparrow(x))\\
        &\rightarrow \bigvee_{(t,t') \in V}
      \Big( P_t(x) \wedge \exists y_1\, \big ( E(x,y_1) \wedge V(y_1) \wedge \exists y_2\, ( E(y_1,y_2) \wedge P_{t'}(y_2) \big ) \Big )\\[2mm]
      \wedge\ 
         (G(x) &\wedge \neg B^\rightarrow(x)) \\
         &\rightarrow \bigvee_{(t,t') \in H}
      \Big( P_t(x) \wedge \exists y_1 \,\big ( E(x,y_1) \wedge H(y_1) \wedge \exists y_2\, ( E(y_1,y_2) \wedge P_{t'}(y_2) \big ) \Big ) 
    \end{align*}

    \item\label{enum:tiling-border-matching} The borders are tiled as required:
    \begin{align*}
      &(B^\leftarrow(x) \rightarrow \bigvee_{t \in L} P_t(x)) \ \wedge \
        (B^\rightarrow(x) \rightarrow \bigvee_{t \in R} P_t(x)) \\ 
        \wedge \ &
          (B^\downarrow(x) \rightarrow \bigvee_{t \in D} P_t(x))  \ \wedge \ 
            (B^\uparrow(x) \rightarrow \bigvee_{t \in U} P_t(x)). 
    \end{align*}
    
\end{enumerate}
It is routine to verify that $\varphi_{S}$ is a GHML$^-$ formula and that it satisfies Condition~(a).  
We argue that it satisfies Condition~(b) as well.

Given a model \(\mathfrak{A}=(A, G^\mathfrak{A}, V^\mathfrak{A}, H^\mathfrak{A}, {B^\leftarrow}^\mathfrak{A}, {B^\rightarrow}^\mathfrak{A}, {B^\uparrow}^\mathfrak{A}, {B^\downarrow}^\mathfrak{A}, E^\mathfrak{A})\) we first construct an intermediate partial function \(g:\mathbb{N}\times \mathbb{N}\to G^\mathfrak{A}\) such that all defined \(g(i,j),g(i+1,j)\) are connected with a common vertex \(v\) that satisfies \(H(v)\), all defined \(g(i,j), g(i,j+1)\) are connected with a common vertex \(v\) that satisfies \(V(v)\), and such that borders are consistent.
Given this function \(g\) we can then construct a solution \(f\) for \(S\).

We construct \(g\) inductively, starting with the lower border, and then defining more rows moving upwards.
The lower border is also defined inductively, and we first set \(g(1,1)\) to a value that satisfies \(B^\rightarrow(\freevar)\land B^\downarrow(\freevar)\land G(\freevar)\).

If \(g(1,1),\ldots,g(k,1)\) are defined, there are now two cases to consider
\begin{itemize}
  \item If \(g(k,1)\in {B^\rightarrow}^{\mathfrak{A}}\), then the lower border is finished.
  \item If \(g(k,1)\not\in {B^\rightarrow}^{\mathfrak{A}}\), then by Formula \ref{enum:tiling-border-no-neighbor}, there exists an element \(a\in G^{\mathfrak{A}}\) such that \(g(k,1)\) has a \(H^\mathfrak{A}\) successor, of which \(a\) is a successor.
  We set \(g(k+1,1)\) to \(a\).
        These formulas also require \(a\in G^\mathfrak{A}\).
        The Formulas \ref{enum:tiling-border-no-neighbor} and \ref{enum:tiling-consistent-borders} require that \(a\not\in {B^\leftarrow}^\mathfrak{A}\) and \(a\in {B^\downarrow}^\mathfrak{A}\).
        Thus, the borders stay consistent.
\end{itemize}
The first case has to be satisfied after finitely many steps. Assume the contrary, then after at most \(|A|+1\) iterations there exist \(i < j\) such that \(g(i,1) = g(j,1)\).
W.l.o.g. this is the only pair of indices that satisfy this equality.
Since \(g(j,1)\) has an \(H\)-predecessor, \(g(i,1)\not\in {B^\leftarrow}^{\mathfrak{A}}\) must hold, thus \(i\neq 1\) and both \(g(i-1,1)\) and \(g(j-1,1)\) are defined and they satisfy \(g(i-1,1)\neq g(j-1,1)\).
Therefore, \(g(j,1)\) has two \(H\)-predecessors or its \(H\)-predecessor has two \(G\) predecessors, which contradicts Formula \ref{enum:tiling-functional-neighbors}.
Hence, the first case is satisfied after finitely many steps and the first row is well-defined.

For the next rows, assume \(g(i,j)\) is defined for all \(1\leq i\leq n\) and \(1\leq j\leq m\).
Then there are again two cases:
\begin{itemize}
  \item If there exists an \(i\) such that \(g(i, m)\in {B^\uparrow}^\mathfrak{A}\), then by Formula \ref{enum:tiling-consistent-borders}, all \(g(j,m)\) are in \({B^\uparrow}^\mathfrak{A}\).
        We can thus stop, and proceed with the definition of the solution \(f\).
        This case has to be satisfied eventually, which can be shown by using Formula \ref{enum:tiling-functional-neighbors} with a similar argument as before.
  \item Otherwise, \(g(1,m)\) has a \(V\) successor, which has a \(G\)-successor which we set as \(g(1,m+1)\). It satisfies \(g(1,m+1)\in {B^\leftarrow}^\mathfrak{A}\), \(g(1,m+1)\not\in {B^\downarrow}^\mathfrak{A}\) and \(g(1,m+1)\in {B^\rightarrow}^\mathfrak{A}\iff g(1,m)\in {B^\rightarrow}^\mathfrak{A}\) by Formula \ref{enum:tiling-consistent-borders}.

        For \(1\leq i< n\), let \(g(i+1,m+1)\) be the vertex \(v\) such that \(g(i,m+1)\) has an \(H\) successor that has \(v\) as a successor and such that \(g(i+1,m)\) has a \(V\) successor that has \(v\) as a successor.
        This vertex exists because of Formula \ref{enum:tiling-create-cell}.
        By Formulas \ref{enum:tiling-border-no-neighbor} and \ref{enum:tiling-consistent-borders}, \(v\not \in {B^\leftarrow}^\mathfrak{A}\), \(v\not\in {B^\downarrow}^\mathfrak{A}\) and \(v\in {B^\rightarrow}^\mathfrak{A}\) iff \(i+1 = n\).
\end{itemize}

Using the function \(g\), it is easy to construct a solution \(f\) for \(S\).
Formula \ref{enum:tiling-unique} can be used to define \(f\), while formulas \ref{enum:tiling-matching} and \ref{enum:tiling-border-matching} ensure that \(f\) is a valid solution.

To get rid of the unary atoms, we replace \(P_k(x)\) with the formula
\begin{alignat*}{1}
  \exists yy_1\cdots y_\ell z \,\big(E(y_k,y_k)\land E(x,y)\land E(y,y)\land E(y,y_1)\land \bigwedge E(y_i,y_{i+1})\land E(y_\ell, z)&\\
  \land \neg\exists z'\, E(z,z') \land \bigwedge E(y_i,y_i)\rightarrow \exists^{\geq 2}y_i'\, E(y_i,y_i')&\big),
\end{alignat*} 
where \(\ell\) is the number of unary relation symbols.

Since \(\varphi_S\) does not require any self-loops, they can be used safely to indicate unary facts.
The variable \(y\) is used to not confuse the facts of neighboring vertices, and \(z\) can be used to construct models in which all \(y_1,\ldots,y_\ell\) are forced to bind to different vertices.
It is easy to adapt the proof above to this setting, although the degree is now bounded by \(5\) instead of \(4\).

For the connected but unbounded case, we use the correspondence between connected GHML$^-$ and connected \sumagg-DHNs.
We will adapt the above formula such that it can be translated into a connected GHML$^-$ formula.
For this we have to introduce a free variable, which we achieve by removing the first existential quantifier in \(\varphi_S\):
\[\varphi_S'(\freevar) = B^\leftarrow(\freevar)\land B^\downarrow(\freevar)\land G(\freevar)\land \forall z\, \vartheta(z).\]
Since removing outermost existential quantifiers does not change satisfiability of first-order formulas, \(\varphi_S'\) is satisfiable iff \(\varphi_S\) is.
The universal quantifier is still unconnected to \(\freevar\), hence we use \(\freevar\) as a ''spypoint'' and introduce edges from \(\freevar\) to every grid vertex with an intermediate vertex that satisfies a new unary relation \(C\).
We thus adapt \(\varphi_S'\) as follows:
\begin{itemize}
  \item The free variable \(\freevar\) has to be connected to itself, therefore we add the conjunct\\ \(\exists y\, (E(\freevar,y)\land C(y)\land E(y,\freevar))\).
  \item We universally quantify all vertices connected to \(\freevar\): \(\forall yx\, ((E(\freevar,y)\land C(y)\land E(y,x))\rightarrow \vartheta(x))\).
        This formula is equivalent to \[
        \neg \exists yx \,(E(\freevar,y)\land C(y)\land E(y,x)\land \neg \vartheta(x)),
        \]
        and is thus still in GHML$^-$.
  \item Since \(\vartheta\) is negated in the formula above, we cannot introduce \(\alpha\) as a new free variable.
        To rediscover \(\alpha\) we first require that each grid vertex has at most one connected vertex that is in the bottom left. We add the following conjunct to \(\vartheta\):
        \[\neg \exists^{\geq 2} yz\, (E(y,x)\land C(y)\land E(z,y)\land B^{\downarrow}(z)\land B^{\leftarrow}(z)\land G(z))\]
  \item We can now adapt Formula~\ref{enum:tiling-create-cell} such that the newly introduced grid vertex also has a connection to a vertex in the bottom left.
        In particular we can require that originating and the new vertex have a common neighbor in the bottom left, which by the previous item is unique.
        We add to Formula~\ref{enum:tiling-create-cell} the following conjunct to the right side:
        \begin{align*}
            \exists c_1c_2c_3\, (&B^\leftarrow(c_1)\land B^\downarrow(c_1)\land G(c_1)\\
            &\land E(c_1,c_2)\land C(c_2)\land E(c_2,x)\land E(c_1,c_3)\land C(c_3)\land E(c_3,u)).
        \end{align*}
  \item Analogously, in Formula~\ref{enum:tiling-border-no-neighbor} we add to the right side:
        \begin{align*}
            \exists c_1c_2c_3\, (&B^\leftarrow(c_1)\land B^\downarrow(c_1)\land G(c_1)\\
            &\land E(c_1,c_2)\land C(c_2)\land E(c_2,x)\land E(c_1,c_3)\land C(c_3)\land E(c_3,y_2)).
        \end{align*}
  \end{itemize}

  It is easy to verify that this formula is satisfiable iff \(\varphi_S\) is satisfiable, and that this formula can be translated into a connected GHML$^-$ formula.
  Since connected GHML$^-$ formulas can be transformed into connected DHNs, this proves the undecidability of connected DHNs in the unbounded degree case.
%
\end{proof}
%

\subsection{Mean Aggregation}

We first give a more thorough introduction to RHML.


RHML is an extension of HML since existential quantifiers can be expressed in RHML as follows:
let \(\varphi =\exists\bar{y}\,(\psi_1(x,\bar{y})\land \psi_2(x,\bar{y}))\) be an HML formula where \(\psi_1\) is a conjunction of relational atoms and \(\psi_2\) is a conjunction of unary HML formulas.
Then \(\varphi\) is equivalent to \[\exists_{> 0}\, \bar{y} \,\big(\psi_2(x,\bar{y}), \psi_1(x,\bar{y})\big).\]
Since \(\exists_{>0}\,\bar y\,(\mu,\nu)\) corresponds to existential quantifiers of HML, we define that \(\exists_{>t}\,\bar y\,(\mu,\nu)\) is not satisfied if there are no values that satisfy \(\nu\).
Analogously, since \(\exists_{\geq 1}\,\bar y \,(\mu,\nu)\) corresponds to the negation of existential quantifiers, we define that \(\exists_{\geq t}\,\bar y \,(\mu,\nu)\) is satisfied if there are no values that satisfy \(\nu\).

As a useful abbreviation, we define 
\begin{align*}
  \exists_{=t}\,\bar{y}&\,(\mu(x,\bar{y}), \nu(x,\bar{y}))
  =\exists_{\geq t}\,\bar{y}\,(\mu(x,\bar{y}), \nu(x,\bar{y}))
                          \land \neg \exists_{>t}\,\bar{y}\,(\mu(x,\bar{y}),\nu(x,\bar{y})).
\end{align*}
Notice that \(\exists_{=t}\,\bar{y}(\,\mu,\nu)\) is also satisfied if there are no values that satisfy \(\nu\).

An RHML formula \(\varphi\) is \emph{connected} if for every subformula \(\exists_{\sim t}\,\bar{y}\,(\mu(x,\bar{y}),\nu(x,\bar{y}))\) of \(\varphi\), the database induced by \(\nu\) is connected and uses \(x\).

We now show how RHML and \(\meanagg\)-DHNs correspond to each other.
We start with translating RHML formulas into \(\meanagg\)-DHNs.
The following lemma proves Point~1 and one of the directions in Point~2 of Theorem~\ref{thm:meancombinew}.
\begin{lemma}\label{lem:ratio-UNFO-to-mean-DHN-bounded-degree}
    \phantom{}
  \begin{enumerate}
  \item Every (connected) RHML formula is equivalent to a (connected) \meanagg-DHN;
  \item let $B \geq 0$;
    on databases of degree at most $B$, every connected RHML formula is equivalent to a connected simple  \meanagg-DHN.
    \end{enumerate}
\end{lemma}

\begin{proof}
  We first describe the parts that both points have in common.
  They will only differ in their combination function for the translation of quantifiers.
  Let \(\varphi\) be an RHML formula and let \(\varphi_1,\ldots,\varphi_k\) be an enumeration of the subformulas of \(\varphi\) such that if \(\varphi_\ell\) is a subformula of \(\varphi_k\), then \(\ell\leq k\).
  In analogy with Lemma \ref{lem:UNFO-to-max-DHN-uniform}, we will construct a \meanagg-DHN \(\calN=(\mathcal{L}_1,\ldots,\mathcal{L}_k,\cls)\), with each layer of output dimension \(k\).
 
  We encode satisfaction of \(\varphi_i\) by a \(1\) in the \(i\)-th component of the feature vectors, and we store a \(0\) if \(\varphi_i\) is not satisfied.
  The DHN will evaluate the \(i\)-th subformula in the \(i\)-th layer, that is, after \(\mathcal{L}_i\) all subformulas \(\varphi_1,\ldots,\varphi_i\) are encoded correctly.
  
  We have seen in Lemma \ref{lem:UNFO-to-max-DHN-uniform} that DHNs can realize negations and disjunctions.
  For these operators, we only used homomorphism queries with one value and without any facts, thus we can use \meanagg instead of \maxagg.

  For the translation of quantifiers, let \(\varphi_\ell(y_0) = \exists_{\sim t}\, y_1\cdots y_n\,(\mu(\bar{y}), \nu(\bar{y}))\), where \({\sim\,\in \{>, \geq\}}\), and let \(\mu = \mu_0(y_0)\land\cdots\land\mu_n(y_n)\), where each \(\mu_i\) is a conjunction of unary formulas.
  
  By assumption, \(\nu\) defines a database \(G_\ell^{y_0}\), which is connected if \(\nu(\bar{y})\) is connected.
  
  The layer \(\mathcal{L}_\ell\) uses the following homomorphism queries:
  \begin{itemize}
    \item \((G_\ell^{y_0}, \xi, \meanagg)\), where \(\xi_i(\bar{x})\) returns 1 if in \(\bar{x}\) all entries that correspond to conjuncts in \(\mu_i\) are set to 1, and returns 0 otherwise.
    \item \((G_\ell^{y_0}, \xi', \meanagg)\), where \(\xi_i'(\bar{x})\) always returns 1.
    \item \((F^\bullet, \chi, \meanagg)\), where \(F^\bullet\) is the database with one value and without any facts, and where \(\chi(\bar{x}) = \bar{x}\) if \(\ell\geq 1\), and \(\chi(\bar{x})\) returns the vector with all entries set to \(0\) otherwise.
  \end{itemize}
  The combination function \(\com_\ell(x_1,x_2,\bar{x}_3)\) does the following:
  if \(x_2\) is \(0\), which only happens if there are no homomorphisms from \(G_\ell^{y_0}\) to the input database, and \(\sim\) is \(\geq\), then \(\com_\ell\) replaces the \(\ell\)-th entry in \(\bar{x}_3\) with \(1\).
  If \(x_2\) is \(0\) and \(\sim\) is \(>\), then \(\com_\ell\) replaces the \(\ell\)-th entry in \(\bar{x}_3\) with \(0\).
  Otherwise, \(\com_\ell\) replaces the \(\ell\)-th position in \(\bar{x}_3\) by \(1\) if \(x_1\sim t\) and by \(0\) otherwise.
  In all cases it does not return \(x_1\) and \(x_2\).

  In the general setting of Point~1 this function is clearly non-continuous at \((t, 1, \bar{x}_3)\) for each \(\bar{x}_3\).
  In the restricted settings of Point~2, there are only finitely many possible input values for \(\com_\ell\).
  \(x_1\) is a rational number and its denominator is bounded by the number of homomorphisms, which itself is bounded because in this case  \(G_\ell^{y_0}\) is connected and we only consider databases of degree at most \(B\).
  Additionally, \(x_2\) and each entry of \(\bar{x}_3\) is either \(0\) or \(1\).
  It is well known that FNN with \trrelu can realize functions with finitely many input values.

  For both points we use the same classification function, namely
  \[    \cls(\bar{x}) = \begin{cases}
        1 & \text{ if } x_k \geq 1\\
        0 & \text{ otherwise}.
    \end{cases}
\]
  It is straightforward to show the correctness of the translation, that is:\\
{\bf Claim.} For all databases $G$, values $v \in V(G)$, and $i,j$
with $1 \leq i \leq j \leq k$, the following holds: \[(\lambda^j_{\calN,G}(v))_i = \begin{cases}
    1&\text{ if } G\models \varphi_i(v),\\
    0&\text{ otherwise.}
\end{cases}\]%
\end{proof}

We now show the second direction of Point~2 in Theorem~\ref{thm:meancombinew}.
\begin{lemma}\label{lem:mean-DHN-to-ratio-UNFO-bounded-degree}
    Let $B\geq 0$. On databases of degree at most $B$, every connected \meanagg-DHN is equivalent to a connected RHML formula.
\end{lemma}

\begin{proof}
\newcommand{\homcount}{\ensuremath{{((B+1)^{d+1})}^{|V(F_i)|}}} 
  Let \(\calN = (\calL_1,\ldots,\calL_k, \cls)\) be a \(\meanagg\)-DHN and let \(d\) be the largest diameter of a homomorphism query in \(\calN\).
  Then for \(1\leq \ell\leq k\), the set
  \begin{align*}
	 \chi_\calN^\ell = \{\lambda_{\calN, G}^\ell(v)\mid\ &G\text{ is a database of degree at most } B,~v\in \mn{adom}(G), \text{ and}\\&\text{there is no value of distance  from  $v$ exceeding } d\ell \}
  \end{align*}
  is finite because there are only finitely many databases that satisfy these conditions.

  It is easy to see that \[\chi_\calN^\ell = \{\lambda_{\calN, G}^\ell(v)\mid G\text{ is a database of degree at most } B, v\in \mn{adom}(G)\},\]
  since $\lambda_{\calN, G}^\ell(v)$ clearly only depends on the \(\ell d\)-neighborhood of  \(v\).

  For each \(\bar{x}\in \chi_\calN^\ell\) we can now proceed to construct a formula \(\psi_{\bar{x}}^\ell(y)\) such that \[G\models \psi_{\bar{x}}^\ell(v)\text{ iff } \lambda_{\calN, G}^\ell(v) = \bar{x}\text{ for all databases } (G,v) \text{ of bounded degree } B.\]
  The construction is by induction on \(\ell\).
  Since the input dimension of DHNs is \(0\), \(\chi_\calN^0\) only contains the empty vector and \(\psi^0_{()}\) can be expressed by any tautology.
  For the induction step, let \(\ell > 0\) and let \(\calL_\ell = ((F_1^\bullet, \mu_1,\meanagg),\ldots,(F_m^\bullet,\mu_m,\mean), \com_\ell)\).
  Then every vector  \(\bar{x}\in\chi_\calN^\ell\) can be represented as \(\com_\ell(\bar{y}_1,\ldots,\bar{y}_m)\) where each \(\bar{y}_i\) is the result of the homomorphism query \((F_i^\bullet, \mu_i, \meanagg)\).
  For  \(1\leq i\leq m\), the computation of \(\bar{y}_i\) can be described as follows:
  there exists a mapping \(r_i\) that describes the distribution of homomorphisms and labelings that results in \(\bar{y}_i\). More precisely, \(r_i\) takes the form
  \[r_i: \left(\chi_\calN^{\ell-1}\right)^{V(F_i)}\to \left\{\frac{i}{j}\mid 0\leq i \leq j\leq \homcount\right\}\] and satisfies \(\sum_\lambda r_i(\lambda) = 1\) or \(\sum_\lambda r_i(\lambda)=0\) and 
  \[\bar{y}_i = \sum_{\lambda\in \left(\chi_\calN^{\ell-1}\right)^{V(F_i)}} \Big(r_i(\lambda) \cdot \prod_{v\in V(F_i)} \mu_v(\lambda(v))\Big).\]
  Note that we use \(\sum_\lambda r_i(\lambda)=0\) to describe the case that there are no homomorphisms from \(F_i^\bullet\) to~\(G^v\), which results in \(\bar{y}_i =\bar{0}\).

  The restriction to mappings \(r_i\) with range \(\{\frac{i}{j}\mid 0\leq i\leq j\leq \homcount\}\) suffices, since we only consider databases of bounded degree \(B\) and connected homomorphism queries.
  That is, the \(d\)-neighborhood of each value contains at most \(\sum_{i=0}^d B^i \leq (B+1)^{d+1}\) values and there are at most \(\homcount\) homomorphisms from \(F_i\) to such a neighborhood. 

  Next we will define an RHML formula which is satisfied by \(G^v\) iff the result of the homomorphism query \((F^v_i, \mu_i, \meanagg)\) is \(\bar{y}_i\).
  To this end, let \(V(F^{v_0}) = \{v_0,v_1,\ldots,v_n\}\) and let \(R_{\bar{y}_i}\) be the set of all functions \(r_i\) that describe a distribution of homomorphisms and labelings that result in \(\bar{y}_i\).
  Let \(\nu\) be the conjunction of atoms which define \(F^v_i\).
  To handle the edge case where there is no homomorphism from \(F_i^v\), we define \(\zeta_{r_i}\) as 
  \[\zeta_{r_i} =\exists_{>0}\,x_1\cdots x_n\,(\top, \nu(x_0,\ldots,x_n)),\]
  if \(\sum_{\lambda}r_i(\lambda) =1\) and \(\zeta_{r_i} =\top\) otherwise, where \(\top\) is any tautology.
  We can now define a formula that checks whether the embedding in the \(\ell\)-th layer for a given value is \(\bar y_i\): \[\xi_{\bar{y}_i}^\ell(x_0) = \bigvee_{r\in R_{\bar{y}_i}} \bigg(\zeta_r\,\land\bigwedge_{\lambda\in \left(\chi_\calN^{\ell-1}\right)^{V(F_i)}}\exists_{=r(\lambda)}\, x_1\cdots x_n\, \Big(\bigwedge_{i=0}^{n}\psi_{{\lambda({v}_i)}}^{\ell-1}(x_i), \nu(x_0,\ldots,x_n)\Big)\bigg).\]
  

  For each \(1\leq i\leq m\) let \(Y_i\) be the set of all possible vectors \(\bar{y}_i\). This set is again finite, since there are only finitely many functions \(r_i\).
  We can thus define \(\psi_{\bar{x}}^\ell\) as follows:
  \[\psi_{\bar{x}}^\ell = \bigvee_{\substack{\bar{x} = \com_\ell(\bar{y}_1,\ldots, \bar{y}_m)\\\bar{y}_i\in Y_i\text{ for all } i}}\bigwedge_{i=1}^m \xi_{\bar{y}_i, i}^\ell.\]
  Since the disjunction enumerates all reachable inputs to \(\com_\ell\) that result in \(\bar{x}\), the formula \(\psi_{\bar{x}}^\ell\) is satisfied by \(G^v\) iff \(\lambda_{\calN, G}^\ell(v) = \bar{x}\).
\end{proof}

\section{Proofs for Section~\ref{sect:DENs}}

We start with a more in-depth introduction to UQAFO.   We say that a UQAFO formula is 
   \emph{of the existential kind} if
   it can be generated by the 
\(\varphi_\exists(\bar{x})\) production, and likewise for formulas of the \emph{universal kind} and \(\varphi_\forall(\bar{x})\).
\begin{example}
  The formula  \(\psi(x_1,x_2,y_1,y_2) = \exists z \, (E(x_1,z) \land E(z,x_2)) \land \forall z \, (E(y_1,z) \land E(z,y_2))\) is a UQAFO formula, but neither
  of the existential kind nor of the universal
  kind. In contrast, \(\exists x \exists y\, (\forall z\, R(x,z) \land \forall z\,S(y,z))\) is a UQAFO formula of the existential kind. 
\end{example}

A straightforward analysis of the defining grammar shows that for every UQAFO formula \(\varphi(\bar{x})\) of the existential kind,  the negation normal form of \(\neg \varphi(\bar{x})\) is a UQAFO formula of the universal kind, and vice versa. We thus obtain the following basic observation.
%
\begin{lemma}\label{lemma:uqafo_neg_exis_univ_swap}
  For every UQAFO formula \(\varphi(\bar{x})\) of the existential kind, there is a UQAFO formula \(\psi(\bar{x})\) of the universal kind such that \(\psi(\bar{x}) \equiv \neg \varphi(\bar{x})\), and vice versa.
  %
\end{lemma}

\lemUQAFOUNFOC*
%
\begin{proof}
The main step towards proving that UQAFO is at least as expressive as UNFOC is to show that for each GHML formula there is an equivalent UQAFO formula that is both of the existential and of the universal kind. We do this by structural induction.
  It is easy to see that this class of UQAFO formulas is closed under disjunction. By
  Lemma~\ref{lemma:uqafo_neg_exis_univ_swap}, it is also closed under negation. It thus
  remains to treat formulas \(\exists^{\geq n_1}y_1\cdots \exists^{\geq n_k}y_k \, \bigvee_i \psi_i(x,y_1,\ldots,y_k)\),
  which constitute the induction start if each $\psi_i$ contains only
  atomic formulas, but which also occur in the induction step. We shall not distinguish between these two kinds of occurrences.
  
  First note that each conjunction \(\psi_i(x,y_1,\ldots,y_k)\) is equivalent to
  a UQAFO formula of the existential kind, by induction hypothesis and since formulas of the existential kind are closed under conjunction.
  Since they are also closed under disjunction,
  \(\bigvee_i \psi_i(x,y_1,\ldots,y_k)\) is also a UQAFO formula of the existential kind.
  The counting quantifiers can
  be rewritten using existential quantifiers and inequalities.
  The resulting formula is of the existential kind and since there is only one free variable, it is also of the universal kind.


By Lemma~\ref{lem:twocombined} and what we have just shown, every UNFOC formula in one free variable is equivalent to a UQAFO formula. Note that every UNFOC formula, independently of the number of free variables, can be written as a combination in terms of $\wedge$ and $\vee$ of the following: atoms $R(\bar x)$, 
formulas $\exists^{\geq n} x \, \varphi$, and formulas $\neg \varphi(x)$. It is easy to see that every such combination  is equivalent to an existential UQAFO formula. In particular,
$\exists^{\geq n} x \, \varphi$ can be rewritten
using existential quantifiers and inequalities
and for formulas $\neg \varphi(x)$ we use that every UNFOC formula in one free variable is equivalent to a UQAFO formula.

  It remains to show that UQAFO is strictly more expressive than GHML. This is the case since the UQAFO formula \[\forall yz\, ((E(x,y)\land E(y,z))\rightarrow E(x,z))\] expresses local transitivity which by Lemma~\ref{prop:nolocaltransinGHMLstar} cannot be expressed in GHML.
\end{proof}

We next argue that EML does not become more expressive when counting quantifiers are added.
  Tuple counting quantifiers of the form \(\exists^{\geq n} \bar{y}\,\psi(x,\bar{y})\) are the most straightforward to express in EML.
  They can be written as \[\exists \bar{y}_1\cdots\bar{y}_n\, \big(\bigwedge_{1\leq i\leq n} \psi(x,\bar{y}_i)\land\bigwedge_{1\leq i<j\leq n} \bar{y}_i \neq \bar{y}_j\big),\]
  where \(\bar{y}_i\neq \bar{y}_j\) is an abbreviation for \(\bigvee_{1\leq k\leq |\bar{y}|} (\bar{y}_i)_k \neq (\bar{y}_j)_k\).
  Using distributivity laws, the disjunctions can be moved in front of the existential quantifiers.
  The resulting formula is then an EML formula.
  By applying this approach inductively, we can express counting quantifiers of the form \(\exists^{\geq n_1}\bar{y}_1\cdots\exists^{\geq n_k}\bar{y}_k \,\psi(x,\bar{y}_1,\ldots,\bar{y}_k)\) as well:
  We can replace the outermost counting quantifier in \(\exists x_1\cdots x_\ell \exists^{\geq n_1}\bar{y}_1\cdots \exists^{\geq n_k}\bar{y}_k\, \psi(x,x_1,\ldots,x_\ell,\bar{y}_1,\ldots,\bar{y}_k)\) by introducing \(n_1\) copies of the variables \(\bar{y}_1\) and we get the equivalent formula
\[\exists x_1\cdots x_\ell\, \bar{y}_1^1\cdots \bar{y}_1^{n_1}\!\!\bigwedge_{1\leq i\leq n_1}\!\!\big( \exists^{\geq n_2}\bar{y}_2\cdots \exists^{\geq n_k}\bar{y}_k\, \psi(x,x_1,\ldots,x_\ell,\bar{y}_1^i, \bar{y}_2,\ldots,\bar{y}_k)\land \!\!\bigwedge_{1\leq i<j\leq n_1}\!\! \bar{y}_1^i \neq \bar{y}_1^j \big).\]
By using variable renaming and scope changes, we can move the conjunction between existential quantifiers inwards behind all existential quantifiers.
We can deal with the disjunctions introduced by \(\bar{y}_1^i\neq \bar{y}_1^j\) as before.

\subsection{Point~1 of Theorem~\ref{thm:DENcombi}}

We prove Point~1 of
Theorem~\ref{thm:DENcombi} 
and the inclusion part of Point~2. To start with, note that, trivially,
every EML formula is a UQAFO formula.
\begin{lemma}
  Every unary UQAFO formula is equivalent to an EML formula.
\end{lemma}
\begin{proof}
  To transform a unary UQAFO formula \(\varphi\) into an EML formula, we group the quantifiers first, and then we replace universal quantifiers with negated existential quantifiers.
  After grouping the quantifiers each subformula \(\psi(\bar{x}) = \exists y\, \psi'(\bar{x},y)\) of \(\varphi\) should satisfy:
  \begin{itemize}
    \item \(\psi'\) is of the form \(\exists z\, \psi''(\bar{x},y,z)\), or
    \item \(\psi'\) is a Boolean combination of relational atoms, (in-)equalities, and unary UQAFO formulas.
  \end{itemize}
  Analogous conditions should be satisfied by subformulas \(\forall y\, \psi'(\bar{x},y)\).

  The translation of UQAFO formulas into one that satisfies these conditions is by induction on the quantifier depth.
  For formulas of quantifier depth \(0\) there is nothing to do.
  For the induction step, we describe how to translate formulas that start with an existential quantifier.
  The case for formulas that start with a universal quantifier is analogous.
  Let \(\psi(\bar{x}) = \exists y\, \psi'(\bar{x},y)\) be a formula of quantifier depth \(i+1\) that does not satisfy any of the above conditions.
  We can write \(\psi'\) as a Boolean combination of relational atoms, (in-)equalities, and UQAFO formulas which \emph{start immediately with a quantifier}.
  Let \(\chi_1,\ldots,\chi_n\) be the relational atoms, (in-)equalities and unary UQAFO formulas, and let \(\xi_1,\ldots,\xi_k\) be the remaining non-unary UQAFO formulas of this Boolean combination.
  Since \(\psi\) is a UQAFO formula, all \(\xi_j\) have to start with an existential quantifier.
  By induction hypothesis, all of  \(\xi_1,\ldots,\xi_k\) can be transformed such that they satisfy the first condition above, in particular, they can be written as \(\xi_j \equiv \exists\bar{y}_j\,\xi_j'\), where \(\xi_j'\) does not start with an existential quantifier.
  By variable renaming, we can assume that all \(\bar{y}_i\) as well as \(\bar{x}\) and \(y\) are pairwise disjoint.
  It is thus safe to move all quantifiers to the front, and therefore \(\psi\) is equivalent to
  \[\exists y\ \bar{y}_1\cdots\bar{y}_k\, \psi'',\]
  where \(\psi''\) is the result of replacing every occurrence \(\xi_i\) with \(\xi_i'\) in \(\psi'\) for all \(1\leq i\leq k\).
  \(\psi''\) is then a Boolean combination of \(\chi_1,\ldots,\chi_n\) as well as the atoms, (in-)equalities and unary UQAFO formulas in \(\xi_1',\ldots,\xi_k'\).
  Hence, \(\psi\) and its subformulas satisfy the conditions above.

  We can now get rid of all universal quantifiers, by replacing each \(\forall \bar{y}\,\psi(x,\bar{y})\), where \(\psi\) is a Boolean combination of relational atoms, (in-)equalities and unary UQAFO formulas with \[\neg\exists\bar{y}\,\neg\psi(x,\bar{y}).\]
  The negation in front of the existential quantifier is allowed in EML since it negates a unary formula.
  This step also replaces every universal quantifier, since the formula satisfies the conditions above.
  By using well-known logical equivalences, we can move the negation in front of \(\psi\) inwards, until it is in front of the atoms, (in-)equalities and unary UQAFO formulas.
  
  In the last step, for each formula \(\exists \bar{y}\, \psi(x,\bar{y})\), we transform \(\psi(x,\bar{y})\) in disjunctive normal form, move the disjunctions in front of the existential quantifiers and use variable renaming to get rid of the remaining equalities.
\end{proof}

We next show that every EML formula can be converted into a certain strict form. In a formula in this form, every subformula  \(\exists \bar{y} \, \psi(x,\bar{y})\),  describes a proper embedding, rather than only a `partial' embedding.
In subsequent constructions, it will be useful to assume that EML formulas are in strict form.

  We say that an EML formula \(\varphi\) is in \emph{strict form} if for each subformula \(\exists \bar{y} \, \psi(x,\bar{y})\), \(\psi\) contains at least the following conjuncts:
  \begin{itemize}
    \item every inequality \(z_1\neq z_2\) for every pair of variables \(z_1,z_2\in \{x\}\cup\bar{y}\) with \(z_1\neq z_2\).
    \item every relational atom \(R(z_1,\ldots,z_k)\) or its negation \(\neg R(z_1,\ldots,z_k)\), for each \(z_1,\ldots,z_k\in \{x\}\cup\bar{y}\) and each relational symbol in the schema.
  \end{itemize} 
\medskip
\begin{lemma}
  For each EML formula there exists an equivalent strict EML formula. 
\end{lemma}
\begin{proof}

We show by induction on the structure of EML formulas how to transform an EML formula into a strict one.
We only have to show in the induction step how to transform existential quantifiers.
Let \(\exists \bar{y}\, \psi(x,\bar{y})\) be an EML formula.
Then \(\psi(x,\bar{y})\) can be written as \(\psi_1(x,\bar{y})\land \psi_2(x,\bar{y})\), where \(\psi_1\) is a conjunction of relational atoms, negated relational atom and inequalities, while \(\psi_2\) is a conjunction of unary EML formulas.
By the induction hypothesis, each conjunct in \(\psi_2\) can be translated into a strict EML formula, thus \(\psi_2\) is equivalent to a conjunction of strict unary EML formulas.
We denote this formula by \(\psi_2'\).
\(\psi_1\) is equivalent to a disjunction \(\bigvee_{i} \varphi_i(x,\bar{y})\) where each \(\varphi_i\) is a conjunction such that:
\begin{itemize}
    \item For each relational atom \(R\) of arity \(k\) and \(z_1,\ldots,z_k\in \{x\}\cup\bar{y}\) it either contains \(R(z_1,\ldots,z_k)\) or \(\neg R(z_1,\ldots,z_k)\) as a conjunct.
    \item For each pair of variables \(z_1,z_2\in \{x\}\cup\bar{y}\), it contains either \(z_1 =z_2\) or \(z_1\neq z_2\) as a conjunct.
\end{itemize}
The formula \(\exists\bar{y}\, \psi(x,\bar{y})\) is then equivalent to \[\bigvee_i \exists \bar{y}\, (\varphi_i(x,\bar{y}) \land \psi_2'(x,\bar{y}).\]
Each equivalence can be removed by variable renaming.
The resulting formula is then a strict EML formula.
\end{proof}

We  next show that max-DENs can realize every strict EML formula and thus also every unary UQAFO formula.
\begin{lemma}
  Every strict EML formula 
is  equivalent to a max-DEN. 
\end{lemma}

%
\begin{proof}
  This proof is analogous to the proof of Lemma~\ref{lem:UNFO-to-max-DHN-uniform}, that is, we translate each strict EML formula to a \maxagg-DEN.
  In particular, we use the definition of \emph{subformula} according to the syntax rules of EML.
  Let \(\varphi\) be a strict EML formula and let \(\varphi_1,\ldots,\varphi_k\) be an enumeration of the subformulas in \(\varphi\) such that if \(\varphi_\ell\) is a subformula of \(\varphi_k\) then \(\ell\leq k\).

  We construct a \maxagg-DEN \(\calN = (\calL_1, \ldots,\calL_k, \cls)\) with \(k\) layers, all of output dimension \(k\).
  As usual, we encode satisfaction of \(\varphi_i\) with \(1\) and falsification with \(0\) in the \(i\)-th component.
  The DEN evaluates one subformula in each layer so that for every \(1\leq i\leq k\), from layer \(i\) on all subformulas \(\varphi_1,\ldots,\varphi_i\) are encoded correctly.

  For formulas \(\varphi_i = \varphi_{i_1} \lor \varphi_{i_2}\) and \(\varphi_i = \neg \varphi_j\) we can adapt the homomorphism queries in the proof of Lemma~\ref{lem:UNFO-to-max-DHN-uniform}. 
  There we used exactly one homomorphism query \((F^\bullet, \mu,\max)\), where \(F^\bullet\) was a single-value database without any facts  and where \(\com\) was an affine function followed by \trrelu.
  For \maxagg-DENs we cannot use \(F^\bullet\) as the only database, but we now must use multiple embedding queries to exactly match the relational atoms satisfied at the values in the input database.
  We thus use \((F^\bullet_1, \mu,\max),\ldots,(F^\bullet_\ell,\mu,\max)\), where \(F^\bullet_1,\ldots,F^\bullet_\ell\) is an enumeration of all single-value databases over the schema \(\Sbf\).
  We replace the combination function by \(\com'(\bar{x}_1,\ldots,\bar{x}_\ell) = \com(\bar{x}_1+\cdots+\bar{x}_\ell)\).
  Since addition is affine, \(\com'\) is still an affine function followed by \trrelu activation.
  Additionally, the embedding query will return the same value as the homomorphism query, since exactly one \(F^\bullet_i\) can embed at each value.
  This translation allows us to reuse the homomorphism queries which we used in Lemma~\ref{lem:UNFO-to-max-DHN-uniform} to encode negations and disjunctions.

  For the last case, assume that \(\varphi_i\) takes the form \(\exists \bar{y}\, \varphi_k(x,\bar{y})\), which is equivalent to the form
  \[\exists \bar{y} \,\Big(\psi_0\land \bigwedge_{z\in \{x\}\cup \bar{y}} \psi_z  \Big),\]
  where \(\psi_0\) is a conjunction of relational atoms and each \(\psi_z\) is a conjunction of (zero or more) EML formulas with the free variable \(z\).
  We use the following embedding query to check whether \(\varphi_i\) itself is satisfied:
  
  \begin{itemize}
    \item \(F_1^x\) is the database that is described by \(\psi_0\). Since \(\varphi_i\) is a strict EML formula, \(F_1^x\) is uniquely defined.
    \item \(\mu_1\) assigns to each value \(z\in \mn{adom}(F_1)\) the transformation function that maps every tuple \(\bar{x}\in \mathbb{R}^k\) to the tuple \((1)\) if the \(j\)-th position of \(\bar{x}\) is \(1\) for all conjuncts \(\varphi_j(z)\) in \(\psi_z\), and \((0)\) otherwise.
          It also returns \((1)\) if \(\psi_z\) is the empty conjunction.
  \end{itemize}
  To copy the previous feature vector, we use the embedding queries \((G^\bullet_1,\mu_2,\max),\ldots,(G^\bullet_\ell,\mu_2,\max)\), where
  \begin{itemize}
    \item \(G^\bullet_1,\ldots,G^\bullet_\ell\) is an enumeration of all single-value databases over the schema \(\Sbf\),
    \item if \(i=1\), \(\mu_2\) returns the \(k\)-dimensional vector with all entries set to \(0\),
    \item if \(i>1\), \(\mu_2\) is the identity function.
  \end{itemize}
  The combination function \(\com_i\) adds all vectors returned by \((G^\bullet_i,\mu_2,\max)\), and overrides the \(i\)\nobreakdash-th position in this sum by the result of the first embedding query.
  It is easy to verify that all transformation functions are affine with \trrelu activation.

  As the classification function we use
  
    \[\cls(\bar{x}) = \begin{cases}
        1 & \text{ if } x_k \geq 1\\
        0 & \text{ otherwise}.
    \end{cases}\]
It is straightforward to show the correctness of the translation, that is:
\\[2mm]
{\bf Claim.} For all databases $D$, values $v \in \mn{adom}(G)$, and $i,j$
with $1 \leq i \leq j \leq k$, the following holds: \[(\lambda^j_{\calN,G}(v))_i = \begin{cases}
    1&\text{ if } G\models \varphi_i(v),\\
    0&\text{ otherwise.}
\end{cases}\]
\end{proof}

\begin{corollary} \label{lem:uqafo_to_sum_den}
  Every unary UQAFO formula is equivalent to a sum-DEN. 
\end{corollary}
\begin{proof}
  To lift the previous result to sum-DENs, one can easily verify that in the previous construction all \(\mu\) return vectors with entries in \(\{0,1\}\).
  Thus each entry in the product of the generalized embedding number is in \(\{0,1\}\).
  In this setting for each multiset \(M\) over \(\{0,1\}^d\) \[\trrelu\left(\sum M\right) = \max (M).\]
  To construct a sum-DEN, one can thus replace max with sum and apply ReLU before applying the FNN.
\end{proof}

To show that each \(\maxagg\)-DEN can be transformed into an equivalent EML formula, we adapt the proofs of Lemma~\ref{lem:finitelymanyvectors} and Lemma~\ref{lem:max-DHN-to-UNFO-uniform}.
\begin{lemma}\label{lem:finitelymanyvectors-den}
  Let \(\Sbf\) be a schema and let $\mathcal{N}=(\calL_1, \dots \calL_k,\cls)$ be a $\maxagg$-DEN over \(\Sbf\).
  Then for $1 \leq \ell \leq k$, the set 
  \[
  \chi^\ell_{\calN}= \{ \lambda^\ell_{\calN,G}(v) \mid G \text{ a \(\Sbf\)-database and } v \in \mn{adom}(G) \}
  \]
  is finite.
\end{lemma}

\begin{lemma}
    Every  $\maxagg$-DEN $\calN$ is equivalent to an EML formula.
\end{lemma}
\begin{proof}
    Let $\calN = (\calL_1, \dots \calL_k,\cls)$ be a  $\maxagg$-DEN.
    By Lemma~\ref{lem:finitelymanyvectors-den},
    we may build an equivalent EML formula
    $\varphi(x)$ by constructing, for every
    $\ell \in \{1,\dots,k\}$ and embedding vector $\bar x \in \chi^\ell_{\calN}$,
    an EML formula $\psi^\ell_{\bar x}(x)$ such that
    \[
      G \models \psi^\ell_{\bar x}(v) \text{ iff } \lambda^\ell_{\calN,G}(v)=\bar x
      \text{ for all pointed \(\Sbf\)-databases } (G,v)
    \]
    and then taking the disjunction of all
    formulas $\psi^k_{\bar x}$ such that $\cls(\bar x)=1$.

    The construction of the formulas  $\psi^\ell_{\bar x}$ is by induction on $\ell$.
    For $\ell=0$, the only relevant vector $\bar x$ is the empty
    vector and we may set $\psi^0_{()}(x)=\text{true}$.
    Now let $\ell>0$ and $\bar x \in \chi^\ell_{\calN}$.
    Further let  $\calL_\ell=(\calF_\ell,\com_\ell)$,
 with $\calF_\ell$ the sequence $(F^{u_1}_1,
\mu_{1},\maxagg), \dots, (F^{u_m}, \mu_{m},\maxagg)$.
We define $\psi^\ell_{\bar x}$ to be a disjunction
that contains one disjunct for every  possible way to choose for every
$i \in \{1,\dots,m\}$ a partition $X_i,Y_i$ of the set of mappings $\lambda:\mn{adom}(F_i) \rightarrow \chi_{\calN}^{\ell-1}$ such that 
$\com_\ell(\bar y_1,\dots, \bar y_m) = \bar x$ where, for $1 \leq i \leq m$,
\begin{align*}
   \bar y_i := \underset{\lambda \in X_i}{\maxagg} \prod_{v \in \mn{adom}(F)} \mu_v(\lambda(v)).
\end{align*}
Note that the definition of $\bar y_i$ corresponds to the evaluation of the embedding query $(F^{u_i}_i,
\mu_{i},\maxagg)$ under any embedding $h$ that is 
`compatible' with  some $\lambda \in X_i$, that is, $h$ maps each $v \in \mn{adom}(F_i)$ to a vertex
$h(v)$ in the input graph that carries the embedding vector $\lambda(v)$. Clearly, we may view
the  database $F_i$ as a
conjunction $\vartheta_i$ of relational
atoms and negations thereof and inequalities.
Let $\bar z_i$ denote the tuple of variables in $\vartheta_i$, ordered arbitrarily but without $x$. Then \(\psi^\ell_{\bar{x}}(x)\) includes as a disjunct
\[
  \bigwedge_{1 \leq i \leq m} \Big (\bigwedge_{\lambda \in X_i}
  \exists \bar z_i \, \big (\vartheta_i
  \wedge \bigwedge_{v \in \mn{adom}(F_i)} \psi^{\ell-1}_{\lambda(v)}(v)
  \big )
  \wedge
 \bigwedge_{\lambda \in Y_i}
  \neg
  \exists \bar z_i \, \big (\vartheta_i
  \wedge \bigwedge_{v \in \mn{adom}(F_i)} \psi^{\ell-1}_{\lambda(v)}(v)
  \big ) \Big ).
\]
It is easy to verify that $\psi^\ell_{\bar x}$ is as required.
\end{proof}

\subsection{Point~2 of Theorem~\ref{thm:DENcombi}}

To prove that UQAFO does not capture \(\sumagg\)-DENs, we use a strategy similar to the one in the proof of Proposition~\ref{thm:sum-DHN-not-in-UNFOC-uniform}.
That is, we  first show that PQR-orders can be defined by \(\sumagg\)-DENs, and then introduce UQAFO games and use them to show that UQAFO cannot express PQR-orders.

\begin{lemma} \label{lem: sum-DEN can express P}
There exists a sum-DEN that accepts pointed graphs iff they are PQR-orders.
\end{lemma}
\begin{proof}
  We first use UQAFO formulas and the translation in Corollary~\ref{lem:uqafo_to_sum_den} to build a \sumagg-DEN that checks some of these properties.
  We then describe how to extend this DEN to check all requirements of a PQR-order.
\begin{itemize}
  \item The following UQAFO formula describes linear orders in which each vertex satisfies either \(P, Q,\) or \(R\), and the vertex bound to the free variable satisfies \(P\), and the last vertex satisfies \(R\).
        \begin{alignat*}{2}
          &\forall yz\, \big((E(x,y)\land E(y,z)) \rightarrow E(x,z)\big) &\\
          &\land\, \forall yz\, \big((E(x,y)\land E(x,z)\land y\neq z) \rightarrow (E(y,z) \leftrightarrow \neg E(z,y))\big) &\\
          &\land\, \forall uvw\, \big((E(x,u)\land E(u,v)\land E(v,w))\rightarrow E(u,w)\big) & \\
          &\land\, E(x,x) \land \forall y \,\big(E(x,y) \rightarrow E(y,y)\big)&\\
          &\land\, P(x)\land \forall y\, \big((E(x,y)\land \forall z\, (y=z\lor\neg E(y,z)))\rightarrow R(y)\big)&\\
          &\land\,\forall y \,\big( E(x,y)\rightarrow ((Q(y)\land\neg P(y)\land\neg R(y))&\\
          &\phantom{\land\,\forall y \,\big( E(x,y)\rightarrow (}\ 
          \lor (\neg Q(y)\land P(y)\land \neg R(y))
          \lor (\neg Q(y)\land\neg P(y)\land R(y)) )\big)&
        \end{alignat*}
        By Corollary \ref{lem:uqafo_to_sum_den}, this can be expressed by a sum-DEN.
        The first conjunct ensures that there does not exist a two-hop successor.
        The second, third, and fourth conjuncts describes linear orders, that is, \(E\) is total, antisymmetric, transitive and reflexive.
        And the last two conjuncts ensure that the free variable satisfies \(P\), the last vertex satisfies \(R\) and every vertex satisfies either \(P,Q\) or \(R\).
  \item The usual approach in FO to express that the immediate successor satisfies \(\psi\) is not in UQAFO:
        \(\varphi(x) = \exists y\, (x\neq y \land E(x,y) \land \forall z\, (x = z\lor y = z\lor\neg E(x,z) \lor \neg E(z,y)) \land \psi(y))\) contains a quantifier alternation with more than one free variable.
        But the DEN can use its counting to check an equivalent property.
        It instead computes and compares the number of \(P,Q\) and \(R\) successors. 
        It is easy to verify that it suffices to check the following property for all vertices \(u\) in the linear order:
        We denote the number of successors of \(u\) that satisfy \(P,Q\) and \(R\) by \(n_P,n_Q\) and \(n_R\) respectively.
        \begin{itemize}
          \item If \(u\) satisfies \(P\), then \(n_P = n_Q = n_R\),
          \item if \(u\) satisfies \(Q\), then \(n_P+1 = n_Q = n_R\), and
          \item if \(u\) satisfies \(R\), then \(n_P+1 = n_Q+1 = n_R\).
        \end{itemize}
        
        All these properties can be computed with a \sumagg-DEN with \trrelu activation, since \(x\geq y\) can be checked for \(x,y\in\mathbb{N}\) with \(\trrelu(x-y+1)\) and FNNs can realize conjunctions.
  \item Last but not least, with another sum aggregation, the DEN can check for \(v\) that all its successors satisfy the property above.
\end{itemize}
\end{proof}

To show that this property cannot be expressed by a UQAFO formula, we 
introduce {Ehrenfeucht-Fra\"iss\'e} games for UQAFO.
We have to define the game for formulas with more than one free variable, thus we allow databases to have a tuple of distinguished values which may contain more than one entry.

\begin{definition}[UQAFO game]
  Let \(\ell \geq 0\). The \(\ell\)-round UQAFO game is played between two players, \emph{Spoiler} and \emph{Duplicator}, on two pointed databases \((D_1,v_1^{(1)},\ldots,v_1^{(n)})\) and \((D_2,v_2^{(1)},\ldots,v_2^{(n)})\).
  A configuration consists of both pointed databases and additionally an index \(s\in\{1,2\}\) that describes which database \emph{Spoiler} is currently playing in. 
  For the initial configuration \emph{Spoiler} can choose \(s\in \{1,2\}\).
  In each of the \(\ell\) rounds, \emph{Spoiler} can make one of the following moves:

  \begin{itemize}
    \item \emph{Spoiler} can play a value \(u_s\) in \(D_s\). Then \emph{Duplicator} responds with a value \(u_{3-s}\) in \(D_{3-s}\). The next configuration consists of the pointed databases \((D_1,v_1^{(1)},\ldots, v_1^{(n)}, u_1)\) and \((D_2, v_2^{(1)},\ldots,v_2^{(n)}, u_2)\).
          The index \(s\) stays unchanged.
    \item \emph{Spoiler} can decide to play in \(D_{3-s}\). To do so, they choose one index \(1\leq k\leq n\) to remember the distinguished values \(v_1^{(k)}\) and \(v_2^{(k)}\).
    Spoiler then chooses a value \(u_{3-s}\) in \(D_{3-s}\). \emph{Duplicator} responds with a value \(u_s\) in \(D_s\).
    The next configuration consists of the databases \((D_1, v_1^{(k)}, u_1)\) and \((D_2, v_2^{(k)}, u_2)\) and index \(3-s\).
  \end{itemize}
  \emph{Spoiler} wins this game if for any configuration the set \(\{(v_1^{(k)}, v_2^{(k)})\mid 1\leq k\leq n\}\) does not describe a partial isomorphism between \(D_1\) and \(D_2\).
  \emph{Duplicator} wins if \emph{Spoiler} does not win after \(\ell\) rounds.
\end{definition}

Intuitively, Spoiler shows in this game that there exists some subdatabase in \(D_s\) while Duplicator tries to find the same subdatabase in \(D_{3-s}\).
If Spoiler plays a value in \(D_s\), they extend this subdatabase.
As in GHML game, Spoiler extends this subdatabase every time they play in \(D_s\).
One could also define an equivalent game more in line with EML in which Spoiler is forced to play the embedding within one step.
Since UQAFO and EML are equally expressive, this does not weaken Spoilers abilities.
That is, given two pointed databases, if Spoiler wins the UQAFO game in \(\ell\) rounds then Spoiler wins the EML game in \(\ell'\) rounds, where \(\ell'\) only depends on \(\ell\).

We now show that Spoiler has a winning strategy in the UQAFO game on databases \(D_1,D_2\), iff there exists a UQAFO formula that distinguishes both databases.

\begin{lemma}\label{lem:uqafo_games_correspondence}
  For all \(\ell\geq 0\), all pointed databases \((D_1,v_1^{(1)},\ldots, v_1^{(n)}),(D_2,v_2^{(1)},\ldots,v_2^{(n)})\) the following statements are equivalent:
  \begin{enumerate}
    \item Spoiler has a winning strategy in the \(\ell\)-round UQAFO game when playing on both databases.
    \item There exists a UQAFO formula \(\varphi(x_1,\ldots,x_n)\) with quantifier depth at most \(\ell\) such that \[D_1\models \varphi(v_1^{(1)}, \ldots,v_1^{(n)})\iff D_2\not\models \varphi(v_2^{(1)},\ldots,v_2^{(n)}).\]
  \end{enumerate}
\end{lemma}
\begin{proof}
  We show the implication \(1.\implies 2.\) by induction on \(\ell\).
  To be more precise, we additionally show that if \(s\) is the index of the current configuration then there exists a formula \(\varphi(x_1,\ldots,x_n)\) of the existential kind of UQAFO such that \(D_s\models \varphi(v_s^{(1)},\ldots,v_s^{(n)})\) and \({D_{3-s}\not\models\varphi(v_{3-s}^{(1)},\ldots,v_{3-s}^{(n)})}\).

  Spoiler wins the \(0\)-round UQAFO game if and only if the set of pairs \(\{(v_1^{(k)}, v_2^{(k)})\mid 1\leq k\leq n\}\) does not describe a partial isomorphism between \(D_1\) and \(D_2\). It is easy to verify that there exists an FO formula \(\varphi(x_1,\ldots,x_n)\) without quantifiers that describes the subdatabase induced by \(v_s^{(1)},\ldots, v_s^{(n)}\) up to isomorphism.
  Then, \(D_s\models \varphi(v_s^{(1)},\ldots,v_s^{(n)})\) and \(D_{3-s}\not\models \varphi(v_{3-s}^{(1)},\ldots,v_{3-s}^{(n)})\).
  Since UQAFO does not restrict the use of the Boolean operators within formulas without quantifiers, \(\varphi\) is a UQAFO formula of the existential kind.

  For the induction step, assume that the statement holds for some value of \(\ell\). We have to show that this statement also holds for \(\ell+1\).
  Assume that Spoiler has a winning strategy in the \(\ell+1\) round game, starting with index \(s\) and pointed databases \((D_1,v_1^{(1)},\ldots,v_1^{(n)}), (D_2, v_2^{(1)},\ldots,v_2^{(n)})\).
  We distinguish the two cases how Spoiler can play:
  \begin{itemize}
    \item\emph{Spoiler plays in \(D_s\):} That is, Spoiler plays a value \(u_s\) of \(D_s\) and no matter which value \(u_{3-s}\) of \(D_{3-s}\) Duplicator responds with, Spoiler has a winning strategy in the \(\ell\)-round UQAFO game on 
          \((D_1,v_1^{(1)},\ldots,v_1^{(n)},u_1)\) and \((D_2,v_2^{(1)},\ldots,v_2^{(n)},u_2)\).
          By the induction hypothesis, for each choice \(u_{3-s}\) of Duplicator, there exists a UQAFO formula \(\varphi_{u_{3-s}}(x_1,\ldots,x_{n+1})\) of the existential kind such that \(D_s\models \varphi_{u_{3-s}}(v_s^{(1)},\ldots,v_s^{(n)},u_s)\) and \(D_{3-s}\not\models\varphi_{u_{3-s}}(v_{3-s}^{(1)},\ldots,v_{3-s}^{(n)}, u_{3-s})\).
    Let \(\Phi = \{\varphi_{u_{3-s}}\}\) be the set of all such formulas, which w.l.o.g does not contain two equivalent formulas. Since there are only finitely many pairwise non-equivalent UQAFO formulas of quantifier depth at most \(\ell\), \(\Phi\) is a finite set.
    We now claim that \(\psi = \exists x_{n+1} \bigwedge_{\varphi\in \Phi}\varphi\) is the formula we are searching for. Obviously, \((D_s,v_s^{(1)},\ldots,v_s^{(n)})\) is a model of \(\psi\).
    And \((D_{3-s},v_{3-s}^{(1)},\ldots,v_{3-s}^{(n)})\) is not a model of \(\psi\) since for each choice of \(u_{3-s}\) for \(x_{n+1}\), \(\varphi_{u_{3-s}}\) is not satisfied.
    \(\psi\) is also of the existential kind, since UQAFO formulas of the existential kind are closed under conjunctions as well as binding variables with existential quantifiers.
    \item\emph{Spoiler plays in \(D_{3-s}\)}: That is, Spoiler chooses an index \(k\) to remember the distinguished values \(v_1^{(k)}\) and  \(v_2^{(k)}\) and then plays a value \(u_{3-s}\). As in the first case and by the induction hypothesis, no matter what value \(u_s\) Duplicator chooses, there always exists a UQAFO formula \(\varphi_{u_{s}}(x_1,x_2)\) of the existential kind such that \((D_{3-s}, v_{3-s}^{(k)}, u_{3-s})\) is a model and \((D_s, v_s^{(k)}, u_s)\) is not a model.
    As before, let \(\Phi = \{\varphi_{u_s}\}\) which again only contains finitely many pairwise non-equivalent formulas.
    Analogously to the first case, \(\psi(x_1) = \exists x_2\, \bigwedge_{\varphi\in\Phi} \varphi\) is satisfied by \(D_{3-s}\) but not by \(D_s\) and it is a UQAFO formula of the existential kind.
          Therefore, \(\neg \psi(x_1)\) has the desired models and by Lemma \ref{lemma:uqafo_neg_exis_univ_swap} there exists an equivalent formula in UQAFO of the universal kind.

    Since \(\neg \psi(x_1)\) is a formula with only one free variable, it is also of the existential kind since it can be derived using the rule \(\varphi_\exists(x_1)\to \varphi_\forall(x_1)\) first.      
   \end{itemize}
  This finishes the implication \(1.\implies 2.\)
  We also show the implication \(2.\implies 1.\) by induction on \(\ell\).
  To be more precise, we show the following statement: Let \((D_1, v_1^{(1)},\ldots v_1^{(n)})\) and \((D_2, v_2^{(1)}, \ldots, v_2^{(n)})\) be two databases and let \(\varphi(x_1,\ldots,x_n)\) be a UQAFO formula such that \((D_1,v_1^{(1)},\ldots,v_1^{(n)})\) is a model of \(\varphi\) and \((D_2,v_2^{(1)},\ldots,v_2^{(n)})\) is not a model of \(\varphi\).
  If \(\varphi\) is of the existential kind then Spoiler has a winning strategy in the UQAFO game when \(s = 1\), and if \(\varphi\) is of the universal kind, then Spoiler has a winning strategy when \(s = 2\).
  The case if \(\varphi\) is neither is also covered, because then \(\varphi\) is a Boolean combination of formulas of existential and formulas of the universal kind.
  Thus, there exists a subformula \(\psi(x_1,\ldots,x_n)\) of \(\varphi(x_1,\ldots,x_n)\) such that \(\psi\) is of the existential or the universal kind and \((D_1,v_1^{(1)},\ldots,v_1^{(n)})\) is model of \(\psi\) and \((D_2,v_2^{(1)},\ldots,v_2^{(n)})\) is not a model of \(\psi\).
  The statement of the lemma follows, since Spoiler can choose in the first round whether they want to play in \(D_1\) or \(D_2\).

  For \(\ell=0\), \(\varphi\) is a formula without quantifiers, that is, it is a Boolean combination of atoms.
  Therefore, the subdatabases induced by \(v_1^{(1)},\ldots,v_1^{(n)}\) and \(v_2^{(1)},\ldots,v_2^{(n)}\) are not isomorphic and thus Spoiler wins the \(0\)-round UQAFO game independent of the database \(D_s\) they are playing in.

  For the induction step, assume this statement is valid for all formulas with quantifier depth at most \(\ell\). We show that the statement also holds for \(\ell+1\).
  We start with the case that \(\varphi\) is of the existential kind.
  Let \(\varphi(x_1,\ldots,x_n)\) be such a UQAFO formula with quantifier depth \(\ell+1\) such that \(D_1\) is a model and \(D_2\) is not a model.
  It is a \(\{\land,\lor\}\)-combination of literals and UQAFO formulas that immediately start with a quantifier and which have quantifier depth at most \(\ell+1\).
  If one of the literals distinguishes \(D_1\) and \(D_2\), then the subdatabases induced by \(v_1^{(n)}\ldots, v_1^{(n)}\) and \(v_2^{(1)},\ldots, v_2^{(n)}\) are again not isomorphic and Spoiler wins the game immediately.
  Otherwise, there exists a subformula \(\psi\) that immediately start with a quantifier that distinguishes \(D_1\) and \(D_2\).
  And since there are no negations at this level, \(D_1\) is a model of \(\psi\) and \(D_2\) is not a model of \(\psi\).
  If \(\psi\) has a quantifier depth of at most \(\ell\), we can apply the induction hypothesis.
  Otherwise, we distinguish two cases:
  \begin{itemize}
    \item \emph{\(\psi = \exists x_{n+1}\,\psi'\):} Then there exists a value \(v_1^{(n+1)}\) of \(D_1\) such that \((D_1,v_1^{(1)},\ldots v_1^{(n+1)})\) is a model of \(\psi'\), and for all \(v_2^{(n+1)}\) of \(D_2\), \((D_2,v_2^{(1)},\ldots,v_2^{(n+1)})\) is not a model of \(\psi'\).
          Since \(\psi\) is of the existential kind, \(\psi'\) has to be of the existential kind of UQAFO as well and by induction hypothesis, Spoiler wins the \(\ell\)-round UQAFO game on these two databases when \(s=1\).
          Therefore, by choosing \(v_1^{(n+1)}\), Spoiler can win the \(\ell+1\)-round UQAFO game on \((D_1,v_1^{(1)}, \ldots,v_1^{(n)})\) and \((D_2, v_2^{(1)},\ldots, v_2^{(n)})\).
    \item \emph{\(\psi = \forall x_{n+1}\,\psi'\):} Since \(\psi\) is of the existential kind, it can have at most one free variable, and by definition, \(\psi'\) is of the universal kind.
          Let this free variable be \(x_i\). 
          There exists a value \(u_2\) of \(D_2\) such that \((D_2,v_2^{(i)},u_2)\) does not satisfy \(\psi'(x_i, x_{n+1})\).
          But for all values \(u_1\) of \(D_1\), \((D_1,v_1^{(i)},u_1)\) is a model of \(\psi'\), and by induction hypothesis, Spoiler wins the \(\ell\)-round UQAFO game when playing on both databases where \(s = 2\).
          Therefore, the winning strategy for Spoiler in the \(\ell+1\) round game on databases \((D_1,v_1^{(1)},\ldots,v_1^{(n)})\) and \((D_2, v_2^{(1)},\ldots, v_2^{(n)})\) where \(s = 1\) is to change the database they are playing in, choosing \((v_1^{(i)}, v_2^{(i)})\) to be remembered, and then choosing \(u_2\).
  \end{itemize}

  To prove the case in which \(\varphi\) is of the universal kind, we can use Lemma \ref{lemma:uqafo_neg_exis_univ_swap} to get a formula \(\varphi'\) which is of the existential kind and which is satisfied by \(D_2\) but not by \(D_1\).
  Using the above argument, we can conclude that Spoiler has a winning strategy on \(D_1\) and \(D_2\) where \(s = 2\).
\end{proof}

To prove that UQAFO formulas cannot define PQR-orders, we define for each natural number \(n\in\mathbb{N}\) two graphs \(G^n_1, G_2^n\):
\(G^n_1\) is a PQR-order that contains \(3\cdot 2^{n+1}+9\) vertices and where the distinguished vertex is the first element in the order.
\(G^n_2\) is defined similarly, but it is missing the \(2^n+2\)-nd vertex that satisfies \(R\).
\(G_1^n\) and \(G_2^n\) are visualized in Figure~\ref{fig:uqafo_indistinguish_graphs}.
\begin{figure}
    \centering

    \begin{tikzpicture}[node distance = 0.65cm]
        \node[circle] (a3) {$\bullet^{P}$};
        \node[circle, left=0.2em of a3] (a2) {$G_1^{v_1}$:}; 
        \node[left= -1em of a3] {\(^{v_1}\)};
        \node[circle, right = of a3] (a4) {$\bullet^{Q}$};
        \node[circle, right = of a4] (a5) {$\bullet^{R}$};
        \node[circle, right = of a5] (a6) {$\bullet^P$};
        \node[circle, right = of a6] (a6b) {$\bullet^Q$};
        \node[circle, right = of a6b] (a7) {{$\bullet^R$}};
        \node[circle, right = of a7] (a9) {$\bullet^P$};
        \node[circle, right = of a9] (a10) {$\bullet^Q$};
        \node[circle, right = of a10] (a11) {$\bullet^R$};

         \path [draw=black,  ->, ] (a3) -- (a4);
         \path [draw=black,  ->, ] (a4) -- (a5);
         \path [draw=black,  ->, ] (a5) -- (a6);
         \path [draw=black,  ->, ] (a6) -- (a6b);
         \path [draw=black,  ->, ] (a6b) -- (a7);
         \path [draw=black,  ->, ] (a7) -- (a9);
         \path [draw=black,  ->, ] (a9) -- (a10);
         \path [draw=black,  ->, ] (a10) -- (a11);

        \draw [decorate, decoration={calligraphic brace, amplitude=5pt}]  (a5.270) -- (a3.270) node[pos=0.5, below=5pt] {$2^n+1$ times};
        \draw [decorate, decoration={calligraphic brace, amplitude=5pt}]  (a11.270) -- (a9.270) node[pos=0.5, below=5pt] {$2^n+1$ times};

        \node[circle, below =4em of a3] (b3) {$\bullet^{P}$};
        \node[circle, left=0.2em of b3] (b2) {$G_2^{v_2}$:}; 
        \node[left= -1em of b3] {\(^{v_2}\)};
        \node[circle, right = of b3] (b4) {$\bullet^{Q}$};
        \node[circle, right = of b4] (b5) {$\bullet^{R}$};
        \node[circle, right = of b5] (b6) {$\bullet^P$};
        \node[circle, right = of b6] (b7) {$\bullet^Q$};
        \node[circle, right = of b7] (b6b) {\phantom{$\bullet^R$}};
        \node[circle, right = of b6b] (b9) {$\bullet^P$};
        \node[circle, right = of b9] (b10) {$\bullet^Q$};
        \node[circle, right = of b10] (b11) {$\bullet^R$};

        \path [draw=black, ->, ] (b3) -- (b4);
        \path [draw=black, ->, ] (b4) -- (b5);
        \path [draw=black, ->, ] (b5) -- (b6);
        \path [draw=black, ->, ] (b6) -- (b7);
        \path [draw=black, ->, ] (b7) -- (b9);
        \path [draw=black, ->, ] (b9) -- (b10);
        \path [draw=black, ->, ] (b10) -- (b11);

        \draw [decorate, decoration={calligraphic brace, amplitude=5pt}]  (b5.270) -- (b3.270) node[pos=0.5, below=5pt] {$2^n+1$ times};
        \draw [decorate, decoration={calligraphic brace, amplitude=5pt}]  (b11.270) -- (b9.270) node[pos=0.5, below=5pt] {$2^n+1$ times};

    \end{tikzpicture}

    \caption{The graphs \(G_1^n\) (above) and \(G_2^n\) (below). We omit reflexive and transitive edges.
    }
    \label{fig:uqafo_indistinguish_graphs}
\end{figure}
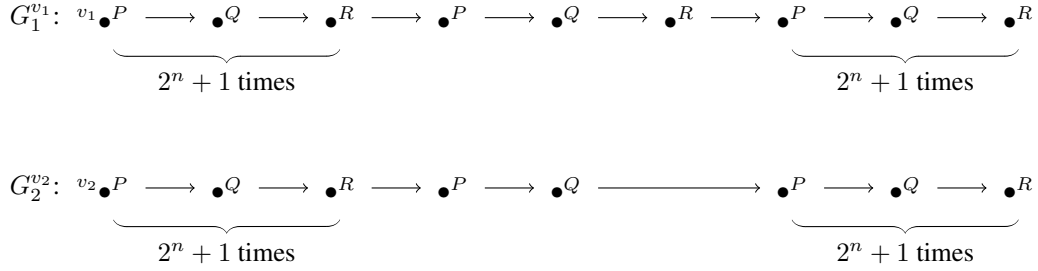
In \(G_1^n\) we denote the vertices satisfying \(P\) by \(p_1^{1},\ldots,p_1^{2^{n+1}+3}\).
Analogously, we denote the vertices satisfying \(Q\) and \(R\) by \(q_1^{1},\ldots,q_1^{2^{n+1}+3}\) and \(r_1^{1},\ldots,r_1^{2^{n+1}+3}\) respectively.
The naming of the vertices in \(G_2^n\) is analogous, but we skip \(r_2^{2^n+2}\),
that is, we denote them by \(r_2^{1},\ldots,r_2^{2^n+1}\) and \(r_2^{2^n+3},\ldots,r_2^{2^{n+1}+3}\).

We will now show that Duplicator can win the \(n\) round UQAFO game on \((G_1^n,v_1)\) and \((G_2^n,v_2)\).
We will show that Duplicator can choose vertices such that all configurations satisfy some properties depending on the number \(\ell\) of rounds that are left to play.
We call such configurations \(\ell\)-valid.

\begin{definition}[\(\ell\)-valid]\label{def:ell-valid}
We define \(\ind(x_i^k) = k\) and \(\pred(x_i^k) = x\) for each vertex \(x_i^k\) in \(G_i\) with \(x\in \{p,q,r\}\).
Let \((G_1^n,v_1^{(1)},\ldots,v_1^{(m)})\), \((G_2^n,v_2^{(1)}, \ldots,v_2^{(m)})\), and index \(s\) be a configuration and let \(0\leq \ell\leq n\).
We call this configuration \(\ell\)-valid if it satisfies the following conditions:
\begin{enumerate}
  \item \label{item:uqafo_game_iso}The set of pairs describes a partial isomorphism, that is, the vertices in each pair satisfy the same unary relations and the sets of played vertices is ordered the same in both graphs.
        Formally  \(\pred(v_1^{(i)}) = \pred(v_2^{(i)})\) for each \(1\leq i\leq m\) and \(\ind(v_1^{(i)})\leq \ind(v_1^{(j)})\leftrightarrow \ind(v_2^{(i)})\leq \ind(v_2^{(j)})\) for each \(1\leq i , j\leq m\).
  \item\label{item:uqafo_game_end}
        If the configuration contains a vertex near the end of one of the graphs, then the corresponding vertex in the other graph has the same position.
        Formally, for all \(1\leq i \leq m\): if 
        \(\ind(v_1^{(i)})\leq 2^\ell\) or \(\ind(v_1^{(i)})\geq 2^{n+1}-2^\ell+4\) or \(\ind(v_2^{(i)})\leq 2^\ell\) or \(\ind(v_2^{(i)})\geq 2^{n+1}-2^\ell+4\), then \(\ind(v_1^{(i)})=\ind(v_2^{(i)})\).
  \item\label{item:uqafo_game_mid}
        There are enough vertices satisfying \(P,Q,\) and \(R\) between played vertices in \(G_{3-s}\).
        Formally, for all \(1\leq i , j \leq m\) such that \(\ind(v_s^{(i)}) > \ind(v_s^{(j)})\), \(v_{3-s}^{(i)}\) and \(v_{3-s}^{(j)}\) satisfy
        \[\ind(v_{3-s}^{(i)})-\ind(v_{3-s}^{(j)}) \geq \min (2^\ell, \ind(v_s^{(i)})-\ind(v_s^{(j)})) +c,\]
        where \[c = \begin{cases}
          1&\qquad\text{ if } \ind(v_{3-s}^{(j)})\leq 2^n+1 \text{ and } \ind(v_{3-s}^{(i)}) > 2^n+1,\\
          0&\qquad\text{ otherwise.}
        \end{cases}\]

        Intuitively, the constant \(c\) ensures that Duplicator can skip vertices with index \(2^n+1\) and thus it is not an issue that \(r^{2^n+1}_2\) is missing.
        Strictly speaking, the constant \(c\) is only necessary if Spoiler plays in \(G_1^n\), but to avoid superfluous case distinctions, we also require it if \(s = 2\).
\end{enumerate}
\end{definition}

We first notice that each \(\ell\)-valid configuration can be extended by the first and last vertex of \(G_1^n\) and \(G_2^n\).
\begin{lemma}\label{lemma:uqafo_config_add_both_ends}
  Let \(\ell\leq n\) and let \(\mathcal{C} = ((G_1^n,v_1^{(1)},\ldots,v_1^{(m)}), (G_2^n,v_2^{(1)}, \ldots,v_2^{(m)}), s)\) be an \(\ell\)-valid configuration.
  Then the configuration consisting of \((G_1^n,v_1^{(1)},\ldots,v_1^{(m)}, p_1^1, r_1^{2^{n+1}+3})\) and \((G_2^n,v_2^{(1)}, \ldots,v_2^{(m)}, p_2^1, r_2^{2^{n+1}+3})\) with the same index \(s\) is also \(\ell\)-valid.
\end{lemma}
\begin{proof}
  Since \(\mathcal{C}\) is \(\ell\)-valid, Conditions \ref{item:uqafo_game_iso} and \ref{item:uqafo_game_end} can be verified easily.
  For Condition~\ref{item:uqafo_game_mid}, we only have to check the newly introduced pairs.
  \(\ind(p_{3-s}^1)-\ind(r_{3-s}^{2^{n+1}+3})\geq 2^\ell + 1\) is true since \(\ell \leq n\).
  We now argue that this condition is satisfied for \(p_{3-s}^1\) and each \(v_{3-s}^{(i)}\).
  The arguments for \(r_{3-s}^{2^{n+1}+3}\) are analogous.
  If \(\ind(v_{3-s}^{(i)})\geq 2^\ell+1\), it is easy to verify that Condition~\ref{item:uqafo_game_mid} is satisfied.
  Otherwise, \(\ind(v_{3-s}^{(i)}) \leq 2^\ell\) and by Condition~\ref{item:uqafo_game_end}, \(\ind(v_s^{(i)}) = \ind(v_{3-s}^{(i)})\).
  Therefore, Condition~\ref{item:uqafo_game_mid} is also satisfied in this case.
  \end{proof}

\begin{theorem}\label{thm:p_not_expressible_in_uqafo}
 Duplicator has a winning strategy in the \(n\) round UQAFO game on graphs \((G_1^n, v_1)\) and \((G_2^n, v_2)\).
\end{theorem}
\begin{proof}
We show that Duplicator can play such that the \((n-\ell)\)-th configuration is \(\ell\)-valid.
The theorem follows, since Condition~\ref{item:uqafo_game_iso} ensures that Spoiler does not win.

At the beginning of the game, that is, for \(\ell = n\), the initial configuration consists of the two graphs \((G_1, v_1)\), \((G_2,v_2)\), and index \(s\), where \(s\in \{1,2\}\) is chosen by Spoiler.
For both choices of \(s\), the configuration is \(n\)-valid.

For \(0<\ell+1 \leq n\) let \(\mathcal{C} = ((G_1,v_1^{(1)},\ldots,v_1^{(m)}),(G_2,v_2^{(1)}, \ldots,v_2^{(m)}),s)\) be a \((\ell+1)\)-valid configuration.
We show that for each move of Spoiler, Duplicator has a response such that the next configuration is \(\ell\)-valid.
We distinguish the two possible options for Spoiler:

  If Spoiler plays in \(G_s^n\), we distinguish the following cases:
        \begin{itemize}
          \item If \emph{Spoiler chooses \(x_s^i\), where \(x\in \{p,q,r\}\) and \(i \leq 2^{\ell}\) or \(i\geq 2^{n+1}-2^{\ell}+4\)},
                then Duplicator responds with \(x_{3-s}^i\).
                This vertex exists, since \(i\neq 2^n+2\) and \(r_1^{2^n+2}\) is the only vertex that has no copy in \(G_2^n\).
                Condition \ref{item:uqafo_game_end} is obviously satisfied.
                Condition \ref{item:uqafo_game_iso} follows from Condition~\ref{item:uqafo_game_end} since it requires matching indices near the ends of the graphs and it implies \(\ind(v_s^{(i)}) < 2^{n+1}-2^\ell+4\leftrightarrow \ind(v_{3-s}^{(i)}) < 2^{n+1}-2^\ell+4\) and \(\ind(v_s^{(i)}) > 2^\ell\leftrightarrow \ind(v_{3-s}^{(i)}) > 2^\ell\).
                We now argue that Condition~\ref{item:uqafo_game_mid} is satisfied if \(i\leq 2^\ell\). The arguments are analogous if \(i\geq 2^{n+1}-2^\ell+4\).
                Since \(\mathcal{C}\) is \(\ell+1\)-valid by induction hypothesis, it suffices to show that this condition is satisfied for \(x_{3-s}^i\) and each \(v_{3-s}^{(j)}\).
                Let \(1\leq j\leq m\).
                If \(\ind(v_{3-s}^{(j)}) < i\), then Condition~\ref{item:uqafo_game_mid} immediately follows from Condition~\ref{item:uqafo_game_end}.
                Otherwise, if \(\ind(v_{3-s}^{(j)})-i \geq 2^\ell+1\), Condition~\ref{item:uqafo_game_mid} is obviously satisfied.
                Otherwise, \(\ind(v_{3-s}^{(j)})\leq i + 2^\ell \leq 2^{\ell+1}\) and by induction hypothesis and Condition~\ref{item:uqafo_game_end}, \(\ind(v_{s}^{(j)}) = \ind(v_{3-s}^{(j)})\).
                Since \(\ell+1\leq n\), this can be used to verify Condition~\ref{item:uqafo_game_mid} easily.
            \item
                If \emph{Spoiler chooses \(x^i_s\), where \(2^{\ell} < i < 2^{n+1}-2^{\ell}+4\)}, then Duplicator responds as follows.
                By Lemma \ref{lemma:uqafo_config_add_both_ends} we can assume that \((p_1^1, p_2^1)\) and  \((r_1^{2^{n+1}+3}, r_2^{2^{n+1}+3})\) are in \(\mathcal{C}\).
                W.l.o.g. we can also assume that Spoiler does not play a vertex already in some pair of \(\mathcal{C}\).
                Let \(v_s^{(j)}\) be the played vertex left of \(x_s^i\), that is, such that \(\ind(v_s^{(j)})\) is maximal with \(\ind(v_s^{(j)}) < i\).
                Duplicator now responds with \(x_{3-s}^{i'}\) such that \(i' = \ind(v_{3-s}^{(j)}) + \min(2^\ell, i- \ind(v_s^{(j)}))+c  \) where \(c\) is
                \[c = \begin{cases}
          1&\quad\text{if } \ind(v_{3-s}^{(j)})\leq 2^n+1 \text{ and } \ind(v_{3-s}^{(j)}) + \min(2^\ell, (i - \ind(v_s^{(j)}))) > 2^n +1,\\
          0&\quad\text{otherwise.}
        \end{cases}\]

                We now show that this configuration is \(\ell\)-valid.
                We start with Condition~\ref{item:uqafo_game_end} by arguing that \(2^\ell < i' < 2^{n+1}-2^\ell+4\).
                First, assume to the contrary that \(i'\leq 2^\ell\).
                Then \(i' = \ind(v_{3-s}^{(j)}) + i - \ind(v_s^{(j)})\) and \(\ind(v_{3-s}^{(j)}) \leq 2^\ell\) and by induction hypothesis \(\ind(v_{3-s}^{(j)}) = \ind(v_{s}^{(j)})\).
                Thus, \(i' = i\) which is a contradiction to \(i > 2^\ell\).
                Now assume \(2^{n+1}-2^\ell+4 \leq i'\).
                Thus, \(\ind(v_{3-s}^{(j)}) \geq 2^{n+1}-2^{\ell+1} + 4-c\), where \(c = 0\), because \(2^{n+1}-2^{\ell+1}+3 \geq 2^n + 1\).
                Again, by induction hypothesis, \(\ind(v_{3-s}^{(j)}) = \ind(v_s^{(j)})\) and \(i=i'\).
                
                For Conditions~\ref{item:uqafo_game_iso} and \ref{item:uqafo_game_mid}, let \(v_{s}^{(k)}\) be the vertex right to \(x^i_s\).
                We show that \(i' < \ind(v_{3-s}^{(k)})\) and \(\ind(v_{3-s}^{(k)})- i' \geq \min(2^\ell, \ind(v_s^{(k)})-i) + c'\), where \(c'\) is defined as in Definition~\ref{def:ell-valid}.
                 Obviously, the former follows from the latter, since \(\min(2^\ell, \ind(v_s^{(k)})-i)+c'\) is at least \(0\).
                As a first case, assume that \(\ind(v_{3-s}^{(j)})\leq 2^n+1\) while \(\ind(v_{3-s}^{(k)})> 2^n+1\).
                By induction hypothesis, \[\ind(v_{3-s}^{(k)})-\ind(v_{3-s}^{(j)}) \geq \min (2^{\ell+1}, \ind(v_s^{(k)})-\ind(v_s^{(j)}))+1.\] 
                Thus,
                \begin{align*}
                  \ind(v_{3-s}^{(k)})- i' &= \ind(v_{3-s}^{(k)}) - \ind(v_{3-s}^{(j)}) - \min(2^\ell, i-\ind(v_s^{(j)})) - c\\
                  &\geq \min(2^{\ell+1}, \ind(v_s^{(k)})- \ind(v_s^{(j)}))-\min(2^\ell, i-\ind(v_s^{(j)})) - c + 1\\
                  &\geq \min(2^{\ell}, \ind(v_s^{(k)}) - i) + (1-c). 
                \end{align*}
                By definition, \(c = 1\) iff \(i' > 2^n+1\), thus by choosing \(c' = (1-c)\) both conditions are satisfied.\\
                As the second case, assume that \(\ind(v_{3-s}^{(j)})\leq 2^n+1\) and \(\ind(v_{3-s}^{(k)}) \leq 2^n+1\), thus \(c=0\).
                By induction hypothesis, \[\ind(v_{3-s}^{(k)})-\ind(v_{3-s}^{(j)}) \geq \min (2^{\ell+1}, \ind(v_s^{(k)})-\ind(v_s^{(j)})).\]
                Therefore,
                \begin{align*}
                  \ind(v_{3-s}^{(k)})- i' &= \ind(v_{3-s}^{(k)}) - \ind(v_{3-s}^{(j)}) - \min(2^\ell, i-\ind(v_s^{(j)}))\\
                  &\geq \min(2^{\ell+1}, \ind(v_s^{(k)})- \ind(v_s^{(j)}))-\min(2^\ell, i-\ind(v_s^{(j)}))\\
                  &\geq \min(2^{\ell}, \ind(v_s^{(k)}) - i). 
                \end{align*}
                The third case, \(\ind(v_{3-s}^{(j)}),\ind(v_{3-s}^{(k)}) > 2^n+1\) is analogous.
        \end{itemize}
  If Spoiler plays in \(G_{3-s}^n\), they first change the configuration to some \(\mathcal{C}' = ((G_1,v^{(i)}_1),(G_2, v^{(i)}_2), {3-s})\). We show that \(\mathcal{C}'\) is \(\ell+1\)-valid. We can then apply the argumentation above.
        Conditions \ref{item:uqafo_game_iso} and \ref{item:uqafo_game_end} immediately carry over to \(\mathcal{C}'\).
        Condition \ref{item:uqafo_game_mid} is also satisfied, since it is only relevant for configurations in which at least two pairs were played.

\end{proof}

\begin{corollary}
    There does not exist a UQAFO formula that accepts exactly the pointed graphs that are PQR-orders.
\end{corollary}
\begin{proof}
  By Theorem~\ref{thm:p_not_expressible_in_uqafo} and Lemma~\ref{lem:uqafo_games_correspondence}, for each \(n\in\mathbb{N}\) there exists no UQAFO formula with quantifier depth \(n\) that distinguishes \(G_1^n\) and \(G_2^n\).
  Since all \(G_1^n\) are PQR-orders and all \(G_2^n\) are not PQR-orders, there exists no UQAFO formula that accepts exactly PQR-orders.
\end{proof}

\subsection{Points~3 and~4 of Theorem~\ref{thm:DENcombi}}

Point 3 follows from the homomorphism–embedding count translations in Theorem~\ref{lovasz homomorphism count implies embedding count}, applied layer-wise to sum aggregation: every embedding query can be simulated by finitely many homomorphism queries and conversely every homomorphism query can be simulated by finitely many embedding queries, with the required linear combinations absorbed into the combination functions.
For Point~4, $\maxagg$-DENs are at least as expressive as $\maxagg$-DHNs because every HML formula is an EML formula. Strictness follows from local transitivity: by Point~1 of Theorem~\ref{thm:DENcombi}
and Lemma~\ref{lem:UQAFOUNFOC}, it is expressible by a max-DEN, while by Point~1 of Theorem~\ref{thm:combi}, $\maxagg$-DHNs have the same expressive power as HML, and local transitivity is not expressible there.

\subsection{Proof of Theorem \ref{thm:sun-expressivity}: Separating \texorpdfstring{$\sumagg$-DHNs}{sum-DHNs} from GNNs with Homomorphism Counts}

We prove that, already over undirected unembedded graphs, simple $\sumagg$-DHNs are strictly more expressive than homomorphism-count enriched GNNs, by showing that the EML formula $\varphi_{\sun}$ is not expressed by a GNN with homomorphism count features and global readout:
\begin{align*}
\varphi_{\sun}(x_1) =
\exists x_2\cdots x_6 \,\Big(\bigwedge_{1 \leq i \leq 6}\!(E(x_i,x_{(i\bmod 6)+1}) \wedge \exists y\,(E(x_i,y) \wedge \exists^{=1} z\,E(y,z)) 
\land\!\!\!\bigwedge_{i < j\leq 6}\!\!\! x_i\neq x_j)
\Big)
\end{align*}
In this section we consider undirected graphs, i.e. databases over schema $\{E\}$ with $\mn{ar}(E)=2$ such that the relation $E$ is irreflexive and symmetric. Given a graph database $D$ we denote the set of vertices $\mn{adom}(D)$ with $V(D)$ and the edge relation as $E(D)$. We use the constant embedding $\lambda_0$ onto $(0)$.
\begin{definition}
Let $P$ be a set of undirected graphs. We say $P$ \emph{is determined by finitely many homomorphism counts} if there are  graphs $F_1, \dots, F_n$ ($n \geq 1$) such that for all graphs $G$ and $H$,
$|\Hom(F_i,G)| = |\Hom(F_i, H)|$ for all $i\leq n$ implies that 
$G \in P \Leftrightarrow H \in P$.

Similarly, a value property of pointed databases $P$ is determined by finitely many homomorphism counts if there are pointed graphs $F^\bullet_1,\dots,F^\bullet_n$ ($n \geq 1$) such that for all pointed graphs $G^\bullet$ and $H^\bullet$,
$|\Hom(F^\bullet_i,G^\bullet)| = |\Hom(F^\bullet_i, H^\bullet)|$ for all $i\leq n$ implies that
$G^\bullet \in P \Leftrightarrow H^\bullet \in P$.

We will refer to the $F_i$ (respectively, $F^\bullet_i$) as pattern graphs.
\end{definition}
We use the following result by Chen et al.:
\begin{theorem}[\cite{chen2025algorithms}, Theorem 6.1]
\label{thm:no_hom_alg_for_isolated}
The set $P_{iso}$ of graphs that include an isolated vertex is not determined by finitely many homomorphism counts.
\end{theorem}

We define GNNs as specific instances of DHNs:
\begin{definition}[GNN]
A GNN with global readout $\calN = (\calL_1, \dots, \calL_k)$ is a DHN over the graph schema $\{E\}$, where for $i = 1, \dots k-1$, $\calL_i = (\calF, \com)$ has two homomorphism queries $(F^\bullet_1, \mu, \agg), (F^\bullet_2, \mu, \agg)$ such that $F^\bullet_1$ is a single vertex $\bullet$, $F^\bullet_2$ is a single edge $\bullet - \circ$ and $\mu$ is the identity function. $\calL_k$ has the same two homomorphism queries and a third query $(F^\bullet_3,\mu)$, where $F^\bullet_3$ consists of two unconnected vertices $\bullet \,\,\, \circ$, and $\mu$ is again the identity function.
\end{definition}
We use a characterization of the expressive power of GNNs in terms of homomorphism counts from trees. Given embedded graphs $(G^\bullet,\lambda_G)$, $(H^\bullet,\lambda_H)$ let $\Hom((G^\bullet,\lambda_G), (H^\bullet,\lambda_H))$ be the set of homomorphisms $h \in \Hom(G^\bullet, H^\bullet)$ such that for all $u \in V(G^\bullet)$, $\lambda_H(h(u))=\lambda_G(u)$:
\begin{theorem}
\label{thm:GNN_hom_counts}
    Let $(G,\lambda_G),(H,\lambda_H)$ be embedded graphs and $u \in V(G), v \in V(H)$.
    The following are equivalent:
    \begin{enumerate}
        \item For every GNN with global readout $\calN$, $\calN(G,\lambda_G)(u) = \calN(H, \lambda_H)(v)$.
        \item For every embedded pointed tree $(F^\bullet,\lambda)$ the following two equalities hold:
        \begin{itemize}
            \item 
        $|\Hom((F^\bullet,\lambda),(G^u,\lambda_G))| = |\Hom((F^\bullet,\lambda),(H^v,\lambda_H))|$
        \item
        $|\Hom((F,\lambda),(G,\lambda_G))| = |\Hom((F,\lambda),(H,\lambda_H))|$
        \end{itemize}
    \end{enumerate}
\end{theorem}
This result stems from a characterization of the separating power of GNNs as that of the Weisfeiler-Leman algorithm \cite{morris2019weisfeiler,xu2019powerful}, and a characterization of the latter as distinguishability under homomorphism counts from trees (see \cite{dell2018Lovasz,dvovrak2010recognizing}, and an explicit proof for pointed homomorphisms in \cite{grohe2020word2vec}). 

A GNN with global readout and homomorphism count features can be equivalently defined as a DHN of which the first layer has homomorphism queries $((F^\bullet_1, \mu, \sumagg), \dots, (F^\bullet_n, \mu, \sumagg))$ where $\mu$ is the identity function, the $F^\bullet_i$ are arbitrary pattern graphs and the subsequent layers constitute a GNN with global readout. We now show that such GNNs do not express $\varphi_{\sun}$. Let $P_{\sun}$ be the value property that contains $G^u$ if and only if $G \models \varphi_{\sun}(u)$:

\begin{figure}
    \centering
    \includegraphics[width=0.5\linewidth]{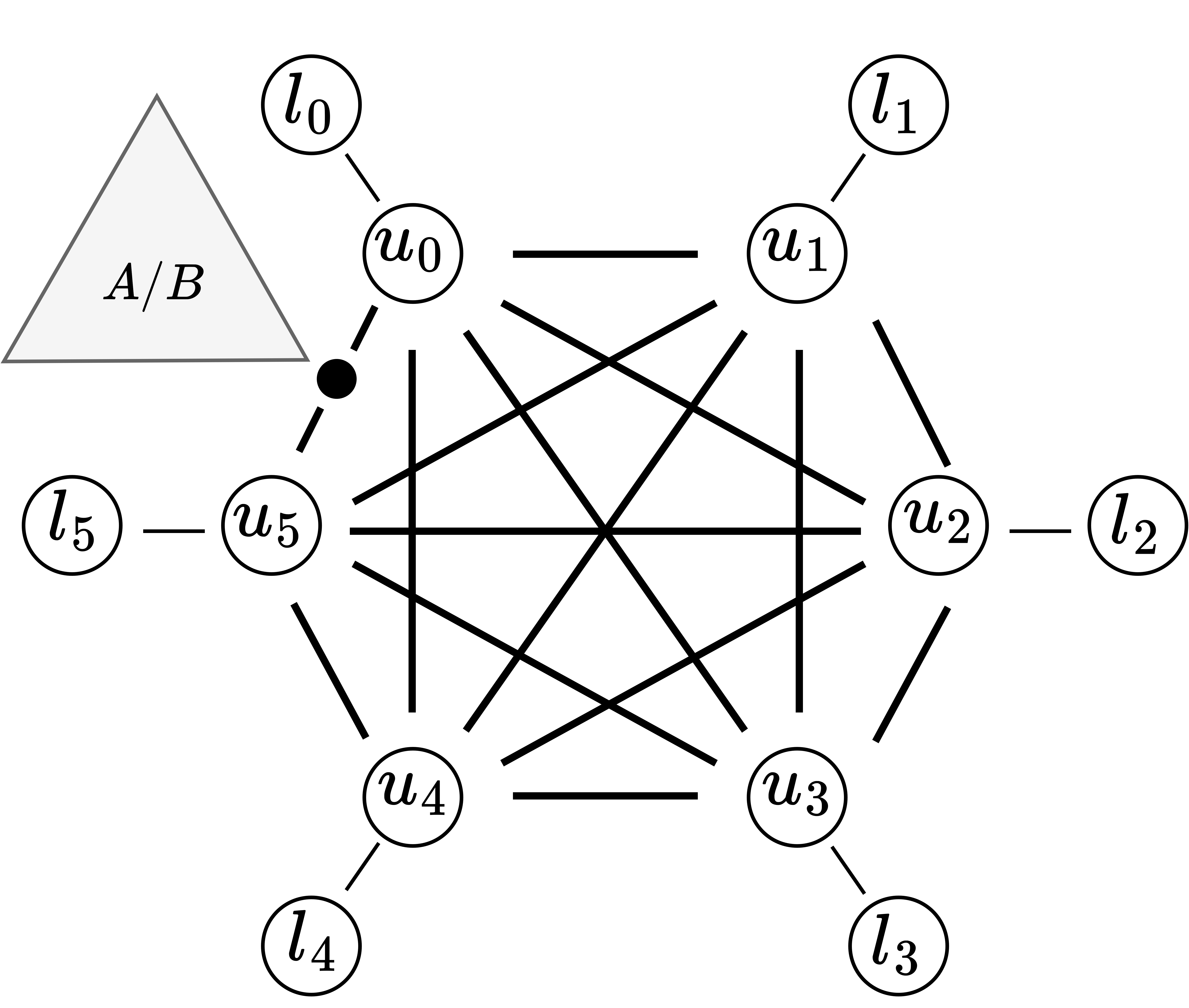}
   \caption{In the skeleton graph $S^{u_0}$ each clique vertex $u_i$ is connected to a leaf vertex $l_i$ and each pair of clique vertices $u_i,u_j$ is connected via a connecting vertex $v_{ij}$. Each connecting vertex is connected to all nodes of a copy of $A$ or $B$ to construct $G^{u_0}$ or $H^{u_0}$.}
   \label{fig:skeleton}
\end{figure}
\begin{theorem}
\label{thm:GNN_does_not_express_phisun}
For every
tuple of pattern graphs $\calK = (F^\bullet_1, \dots, F^\bullet_n)$ there are pointed graphs $G^u \in P_{\sun}$ and $H^v \not\in P_{\sun}$ 
such that for all GNNs $\calN$ with global readout and homomorphism counts from $\calK$, $\calN(G^u, \lambda_0) = \calN(H^v,\lambda_0)$. Thus, no $\calK$-enriched GNN expresses $P_{\sun}$.
\end{theorem}
\begin{proof}    
By Theorem \ref{thm:no_hom_alg_for_isolated} there exist graphs $A$ and $B$ where $A$ has an isolated vertex, $B$ does not, and for every $F^\bullet_i \in \calK:$ and $V\subseteq V(F_i)$, $|\Hom(F_i|V,A)|=|\Hom(F_i|V,B)|$.
We define a skeleton graph $S^{u_0}$ depicted in Figure \ref{fig:skeleton}, consisting of \emph{clique vertices} $u_0, \dots, u_5$, \emph{leaf vertices} $l_0, \dots l_5$ and \emph{connecting vertices} $v_{ij}$ for each $i,j=0,\dots,5$ such that $i<j$. $S^{u_0}$ has edges $(u_i,l_i)$,$(u_i, v_{ij}),(u_j,v_{ij})$ for each $u_i,u_j$. We now construct $G^{u_0}$ and $H^{u_0}$ by attaching to each connecting vertex a copy of either $A$ or $B$, by adding edges from the connecting vertex to all vertices in the copy of $A$ or $B$.
For $G^{u_0}$ divide the clique vertices into  $(u_0,u_1,u_2)$ and $(u_3,u_4,u_5)$, and add $A$ to $v_{ij}$ if $u_i,u_j$ are in the same triple, and $B$ otherwise. For $H^{u_0}$ add $A$ to $v_{ij}$ if $j=i+1$ or $i=0,j=5$, and $B$ otherwise. We show that homomorphism counts from $\calK$ to vertices in $V(S)$ are equal for $G$ and $H$. More precisely:
\begin{itemize}
    \item[(\(*\))] Let $s_1,s_2 \in V(S)$ be both clique vertices, both leaf vertices, or both connecting vertices, and let $F^\bullet_i \in \calK$. Then:
    \begin{align*}
        |\Hom(F^\bullet_i,G^{s_1})| = |\Hom(F^\bullet_i,H^{s_2})|.
    \end{align*}
\end{itemize}
For some pointed subgraph $F^\bullet_S$ of $F^\bullet_i$ and some $s \in V(S)$ let $h_S \in \Hom(F^\bullet_S,S^s)$. Let $C_1,\dots C_m$ be the connected components in $V(F^\bullet_S)\setminus V(F^\bullet_i)$. Suppose $h_S$ is extended to a $h \in \Hom(F^\bullet_i|(V(F^\bullet_S) \cup V(C_j)),G^s)$ for some $j \leq m$ such that the vertices in $C_j$ are mapped outside of $S^s$. 
Since $C_j$ is a connected component it must be mapped homomorphically by $h$ to a copy of either $A$ or $B$. If $C_j$ is not connected to $F_S$, every such map is an extension of $h$. Otherwise, if $C_j$ is connected to vertices in $F_S$, they must all be mapped by $h_S$ to a single vertex $s_j$ in $S$. Then $h$ is a homomorphism extending $h_S$ if and only if it homomorphically maps $C_j$ onto a single copy of $A$ or $B$, depending on whether $s_j$ is attached to $A$ or $B$. The total number of such homomorphisms $h$ extending $h_S$ is then:
\begin{align*}
 &\prod\llbrace 6\cdot |\Hom(C_j,A)| + 9\cdot |\Hom(C_j,B)| \,\,\mid\,\, j = 1, \dots, m \text{ and }C_j \text{ is not connected to } F_S\rrbrace\\  
 &\qquad \cdot \prod\llbrace|\Hom(C_j, A)| \,\,\mid\,\, j = 1, \dots, m \text{ and }C_j \text{ is connected to } F_S \text{ and } s_j \text{ is attached to $A$}\rrbrace\\
 &\qquad \cdot \prod\llbrace|\Hom(C_j, B)| \,\,\mid\,\, j = 1, \dots, m \text{ and }C_j \text{ is connected to } F_S \text{ and } s_j \text{ is attached to $B$}\rrbrace
\end{align*}
Since $|\Hom(C_j, A)| = |\Hom(C_j, B)|$ for each induced subgraph $C_j$ of $F^\bullet_i$,
the number of extensions of $h_S$ to a homomorphism $h$ from $F^\bullet_i$ to $G^{s}$ only depends on the homomorphisms onto $S^{s}$. The same exact argument applies for homomorphisms to $H^s$. 
Since permutations between pairs of clique vertices, connecting vertices or leaf vertices are automorphisms on $S$ this proves (\(*\)). Analogous reasoning shows:
\begin{itemize}
    \item[(\(**\))] Let $c \in V(A) \cup V(B)$ and let $F^\bullet_i \in \calK$. Then:
    \begin{align*}
    |\Hom(F^\bullet_i,G^{c})| = |\Hom(F^\bullet_i,H^{c})|
    \end{align*}
\end{itemize}
where now if $c \in V(A)$, instead of counting the number of extensions $h \in \Hom(F^\bullet_i, G^\bullet)$ of each homomorphism $h_S \in \Hom(F^\bullet_S, S^\bullet)$ to a homomorphism, we now count the number of extensions $h \in \Hom(F^\bullet_i, G^\bullet)$ of each homomorphism  $h_{S\cup A}$ from a subgraph $F^\bullet_{S \cup A}$ of $F^\bullet_i$ onto $S$ with a single copy of $A^c$ attached, and analogously if $c \in V(B)$.

We now argue that from ($*$) and ($**$) the theorem follows. Firstly, $G^{u_0} \in P_{\sun}$ since $u_0$ is on a $6$ cycle of vertices with degree $1$ neighbors in $G$: $u_0, v_{01}, u_1, v_{12}, u_2, v_{02}, u_0$, since $v_{01}, v_{12}, v_{02}$ are all connected to a copy of $A$, which has an isolated vertex. Further, $H^{u_0} \not\in P_{\sun}$, since every $6$-cycle from $u_0$ with $6$ distinct vertices passes a connecting vertex attached to $B$.
Let $\calN$ be a GNN with homomorphism counts from $\calK$. 
Suppose $s_1 \in V(G), s_2 \in V(H)$ are both clique vertices, leaf vertices, connecting vertices both attached to $A$ or both attached to $B$, or the same vertices in two copies of either $A$ or $B$. Then by ($*$) and ($**$), the first layer of $\calN$ assigns the same embedding to $s_1$ and $s_2$. 
Let $\lambda_G,\lambda_H$ be the embeddings generated by the first layer of $\calN$, then if $T^r$ is a depth $0$ tree, $|\Hom((T^r,\lambda), (G^{s_1},\lambda_G))| = |\Hom((T^r,\lambda), (G^{s_2},\lambda_H))|$. Now suppose this holds for trees of depth $d\geq 0$ and let $(T^r,\lambda)$ be a embedded tree of depth $d+1$ with subtrees from children $T^{c_1}_1, \dots, T^{c_m}_m$. Then since $s_1,s_2$ have the same number of neighbors that are clique vertices, leaf vertices, connecting vertices attached to $A$ or connecting vertices attached to $B$, and since they are connected to the same vertices in $A$ and $B$:
\begin{align*}
    |\Hom((T^r,\lambda), (G^{s_1},\lambda_G))| &= \prod_{i=1}^m \left( \sum_{s \in \calN(s_1)} |\Hom((T^{c_i}_{i}|(V(T)\setminus \{r\}),\lambda), (G^s, \lambda_G))|\right)\\
    &= \prod_{i=1}^m \left(\sum_{s \in \calN(s_2)} |\Hom((T^{c_i}_{i}|(V(T)\setminus \{r\}),\lambda), (H^s, \lambda_H))|\right)\\
    &= |\Hom((T^r,\lambda), (H^{s_2},\lambda_H))|
\end{align*}
Where the second equality holds since $(G^s,\lambda_G)$ and $(H^s,\lambda_H)$ are indistinguishable by trees of depth $d$. By induction $(G^{s_1},\lambda_G)$ and $(H^{s_2},\lambda_H)$ are indistinguishable by homomorphism counts from embedded trees. Using Theorem \ref{thm:GNN_hom_counts} then for every GNN $\calN$ with homomorphism counts from $\calK$, $\calN(G^{u_0},\lambda_0) = \calN(H^{u_0},\lambda_0)$, where $\lambda_0$ is the constant embedding onto $(0)$.
\end{proof}

Theorem \ref{thm:GNN_does_not_express_phisun} gives an example of a value property 
on undirected unembedded graphs that is expressed by a connected $\sumagg$ DHN, but not by a GNN with any aggregation and combination function, global readout and homomorphism count features. This proves Theorem \ref{thm:sun-expressivity}. As an aside, we note that, although in this paper we study expressive power only at the level of node classifiers, the same argument also shows that
the graph property $\exists x_1\varphi_{\sun}(x_1)$ is not expressible by a GNN with global readout and homomorphism count features, while it is expressible by a $\sumagg$-DHN.

\section{Proofs for Section~\ref{sect:HMLC}}

An HML+C formula or term is \emph{arithmetical} if it does not use individual variables. Arithmetical HML+C-formulas and terms
coincide with arithmetical FO+C-formulas and terms as used in \cite{DBLP:journals/theoretics/Grohe24}. This is a main reason why the technical development from \cite{DBLP:journals/theoretics/Grohe24} carries over. 

We recall that, in this section, we work with activation functions that are Lipschitz continuous and piecewise linear, with all relevant parameters represented by dyadic rationals. For the sake of completeness, we make this precise.

A function $f : \mathbb{R}^d \to \mathbb{R}^{d'}$ is \emph{Lipschitz continuous} if there is some constant $\lambda$, called a \emph{Lipschitz constant} for $f$, such that for all $x,y \in \mathbb{R}^d$ it holds that\[\|f(x)-f(y)\|_{\infty} \leq \lambda \|x-y\|_{\infty}.\]A function $L : \mathbb{R} \to \mathbb{R}$ is \emph{piecewise linear} if there are $n \in \mathbb{N}$, $a_0,\ldots,a_n,b_0,\ldots,b_n,t_1,\ldots,t_n \in \mathbb{R}$ such that $t_1 < t_2 < \cdots < t_n$ and\[L(x) =\begin{cases}a_0x+b_0, & \text{if } x < t_1,\\a_ix+b_i, & \text{if } t_i \leq x < t_{i+1} \text{ for some } i < n,\\a_nx+b_n, & \text{if } x \geq t_n,\end{cases}\]if $n \geq 1$, or $L(x)=a_0x+b_0$ for all $x$ if $n=0$. We assume that  representations of a piecewise linear function $L$ always use a minimal number $n+1$ of pieces. We call $t_1,\ldots,t_n$ the \emph{thresholds} of $L$; $a_1,\dots,a_n$ the \emph{slopes},
and $b_1,\dots,b_n$ the \emph{intercepts}. $L$ is \emph{dyadic rational} if all its thresholds, slopes, and intercepts are dyadic rationals. To ensure that
$L$ is continuous, we require that
$a_{i-1}t_i+b_{i-1}=a_it_i+b_i$ for $1 < i \leq n$.
Observe that if $L$ is continuous then it is Lipschitz continuous with Lipschitz constant $\max_{0 \leq i \leq n} a_i$.

\subsection{Dyadic Rationals}

We first describe how we represent dyadic rationals in HML+C.
The same representation is used in \cite{DBLP:journals/theoretics/Grohe24}.

For \(r \in \{0,1\}\), a finite set \(I \subseteq \mathbb N\), and \(s,t \in \mathbb N\), let
\[
\llangle r,I,s,t\rrangle
:=
(-1)^r \cdot 2^{-s} \cdot \sum_{i \in I,\ i<t} 2^i.
\]
Intuitively, \(r\) is the sign bit, \(I\) specifies the binary support of the
numerator, \(s\) is the denominator exponent, and \(t\) is a bound on the
relevant bit positions.
Every dyadic rational \(a\)  has a unique \emph{canonical
representation}
\(
\operatorname{crep}(a)=(r,I,s,t)
\)
such that
\begin{enumerate}
\item \(a=\langle r,I,s,t\rangle\),
\item \(I\subseteq \{0,\dots,t-1\}\),
\item \(s=0\) or \(0\in I\),
\item \(a=0\) implies \(t=0\), and \(a\neq 0\) implies \(t-1\in I\).
\end{enumerate}
In particular,
\(
\operatorname{crep}(0)=(0,\emptyset,0,0).
\)

HML+C is a two-typed logic, and we  introduce  notation for dealing with that. Each individual (or value) variable has type $\mathsf v$, and each number variable has type $\mathsf n$. A $k$-tuple
$(z_1,\dots,z_k)$ of variables has  type $(t_1,\dots,t_k) \in \{ \mathsf v, \mathsf n \}^k$, where each $t_i$ is the type of $z_i$.
For a database $D$ and a tuple $\bar t =
(t_1,\dots,t_k) \in \{ \mathsf v, \mathsf n \}^k$, $\mn{adom}(D)^{\bar t}$ denotes the set of all
$|\bar t|$-tuples $\bar c = (c_1,\dots,c_k)$ with elements from $\mn{adom}(D) \cup \mathbb{N}$ such that 
$c_i \in \mn{adom}(D)$ if $t_i = \mathsf v$
and $c_i \in \mathbb{N}$ if $t_i = \mathsf n$.

Following \cite{DBLP:journals/theoretics/Grohe24},  we extend HML+C with \emph{relation variables} of type \(\bar t\in\{\mathsf v,\mathsf n\}^k \), denoted by letters $X,Y$,
and \emph{function variables} of type \(
\bar t \to \mathsf n
\), denoted by letters $U,V$. The former can be used in place of relation symbols in HML+C formulas, and the latter can be used in the construction of HML+C terms (which are always numerical). Relation and function variables will always be substituted away in the 
formulas that we ultimately construct. 
We shall not use second-order features, 
that is, we never quantify over relation or function variables.  Let $D$ be a database. In the presence of function and relation variables, an \emph{assignment}  over $D$ is a mapping $\alpha$ that maps every individual variable to an element of
$\mn{adom}(D)$, every number variable to an element of $\mathbb{N}$, every relation variable of type \(\bar t\in\{\mathsf v,\mathsf n\}^k \) to a relation over
$\mn{adom}(D)^{\bar t}$, and every
function variable of type \(
\bar t \to t
\) to a function from $\mn{adom}(D)^{\bar t}$ to $\mathbb{N}$. The semantics of HML+C formulas with relation and function variables is then as expected.
For a database $D$,
an assignment $\alpha$ over $D$, and an HML+C-term $t$, possibly using
relation and function variables, we write
$\llbracket t \rrbracket^{(D,\alpha)}$
to denote the element of
$\mn{adom}(D)$ that $t$ evaluates to under $\alpha$.



An \emph{r-schema of type} \(\bar t \to \mathsf r\) is a tuple
\(
Z
=
\bigl(
Z_{\mathsf{sg}},
Z_{\mathsf{Ind}},
Z_{\mathsf{dn}},
Z_{\mathsf{bd}}
\bigr),
\)
where $Z_{\mathsf{sg}}$ is a relation 
variable of type $\bar t$, $Z_{\mathsf{Ind}}$ is a relation variable of type $\mathsf n \bar t$, and $Z_{\mathsf{dn}}(\bar z), Z_{\mathsf{bd}}$ are function
variables of type $\bar t \rightarrow \mathsf n$.\footnote{The range $\mathsf r$ in  \(\bar t \to \mathsf r\) is supposed to indicate the dyadic rationals.} For a database $D$, an
assignment $\alpha$  over $D$, and a tuple $\bar c \in \mn{adom}(D)^{\bar t}$, we let
\[
r_{Z}^{D,\alpha}(\bar c)
:=
\begin{cases}
1 & \text{if } \bar c\in \alpha(Z_{\mathsf{sg}}),\\
0 & \text{otherwise,}
\end{cases}
\]
and
\[
\langle\!\langle Z \rangle\!\rangle^{(D,\alpha)}
=
(-1)^{r_{Z}^{D,\alpha}}
\cdot
2^{\alpha(Z_{\mathsf{dn}})(\bar c)}
\cdot
\sum_{\substack{
(i,\bar c) \in \alpha(Z_{\mathsf{Ind}})
\\
i<\alpha(Z_{\mathsf{bd}})
}}
2^i.
\]
Note that with \(I = \{i \in \mathbb{N} \mid (i,\bar{c}) \in \alpha(Z_{\mathsf{Ind}})\}\), \(s = \alpha(Z_{\mathsf{dn}})(\bar{c})\), and \(t = \alpha(Z_{\mathsf{bd}})(\bar{c})\), we obtain \(\llangle Z\rrangle^{(A,\alpha)}(\bar{c}) = \llangle r, I, s, t\rrangle\).

We now define r-expressions, which are similar to r-schemas except that relation and function variables are replaced by formulas and terms.
An \emph{r-expression} \(\rho(\bar z)\) of type \(\bar t \to \mathsf r\) is a tuple
\[
\rho(\bar z)
=
\bigl(
\rho_{\mathsf{sg}}(\bar z),
\rho_{\mathsf{Ind}}(i,\bar z),
\rho_{\mathsf{dn}}(\bar z),
\rho_{\mathsf{bd}}(\bar z)
\bigr),
\]
where
\(\rho_{\mathsf{sg}}(\bar z)\) is an HML\(+\)C formula,
 \(\rho_{\mathsf{Ind}}(i,\bar z)\) is an HML\(+\)C formula with one
additional free number variable \(i\),
and \(\rho_{\mathsf{dn}}(\bar z)\) and \(\rho_{\mathsf{bd}}(\bar z)\) are
HML\(+\)C terms.
Here \(\bar z\) is a tuple of variables of type~\(\bar t\). This follows the definitions in \cite{DBLP:journals/theoretics/Grohe24}, but with FO\(+\)C formulas and terms
replaced by HML\(+\)C formulas and terms.
An r-expression is \emph{arithmetical}
if it all formulas and terms in it are arithmetical.
For a database $D$, an
assignment $\alpha$  over $D$, and a tuple $\bar c \in \mn{adom}(D)^{\bar t}$, we let
\[
r_\rho^{D,\alpha}
:=
\begin{cases}
1 & \text{if } D \models \rho_{\mathsf{sg}}[\bar c],\\
0 & \text{otherwise,}
\end{cases}
\qquad\qquad
I_\rho^{D,\alpha}
:=
\left\{
i\in\mathbb N :
D \models \rho_{\mathsf{Ind}}[i,\bar c]
\right\}.
\]
and
\[
\langle\!\langle \rho \rangle \! \rangle^{(D,\alpha)}
:=
\left\langle\!\!\left\langle
r_\rho^{D,\alpha},
I_\rho^{D,\alpha},
\llbracket \rho_{\mathsf{dn}}\rrbracket^{(D,\alpha)},
\llbracket \rho_{\mathsf{bd}}\rrbracket^{(D,\alpha)}
\right\rangle\!\!\right\rangle.
\]
%

In Section~3 of \cite{DBLP:journals/theoretics/Grohe24}, Grohe develops a
machinery for performing 
arithmetic on dyadic rationals,
represented by r-schemas,
inside FO+C.  He then adapts this
machinery to GFO+C, the guarded fragment of FO with counting. The latter is relatively straightforward, mainly because most lemmas are purely arithmetical, and arithmetical FO+C formulas coincide with arithmetical GFO+C formulas, so these lemmas do
not need to be modified. Here, we adapt his machinery to HML+C,
exploiting that arithmetical FO+C formulas also coincide with
arithmetical HML+C formulas. The adaptation exactly parallels that for
GFO+C in \cite{DBLP:journals/theoretics/Grohe24}. We start with
stating two purely
arithmetical lemmas from~\cite{DBLP:journals/theoretics/Grohe24} that we are going to use as is.
\begin{lemma}[\cite{DBLP:journals/theoretics/Grohe24}, Lemma~3.18]
    \label{lem:craplemma}
    Let $Z$ be an r-schema of type $\mathsf{v} \to \mathsf{r}$ Then there is an arithmetical r-expression $\mathsf{crep}$ such that
    for all databases $D$, assignments $\alpha$ over $D$ and $a \in \mn{adom}(D)$, $\mathsf{crep}(a)$ defines the canonical representation of $\langle\!\langle Z \rangle\!\rangle^{(D,\alpha)}(a)$ in $(D,\alpha)$.
\end{lemma}
Grohe actually states his Lemma 3.18 for an r-schema $Z$ of type
$\emptyset \to \mathsf{r}$, but he freely applies it also to r-schemas
of type $\mathsf v \to \mathsf r$. 
This should be
understood
as follows. We take a fresh r-schema $Y$ of type $\emptyset \to
\mathsf r$ and apply  Lemma~\ref{lem:craplemma} to $Y$.
We then substitute in the resulting r-expression $Y$ with $Z$. This requires us to insert
an individual variable $x_v$: if $Y = (
Y_{\mathsf{sg}},
Y_{\mathsf{Ind}},
Y_{\mathsf{dn}},
Y_{\mathsf{bd}}
)$ and likewise for $Z$, then any atom
$Y_{\mathsf{sg}}()$
would be replaced
with $Z_{\mathsf{sg}}(x_v)$,
any atom $Y_{\mathsf{Ind}}(t)$
with $Y_{\mathsf{Ind}}(t,x_v)$, any constant $Y_{\mathsf{dn}}$ with the term $Z_{\mathsf{dn}}(x_v)$, and any constant $Y_{\mathsf{bd}}$ with the term $Z_{\mathsf{bd}}(x_v)$.  Using the same kind of substitution we may even apply
Lemma~\ref{lem:craplemma} to an r-expression rather than only to
an r-schema. Similar considerations apply to the lemmas below.
\begin{lemma}
[\cite{DBLP:journals/theoretics/Grohe24}, Lemma~3.19]
\label{lem:grohearithmetic}
Let $Z_1, Z_2$ be r-schemas of type $\mathsf v \to \mathsf r$.
    Then there is an arithmetical r-expression $\mathsf{add}$, $\mathsf{sub}$, and
    $\mathsf{mul}$ of type $\mathsf v \to \mathsf r$ such that for all
    databases $D$, assignments $\alpha$ over~$D$, and $v \in \mn{adom}(D)$,
    \[
     \begin{array}{r@{\;}c@{\;}l}
    \langle\!\langle \mathsf{add} \rangle\!\rangle^{(D,\alpha)}(v)
    &=& \langle\!\langle Z_1 \rangle\!\rangle^{(D,\alpha)}(v) + \langle\!\langle Z_2 \rangle\!\rangle^{(D,\alpha)}(v),
    \\[1mm]
    \langle\!\langle \mathsf{sub} \rangle\!\rangle^{(D,\alpha)}(v)
    &=& \langle\!\langle Z_1 \rangle\!\rangle^{(D,\alpha)}(v) - \langle\!\langle Z_2 \rangle\!\rangle^{(D,\alpha)}(v),
    \\[1mm]
    \langle\!\langle \mathsf{mul} \rangle\!\rangle^{(D,\alpha)}(v)
    &=& \langle\!\langle Z_1 \rangle\!\rangle^{(D,\alpha)}(v) \cdot \langle\!\langle Z_2 \rangle\!\rangle^{(D,\alpha)}(v).
    \end{array}
    \]

\end{lemma}
We now state the first lemma that needs adaptation, the HML\(+\)C analogue of
Grohe's Lemmas~3.15 for FO+C and 3.32 
for GFO+C. The `u' in term names stand for unary representation,
referring to the fact that non-negative integers are represented here
directly by a single number term.
%
\begin{lemma}
\label{lem:wrootedcounting}
%
%
Let $X$ be a relation variable of type $\mathsf v^k \mathsf n^\ell$
and let $U,V$  be function variables of type $\mathsf v^kn^\ell \to \mathsf n$ and $v^k \to \mathsf n$, respectively. Further let $F^{v_1}$ be a pointed database with $\mn{adom}(F)=\{v_1,\dots,v_k\}$.
Then 
\begin{enumerate}

\item there is a HML+C-term $\mathsf{u\text{-}itadd}(x)$  such that for all databases $D$,
assignments $\alpha$ over $D$, and $a \in \mn{adom}(D)$, 
\[
\llbracket \mathsf{u\text{-}itadd} \rrbracket^{(D,\alpha)}(a)
=
\sum_{\bar a\bar b}
\alpha(U)(\bar a\bar b)
,
\]
where the sum ranges over all $\bar a \bar b \in \mn{adom}(D)^k \times \mathbb{N}^\ell$ such that the following conditions are satisfied:
\begin{enumerate}

    \item $\bar a\bar b \in \alpha(X)$;

    \item there is a homomorphism $h$ from $F$ to $D$ with $h(v_1)=a$
and $h(v_1,\dots,v_k)=\bar a$;

    \item $\bar b =(b_1,\dots,b_\ell) \in \mathbb{N}^\ell$ with 
    $b_i < \alpha(V)(\bar a)$ for all $i$;
    
\end{enumerate}

\item there are HML+C-terms $\mathsf{u\text{-}max}(x)$ and
$\mathsf{u\text{-}min}(x)$ such
that for all databases $D$, assignments
$\alpha$ over $D$, and $a \in \mn{adom}(D)$, 
\[
\llbracket \mathsf{u\text{-}max} \rrbracket^{(D,\alpha)}(a)
=
\max_{\bar a\bar b}
\alpha(U)(\bar a)
\quad \text{ and } \quad
\llbracket \mathsf{u\text{-}min} \rrbracket^{(D,\alpha)}(a)
=
\min_{\bar a\bar b}
\alpha(U)(\bar a)
\]
where max and min range over all $\bar a \bar b \in \mn{adom}(D)^k \times \mathbb{N}^\ell$ that satisfy Conditions~(a)-(c) from Point~1.
\end{enumerate}
\end{lemma}
\begin{proof}
 To simplify notation, in the proof we assume that $\ell=1$. The generalization
 to arbitrary $\ell$ is straightforward.

For Point~1, let $\bar x =x x_2 \cdots x_k$ and set 
\[
\operatorname{u-itadd}(x)
:=
\#(x_2,\dots,x_k, y < V(\bar x), z < U(\bar x,y))
.
(\psi_F(\bar x) \wedge X(\bar x,y))
\]
where $\psi_F$ is $F$ viewed as a set of atoms  with the vertices $v_1,\dots,v_k$ represented by the
variables $x,x_2,\dots,x_k$.
This formula check conditions (a), (b), and (c) by including \(X(\bar x, y)\), \(\psi_F(\bar x)\), and \(y <V(\bar x)\).
The summation is then done with the adaptive counting bound \(z<U(\bar x, y)\).
Then
define
\[
\operatorname{u-max}(x) :=
\# (z < \operatorname{u-itadd}(x)).
(\exists x_2 \cdots \exists x_k \exists y < V(\bar x))(X(\bar x,y) \wedge z < U(\bar x,y)),
\]
%
where the existential quantifiers are an abbreviation for \[1\leq \#(x_1,\ldots, x_k, y<V(\bar x)).(\psi_F(\bar x)\land X(\bar x, y)).\]
It suffices to bound \(z\) by \(\operatorname{u-itadd}\), since it is an upper bound for \(\operatorname{u-max}\).
For \(\operatorname{u-min}\), the existential quantifier can be replaced with universal quantifiers:
\[
\operatorname{u-min}(x) :=
\# (z < \operatorname{u-itadd}(x)).
(\forall x_2 \cdots \forall x_k \forall y < V(\bar x))(X(\bar x,y) \rightarrow z < U(\bar x,y)).
\]

\end{proof}
We next replace the unary representation of non-negative integers by
binary representation. The subsequent lemma is the HML+C analogue of
Grohe's Lemma 3.16 for FO+C. He does not make explicit an analogue for
GFO+C, but it seems to be used in the (suppressed) proof of his Lemma
3.33.
It defines terms that represent the iterated addition of multiple nonnegative integers, where the number of summands is not fixed.
To formulate it, we need a representation of a non-negative integer in terms of an HML+C formula and an HML+C term. For an HML+C formula $\chi(n,\bar z)$ and an HML+C term $t(\bar z)$, we define, for
every database $D$, assignment $\alpha$ over $D$, and tuple
$\bar a\in \operatorname{adom}(D)^{\bar t}$,
\[
\llangle \chi,t\rrangle^{(D,\alpha)}(\bar a)
\;:=\;
\sum_{\substack{i<\llbracket t \rrbracket^{(D,\alpha)}(\bar a)\\ D\models \chi[i,\bar a]}} 2^i.
\]
\begin{lemma}
  \label{lem:itaddintegers}
    Let $X,Y$ be  relation variables of type $\mathsf v^k \mathsf n^\ell$
and $nv^kn^\ell$, and let $U,V$  be function variables of type $\mathsf v^k \mathsf n^\ell \to \mathsf n$ and $v^k \to \mathsf n$, respectively. Further let $F^{v_1}$ be a pointed database with $\mn{adom}(F)=\{v_1,\dots,v_k\}$.
Then there is an HML+C formula
$\mathsf{itadd}$ and an HML+C term $\mathsf{bd\text{-}itadd}(x)$  such that for all databases $D$,
assignments $\alpha$ over $D$, and $a \in \mn{adom}(D)$, 
\[
\langle\!\langle \mathsf{itadd},\mathsf{bd\text{-}itadd} \rangle\!\rangle^{(D,\alpha)}(a)
=
\sum_{\bar a\bar b}
\langle\!\langle Y,U \rangle\!\rangle
(\bar a\bar b)
,
\]
where the sum ranges over all $\bar a \bar b \in \mn{adom}(D)^k \times \mathbb{N}^\ell$ such that Conditions~(a)-(c)
from Lemma~\ref{lem:wrootedcounting} are satisfied.
\end{lemma}
\begin{proof}
To simplify notation, we assume that $\ell=1$. The generalization to arbitrary $\ell$ is straightforward.


For $a\in \operatorname{adom}(D)$, let $I_a$ be the set of all tuples
$\bar a b \in \operatorname{adom}(D)^k\times \mathbb N$ such that Conditions~(a)--(c) from
Lemma~\ref{lem:wrootedcounting} are satisfied.
For  $\bar a b\in I_a$, for brevity let
\(
n_{\bar a b}:=\llangle Y,U\rrangle^{(D,\alpha)}(\bar a,b).
\)

Set
\[
m(x):=\#(x_2,\dots,x_k,y<V(\bar x)).(\psi_F(\bar x)\wedge X(\bar x,y)),
\]
where  $\bar x = x x_2 \cdots x_k$ and $\psi_F$ is $F$ viewed as a conjunction of atoms with the vertices
$v_1,\dots,v_k$ represented by the variables $x,x_2,\dots,x_k$. Then $m(a)=|I_a|$.
For the term 
 $u\text{-}\max(x)$ provided by Lemma~\ref{lem:wrootedcounting} for the variables $X,U,V$, we have

\[
u\text{-}\max(a)=\max_{\bar a b\in I_a}\alpha(U)(\bar a,b).
\]
Hence every $n_{\bar a b}$ is smaller than $2^{u\text{-}\max(a)}$, and therefore
\(
\sum_{\bar a b\in I_a} n_{\bar a b}<2^{u\text{-}\max(a)+m(a)}.
\)
Consequently we may set
\[
\mathsf{bd}\text{-}\mathsf{itadd}(x):=u\text{-}\max(ax)+m(x).
\]

Next, for every $i\in\mathbb N$, define the HML+C term
\[
s_i(x):=\#(x_2,\dots,x_k,y<V(\bar x)).
\bigl(\psi_F(\bar x)\wedge X(\bar x,y)\wedge Y(i,\bar x,y)\bigr).
\]
Then $s_i(a)$ is the number of tuples $\bar a b\in I_a$ such that the $i$-th bit of
$n_{\bar a b}$ is $1$. Consequently,
\[
\sum_{\bar a b\in I_a} n_{\bar a b}
=
\sum_{i<\mathsf{bd}\text{-}\mathsf{itadd}(a)} s_i(a)\cdot 2^i.
\]

Now the rest is exactly as in Grohe's proof of Lemma~3.16. Using the arithmetical
construction from the proof of his Lemma~3.10, applied to the coefficients $s_i(a)$, one obtains
an arithmetical formula that expresses whether a given bit position belongs to the binary
support of the number
\[
\sum_{i<\mathsf{bd}\text{-}\mathsf{itadd}(a)} s_i(a)\cdot 2^i.
\]
Since arithmetical FO+C formulas and terms coincide with arithmetical HML+C formulas and
terms, this formula can be used unchanged here. We choose $\mathsf{itadd}$ to be the resulting HML+C formula.
\end{proof}
The subsequent Lemma~\ref{lem:summation} is the HML\(+\)C counterpart of Grohe's Lemma~3.20 for FO+C and Lemma~3.33
for GFO+C, which defines formulas for the iterated addition of dyadic rationals.
%
\begin{lemma}
\label{lem:summation}
Let $Z$ be an r-schema of type $\mathsf v^k \mathsf n^\ell \to \mathsf r$. Let $X$ be a relation variable of type $\mathsf v^k \mathsf n^\ell$,
and let $V$ be a function variable of
type $\mathsf v^k \rightarrow \mathsf n$. Furthermore, let $F^{v_1}$ be a pointed database with $\mn{adom}(F)=\{v_1,\dots,v_k\}$.
Then there is an r-expression $\mathsf{itadd}(x)$ of type $\mathsf v \to\mathsf r$ such that for all databases $D$,
assignments $\alpha$ over $D$, and $a \in \mn{adom}(D)$, we have
\[
\langle\!\langle \mathsf{itadd} \rangle\!\rangle^{(D,\alpha)}(a)
=
\sum_{\bar a\bar b}
\langle\!\langle Z \rangle\!\rangle^{(D,\alpha)}(\bar a,\bar b),
\]
where the sum ranges over all $\bar a \bar b \in \mn{adom}(D)^k \times \mathbb{N}^\ell$ such that Conditions~(a)-(c)
from Lemma~\ref{lem:wrootedcounting} are satisfied.
\end{lemma}
\begin{proof}
The proof is the same as Grohe's proof of Lemma~3.20, with Lemma~\ref{lem:wrootedcounting} and Lemma~\ref{lem:itaddintegers} in place of
his Lemmas~3.15 and~3.16.

Let
\(
Z=(Z^{\mathsf{sg}},Z^{\mathsf{Ind}},Z^{\mathsf{dn}},Z^{\mathsf{bd}}).
\)
Take a database $D$, an assignment $\alpha$, and an $a\in \operatorname{adom}(D)$. Let
\[
I_a:=\{\bar a\bar b\in \operatorname{adom}(D)^k\times \mathbb N^\ell :
\text{\(\bar a\bar b\) satisfies Conditions~(a)--(c) of Lemma~\ref{lem:wrootedcounting}}\}.
\]
To construct the desired r-expression $\mathsf{itadd}(x)$ of type $\mathsf v\to \mathsf r$, we split the index set into
\[
I_a^+:=\{\bar a\bar b\in I_a : \llangle Z\rrangle^{(D,\alpha)}(\bar a,\bar b)\ge 0\}
\quad
\text{ and }
\quad
I_a^-:=\{\bar a\bar b\in I_a : \llangle Z\rrangle^{(D,\alpha)}(\bar a,\bar b)< 0\}.
\]
It suffices to construct r-expressions $\mathsf{itadd}^+(x)$ and $\mathsf{itadd}^-(x)$ such that
\[
\llangle \mathsf{itadd}^+\rrangle^{(D,\alpha)}(a)
=
\sum_{\bar a\bar b\in I_a^+}\llangle Z\rrangle^{(D,\alpha)}(\bar a,\bar b)
\quad
\text{ and }
\quad
\llangle \mathsf{itadd}^-\rrangle^{(D,\alpha)}(a)
=
\sum_{\bar a\bar b\in I_a^-}\bigl|\llangle Z\rrangle^{(D,\alpha)}(\bar a,\bar b)\bigr|.
\]
Then the required expression $\mathsf{itadd}(x)$ is obtained by applying the subtraction expression
from Lemma~\ref{lem:grohearithmetic} to $\mathsf{itadd}^+(x)$ and $\mathsf{itadd}^-(x)$. 

We only construct $\mathsf{itadd}^+(x)$; the construction of $\mathsf{itadd}^-(x)$ is identical after
replacing each summand by its absolute value, that is, simply ignoring
its sign represented by $Z^{\mathsf{sg}}$. 
For each $\bar a\bar b\in I_a^+$, let
\[
\llangle Z\rrangle^{(D,\alpha)}(\bar a,\bar b)=m_{\bar a,\bar b}\,2^{-d_{\bar a,\bar b}},
\]
where $m_{\bar a,\bar b}\in \mathbb N$ is represented by $(Z^{\mathsf{Ind}},Z^{\mathsf{bd}})$, and
\(
d_{\bar a,\bar b}=\alpha(Z^{\mathsf{dn}})(\bar a,\bar b).
\)
Choose a common denominator by setting
\(
d^+(a):=\max\{d_{\bar a,\bar b} : \bar a\bar b\in I_a^+\}.
\)
By Point~2 of Lemma~\ref{lem:wrootedcounting} applied with
$U=Z^{\mathsf{dn}}$, $d^+(x)$ is definable by an HML+C-term.
%
%
Then
\[
m_{\bar a,\bar b}\,2^{-d_{\bar a,\bar b}}
=
\bigl(m_{\bar a,\bar b}\,2^{d^+(a)-d_{\bar a,\bar b}}\bigr)\,2^{-d^+(a)},
\]
and therefore
\[
\sum_{\bar a\bar b\in I_a^+}\llangle Z\rrangle^{(D,\alpha)}(\bar a,\bar b)
=
N^+(a)\,2^{-d^+(a)}
\quad
\text{ where }
N^+(a):=
\sum_{\bar a\bar b\in I_a^+}m_{\bar a,\bar b}\,2^{d^+(a)-d_{\bar a,\bar b}}.
\]

Now apply Lemma~\ref{lem:itaddintegers} to the family of shifted numerators
\[
m_{\bar a,\bar b}\,2^{d^+(x)-d_{\bar a,\bar b}}
\qquad
(\bar a\bar b\in I_a^+),
\]
that is, with $Y$ replaced by the bit relation for the shifted
numerator, obtained from $Z^{\mathsf{ind}}$ by shifting bit positions
by $d^+(x) - Z^{\mathsf{dn}}(\bar x,\bar y)$, and $U$ replaced by the
corresponding shifted bit bound, namely
$Z^{\mathsf{bd}}+d^+(Z)-Z^{\mathsf{dn}}(\bar x,\bar y)$.
This yields an HML+C formula $\mathsf{itadd}^{+,\mathsf{Ind}}$ and an HML+C-term
$\mathsf{itadd}^{+,\mathsf{bd}}(x)$ such that
\(
\llangle \mathsf{itadd}^{+,\mathsf{Ind}},\mathsf{itadd}^{+,\mathsf{bd}}\rrangle^{(D,\alpha)}(a)=N^+(a).
\)
We then define $\mathsf{itadd}^+(x)$ as the r-expression whose sign component is constantly $0$,
whose numerator components are
\(
(\mathsf{itadd}^{+,\mathsf{Ind}},\mathsf{itadd}^{+,\mathsf{bd}}),
\)
and whose denominator component is $d^+(x)$. By construction,
\[
\llangle \mathsf{itadd}^+\rrangle^{(D,\alpha)}(a)
=
\sum_{\bar a\bar b\in I_a^+}\llangle Z\rrangle^{(D,\alpha)}(\bar a,\bar b).
\]
%
\end{proof}   

The following lemma is a main result of Grohe's technical development,  its proof purely arithmetical and thus unaffected by the choice of FO+C, GFO+C, or HML+C. 
\begin{lemma}[\cite{DBLP:journals/theoretics/Grohe24}, Corollary~3.25]
\label{cor:grohe325}
Let
\(
f 
\)
be a function computed by an FNN that uses dyadic rational piecewise-linear activation
function and dyadic rational coefficients. Further let
\(
Z_1,\dots,Z_d
\)
be  r-schemas of type \(\mathsf v \to \mathsf r\). Then there exist
 arithmetical \(r\)-expressions
\(
\mathsf{eval}_1,\dots,\mathsf{eval}_{d'}
\)
of type \(\mathsf v \to \mathsf r\) such that for every database \(D\) and assignment $\alpha$ over $D$,
\[
f(\langle\!\langle Z_1\rangle\!\rangle^{D,\alpha},\dots,\langle\!\langle Z_d\rangle\!\rangle^{D,\alpha})
=
(\langle\!\langle
\mathsf{eval}_1\rangle \!\rangle^{D,\alpha},\dots,
 \langle\!\langle \mathsf{eval}_{d'}
 \rangle\!\rangle^{D,\alpha}).
\]
\end{lemma}
Like Lemma~\ref{lem:craplemma}, Grohe proves his Corollary 3.25
for r-schemas of type $\emptyset \to \mathsf r$. Again, the transfer
is trivial.

\subsection{Proof of Theorem~\ref{thm:DHNstoHMLplucS}}

\begin{lemma}
 \label{lem:onelayerHMLC}
Let \(Q=(F^\bullet,\mu,\mathsf{sum})\) be a homomorphism query of input dimension \(d\) and output dimension \(d'\) that uses dyadic rational piecewise-linear activation and dyadic rational coefficients. Let
\(
Z_1,\dots,Z_d
\)
be \(r\)-schemas of type \(\mathsf v\to \mathsf r\). For every database \(D\), assignment $\alpha$ over $D$, and $v\in\operatorname{adom}(D)$, let
\[
\lambda_{D,\alpha}(v):=(\langle\!\langle Z_1\rangle\!\rangle^{D,\alpha}(v),\dots,\langle\!\langle Z_d\rangle\!\rangle^{D,\alpha}(v)).
\]
Then there exist \(r\)-expressions
\(
Q\text{-}\mathsf{eval}_1(x),\dots,Q\text{-}\mathsf{eval}_{d'}(x)
\)
of type $\mathsf v \rightarrow \mathsf r$ 
such that for every database~\(D\), assignment $\alpha$ over $D$, and  \(u\in\operatorname{adom}(D)\),
\[
(\langle\!\langle Q\text{-}\mathsf{eval}_1\rangle\!\rangle^{D,\alpha}(u),\dots,\langle\!\langle Q\text{-}\mathsf{eval}_{d'}\rangle\!\rangle^{D,\alpha}(u))
=
\operatorname{eval}(Q,(D^u,\lambda_{D,\alpha})).
\]
\end{lemma}
\begin{proof}
Let  \(Q=(F^{v_1},\mu,\mathsf{sum})\), $d$, $d'$, and 
\(
Z_1,\dots,Z_d\)
be as in the lemma. Assume that $\mn{adom}(F)=\{v_1,\dots,v_k\}$.
For any $v_i \in \mn{adom}(F)$, we
may apply Lemma~\ref{cor:grohe325} to the transformation function
$\mu_{v_i}$ to obtain \(r\)-expressions
\(
\mu_v\!\operatorname{-expr}_{1},\dots,\mu_v\!\operatorname{-expr}_{d'}
\)
of type $\mathsf v \rightarrow \mathsf r$
such that for every database \(D\), assignment $\alpha$ over $D$, and
$a \in \mn{adom}(D)$,
\[
\mu_v(\langle\!\langle Z_1\rangle\!\rangle^{D,\alpha}(a),\dots,\langle\!\langle Z_d\rangle\!\rangle^{D,\alpha}(a)) =
(\langle\!\langle \mu_v\!\operatorname{-expr}_{1}\rangle\!\rangle^{D,\alpha}(a),\dots,
 \langle\!\langle \mu_v\!\operatorname{-expr}_{d'}\rangle\!\rangle^{D,\alpha}(a)).
\]
Now let $1 \leq i \leq d'$.  
Take the arithmetical r-expression $\operatorname{mul}$
from Lemma~\ref{lem:grohearithmetic} and  substitute the r-schemas $Z_1$ and $Z_2$ of
type  $\emptyset \to \mathsf r$
with the r-expressions
\(
\mu_{v_1}\!\operatorname{-expr}_{i}\)
and
\(
\mu_{v_2}\!\operatorname{-expr}_{i}\)
of type $\mathsf v \to \mathsf r$. Put the resulting r-expression
of type $\mathsf v^2 \to \mathsf r$ as a substitution for $Z_1$ into another instance of
$\operatorname{mul}$, and substitute $Z_2$
with \(
\mu_{v_3}\!\operatorname{-expr}_{i}\). Keep going
until \(
\mu_{v_k}\!\operatorname{-expr}_{i}\). Clearly, the resulting r-expression \(
\Pi\operatorname{-expr}_i(x,\bar z)
\) of type $\mathsf v^k \to \mathsf r$ is such that  for every
$\mathbb{R}^d$-embedded pointed \Sbf-database $(D^u,\lambda)$ and
every homomorphism $h$ from $F$ to $D$ with $h(v_1)=u$,
the value
\[
\langle\!\langle \Pi\operatorname{-expr}_i\rangle\!\rangle^{D,\alpha}(h(v_1),\dots,h(v_k))
\]
is exactly the \(i\)-th coordinate of the contribution of the corresponding
homomorphism to 
the evaluation
$\res((F^\bullet,\mu,\agg),(D^\bullet,\lambda))$.

To obtain from each $\Pi\operatorname{-expr}_i$
the desired r-expression 
$Q\operatorname{-eval}_i$, we use Lemma~\ref{lem:summation} in the special case $\ell=0$.
More precisely, we use $\Pi\operatorname{-expr}_i$ in place of the
r-schema $Z$ and take as relation variable $X$ a fresh
variable of type $\emptyset$ that is interpreted by $\{ () \}$, the
full nullary relation. $V$ is irrelevant since $\ell=0$, Condition~(a)
from Lemma~\ref{lem:wrootedcounting} is vacuous, and Condition~(c) disappears.
Hence, we obtain an r-expression
\(
Q\text{-eval}_j(x)
\)
such that for every database $D$, assignment $\alpha$ over $D$, and
$a\in\operatorname{adom}(D)$,
\[
\langle\!\langle Q\text{-eval}_i\rangle\!\rangle^{D,\alpha}(a)
=
\sum_{h}
\langle\!\langle \Pi\text{-expr}_i\rangle\!\rangle^{D,\alpha}(h(v_1),\dots,h(v_k)),
\]
where the sum ranges over all homomorphisms $h$ from $F$ to $D$ with
$h(v_1)=a$.
\end{proof}

\newcommand{\Nmc}{\mathcal{N}}
\newcommand{\Lmc}{\mathcal{L}}

\begin{proof}[Proof of Theorem~\ref{thm:DHNstoHMLplucS}.]
  Let $\Nmc=(\Lmc_1,\ldots,\Lmc_\ell,\mathsf{cls})$ be a sum-DHN that
  uses dyadic rational piecewise-linear activation and dyadic rational
  coefficients. For $0\leq t\leq \ell$, let $d_t$ be the output
  dimension of the $t$-th layer, where $d_0=0$. Recall that for a database
  $D$,  $\Nmc$ produces a sequence of embedded databases
\[
(D,\lambda^0_{\Nmc,D}),\ldots,(D,\lambda^\ell_{\Nmc,D}).
\]
We show by
induction on $t\in\{0,\ldots,\ell\}$ that there exist r-expressions
\(
\rho^{(t)}_1(x),\ldots,\rho^{(t)}_{d_t}(x)
\)
of type $v\to \mathsf r$ such that for every database $D$ and every $a\in\operatorname{adom}(D)$,
\[
(\llangle \rho^{(t)}_1\rrangle^{(D,\alpha)}(a),\ldots,\llangle \rho^{(t)}_{d_t}\rrangle^{(D,\alpha)}(a))
=
\lambda^t_{\Nmc,D}(a).
\]

For $t=0$, there is nothing to show.
Now let $1\leq t\leq \ell$, and suppose that
$\rho^{(t-1)}_1(x),\ldots,\rho^{(t-1)}_{d_{t-1}}(x)$ have already been constructed. Let
\(
\Lmc_t=(\mathcal F_t,\mathsf{com}_t),
\)
where
\(
\mathcal F_t=(Q_{t,1},\ldots,Q_{t,m_t})
\)
and each
\(
Q_{t,s}=(F^\bullet_{t,s},\mu_{t,s},\mathsf{sum})
\)
is a homomorphism query of input dimension $d_{t-1}$ and output dimension $e_{t,s}$.
Take any $s$ with $1 \leq s \leq m_t$. Apply Lemma~\ref{lem:onelayerHMLC} to the query $Q_{t,s}$ and to the r-schemas represented by
\(
\rho^{(t-1)}_1(x),\ldots,\rho^{(t-1)}_{d_{t-1}}(x)
\)
to obtain r-expressions
\(
\sigma^{(t,s)}_1(x),\ldots,\sigma^{(t,s)}_{e_{t,s}}(x)
\)
of type $\mathsf v \to \mathsf r$ such that for every database $D$ and every $a\in\operatorname{adom}(D)$,
\[
(\llangle \sigma^{(t,s)}_1\rrangle^{(D,\alpha)}(a),\ldots,
\llangle \sigma^{(t,s)}_{e_{t,s}}\rrangle^{(D,\alpha)}(a))
=
\operatorname{eval}(Q_{t,s},(D_a,\lambda^{t-1}_{\Nmc,D})).
\]

Recall that $\mathsf{com}_t$ is represented by an FNN with dyadic rational piecewise-linear activation and dyadic
rational coefficients. Hence, by Lemma~\ref{cor:grohe325}, applied to the family of r-expressions
\(
\sigma^{(t,s)}_j(x)
\) with $1 \leq s \leq m_t$ and $1 \leq j \leq e_{t,s}$,
there exist r-expressions
\(
\rho^{(t)}_1(x),\ldots,\rho^{(t)}_{d_t}(x)
\)
of type $\mathsf v \to \mathsf r$ such that for every database $D$ and every $a\in\operatorname{adom}(D)$,
\[
\begin{array}{r@{}l}
(\llangle \rho^{(t)}_1\rrangle^{(D,\alpha)}(a),&\ldots,
\llangle \rho^{(t)}_{d_t}\rrangle^{(D,\alpha)}(a)) \\[1mm]
& =
\mathsf{com}_t\bigl(
\operatorname{eval}(Q_{t,1},(D_a,\lambda^{t-1}_{\Nmc,D})),\ldots,
\operatorname{eval}(Q_{t,m_t},(D_a,\lambda^{t-1}_{\Nmc,D}))
\bigr) = \lambda^t_{\Nmc,D}(a).
\end{array}
\]
This completes the induction
step.

We have thus obtained r-expressions
\(
\rho^{(\ell)}_1(x),\ldots,\rho^{(\ell)}_{d_\ell}(x)
\)
representing the final embedding $\lambda^\ell_{\Nmc,D}(a)$. It remains to deal with the classification
function $\mathsf{cls}:\mathbb R^{d_\ell}\to\{0,1\}$. Since $\mathsf{cls}$ is again represented by an
FNN that uses dyadic rational piecewise-linear activation and dyadic rational coefficients, Lemma~\ref{cor:grohe325} yields an
r-expression
\(
\mathsf{cls}\text{-}\mathsf{eval}(x)
\)
of type $\mathsf v \to \mathsf r$ such that for every database $D$ and every $a\in\operatorname{adom}(D)$,
\[
\llangle \mathsf{cls}\text{-}\mathsf{eval}\rrangle^{(D,\alpha)}(a)
=
\mathsf{cls}(\lambda^\ell_{\Nmc,D}(a)).
\]
Since the codomain of $\mathsf{cls}$ is $\{0,1\}$, this value is always either $0$ or $1$. Applying Lemma~\ref{lem:craplemma}, we may assume that
$\mathsf{cls}\text{-}\mathsf{eval}(x)$ is in canonical representation. $\Nmc$ accepts $D_a$ if and only if
\(
\llangle \mathsf{cls}\text{-}\mathsf{eval}\rrangle^{(D,\alpha)}(a)=1,
\)
and, by canonicity, $\Nmc$ is thus equivalent to
\[
\varphi_\Nmc(x):=
\neg \mathsf{cls}\text{-}\mathsf{eval}^{\mathsf{sg}}(x)
\wedge
\mathsf{cls}\text{-}\mathsf{eval}^{\mathsf{dn}}(x)=0
\wedge
\mathsf{cls}\text{-}\mathsf{eval}^{\mathsf{bd}}(x)=1
\wedge
\mathsf{cls}\text{-}\mathsf{eval}^{\mathsf{Ind}}(0,x).
\]
\end{proof}

\section{Details on Experiments}\label{app:experiments}

Both datasets, the code used to generate them, as well as our implementation of DHNs and the experiment code are available in the supplementary material.

\subsection{Local Transitivity Dataset} 

The local transitivity is a synthetic dataset for vertex classification in a directed graph generated such that the positive examples are locally transitive, and the negative examples are not.
To produce this dataset, we start with a base graph that consists of 200 disjoint copies
of a linear-order graph with 20 vertices. In this base graph, all $200 \cdot 20 = 4000$ vertices are locally transitive. We then randomly delete $4000$ edges, resulting 1921 vertices that are locally transitive, which we pick as positive examples, and 2079 vertices that are not locally transitive, which we pick as negative examples.
In this dataset, every vertex is labeled with a 1-dimensional vector that is constant.

As an example, Figure~\ref{fig:lt-example} shows an excerpt of the first six vertices (and related edges) of this graph. Vertex classes are indicated by $+$ and $-$.

\begin{figure}
\centering
\begin{tikzpicture}
    \node[draw, circle, inner sep = 0pt] (0) at (0, 0) {$+$};
    \node[draw, circle, inner sep = 0pt] (1) at (1, 0) {$-$};
    \node[draw, circle, inner sep = 0pt] (2) at (2, 0) {$-$};
    \node[draw, circle, inner sep = 0pt] (3) at (3, 0) {$+$};
    \node[draw, circle, inner sep = 0pt] (4) at (4, 0) {$+$};
    \node[draw, circle, inner sep = 0pt] (5) at (5, 0) {$+$};

    \draw[->] (0) to (1);
    \draw[->, bend left] (0) to (2);
    \draw[->, bend left] (0) to (3);
    \draw[->, bend left] (0) to (4);
    \draw[->, bend left] (0) to (5);
    \draw[->] (1) to (2);
    \draw[->, bend right] (1) to (4);
    \draw[->, bend right] (1) to (5);
    \draw[->] (2) to (3);
    \draw[->, bend left] (2) to (4);
    \draw[->] (3) to (4);
    \draw[->, bend right] (3) to (5);
    \draw[->] (4) to (5);
    \end{tikzpicture} 
    \caption{Except from the local transitivity dataset with node classifications}
    \label{fig:lt-example}
\end{figure}
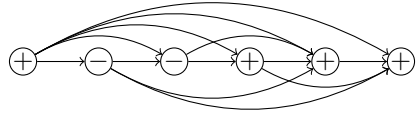

\subsection{Sun Property Dataset}

The sun property dataset is a synthetic dataset for vertex classification in undirected graphs in which positive examples satisfy the sun property (formalized as $\varphi_{\sun}$ in Section~\ref{sect:DENs}) and negative examples do not. Recall that the sun property demands a vertex to lie on a 6-cycle on which every vertex has a neighbor with degree 1.
Thus, the vertices of the 6-cycle of the graph in Figure~\ref{fig:sp-example} satisfy this property, and one can obtain negative examples by randomly attaching additional neighbors to the outer vertices of that graph.
From these basic examples, we obtain more diverse examples that also contain multiple 6-cycles by attaching random negative examples to the cycle vertices of basic examples. 

The dataset is then the disjoint union of 100 such cycles that satisfy the property and 100 such cycles that do not satisfy the property, where we increase the degree of the neighbors to at most 8, resulting in 600 positive example vertices and 600 negative example vertices. Note that a large part of the vertices in this dataset do not serve as examples, and that again all vertices are labeled with the constant 1-vector.

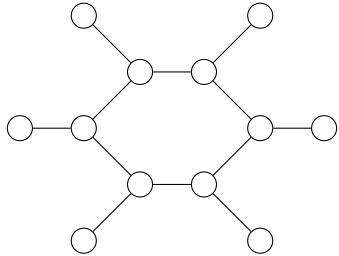
\begin{figure}
\centering
\begin{tikzpicture}[node distance = 0.5cm and 0.5cm]
    \node[draw, circle] (0) at (0, 0) {};
    \node[draw, circle, above left  = of 0] (0b) {};
    \node[draw, circle, right = of 0] (1) {};
    \node[draw, circle, above right = of 1] (1b) {};
    \node[draw, circle, below right = of 1] (2) {};
    \node[draw, circle, right  = of 2] (2b) {};
    \node[draw, circle, below left = of 2] (3)  {};
        \node[draw, circle, below right  = of 3] (3b) {};
    \node[draw, circle, left = of 3] (4)  {};
        \node[draw, circle, below left  = of 4] (4b) {};
    \node[draw, circle, above left = of 4] (5)  {};
        \node[draw, circle, left  = of 5] (5b) {};
    \draw[-] (0) to (1);
    \draw[-] (0) to (0b);

    \draw[-] (1) to (2);
        \draw[-] (1) to (1b);

    \draw[-] (2) to (3);
        \draw[-] (2) to (2b);

    \draw[-] (3) to (4);
        \draw[-] (3) to (3b);

    \draw[-] (4) to (5);
        \draw[-] (4) to (4b);
    \draw[-] (5) to (0);
        \draw[-] (5) to (5b);
    \end{tikzpicture} 
    \caption{Positive example of the sun property}
    \label{fig:sp-example}
\end{figure}

\subsection{DHN Implementation}

Our implementation of DHNs largely follows the one of \cite{maehara2024deep} and uses PyTorch 2.8.0 and PyTorch Geometric 2.7.0. The full code is available in the supplementary material. We compute the required homomorphisms using a naive backtracking algorithm, which proved sufficient for our homomorphism queries and datasets.
The transformation functions are trainable FNNs with a single hidden layer and a hidden/output dimension of 32. The combination functions are trainable FNNs with 3 hidden layers of dimension 32. All FNNs use leaky ReLU as their activation function. We include node-wise LayerNorm in-between DHN layers.
For vertex classification, a final FNN with one hidden layer projects the vertex embedding to a single dimension.

For the experiment on the local transitivity dataset we equip DHNs with 13 homomorphism queries with up to 3 vertices. These are chosen to represent all injective-homomorphically different graphs with up to 3 vertices that have matches in the dataset. They are depicted in Figure~\ref{fig:lt-hom-queries}, where the root of the homomorphism query is highlighted in black. 

\begin{figure}
\begin{tikzpicture}[baseline=(current bounding box.north)]
  \node[circle, fill] (0) at (0, 0) {};
  \node[circle, draw] (1) at (0, -1) {};
  \draw[->] (0) to (1);
\end{tikzpicture} \quad
\begin{tikzpicture}[baseline=(current bounding box.north)]
  \node[circle, fill] (0) at (0, 0) {};
  \node[circle, draw] (1) at (0, -1) {};
  \draw[->] (1) to (0);
\end{tikzpicture} \quad
\begin{tikzpicture}[baseline=(current bounding box.north)]
  \node[circle, fill] (0) at (0, 0) {};
  \node[circle, draw] (1) at (0, -1) {};
  \node[circle, draw] (2) at (0, -2) {};
  \draw[->] (0) to (1);
  \draw[->] (1) to (2);
\end{tikzpicture} \quad
\begin{tikzpicture}[baseline=(current bounding box.north)]
  \node[circle, draw] (0) at (0, 0) {};
  \node[circle, fill] (1) at (0, -1) {};
  \node[circle, draw] (2) at (0, -2) {};
  \draw[->] (0) to (1);
  \draw[->] (1) to (2);
\end{tikzpicture} \quad
\begin{tikzpicture}[baseline=(current bounding box.north)]
  \node[circle, draw] (0) at (0, 0) {};
  \node[circle, draw] (1) at (0, -1) {};
  \node[circle, fill] (2) at (0, -2) {};
  \draw[->] (0) to (1);
  \draw[->] (1) to (2);
\end{tikzpicture} \quad
\begin{tikzpicture}[baseline=(current bounding box.north)]
  \node[circle, draw] (0) at (0, 0) {};
  \node[circle, fill] (1) at (0, -1) {};
  \node[circle, draw] (2) at (0, -2) {};
  \draw[->] (0) to (1);
  \draw[->] (2) to (1);
\end{tikzpicture} \quad
\begin{tikzpicture}[baseline=(current bounding box.north)]
  \node[circle, fill] (0) at (0, 0) {};
  \node[circle, draw] (1) at (0, -1) {};
  \node[circle, draw] (2) at (0, -2) {};
  \draw[->] (0) to (1);
  \draw[->] (2) to (1);
\end{tikzpicture} \quad
\begin{tikzpicture}[baseline=(current bounding box.north)]
  \node[circle, fill] (0) at (0, 0) {};
  \node[circle, draw] (1) at (0, -1) {};
  \node[circle, draw] (2) at (0, -2) {};
  \draw[->] (1) to (2);
  \draw[->] (1) to (0);
\end{tikzpicture} \quad
\begin{tikzpicture}[baseline=(current bounding box.north)]
  \node[circle, draw] (0) at (0, 0) {};
  \node[circle, fill] (1) at (0, -1) {};
  \node[circle, draw] (2) at (0, -2) {};
  \draw[->] (1) to (2);
  \draw[->] (1) to (0);
\end{tikzpicture}
\begin{tikzpicture}[baseline=(current bounding box.north)]
  \node[circle, fill] (0) at (0, 0) {};
  \node[circle, draw] (1) at (-0.66, -1) {};
  \node[circle, draw] (2) at (0.66, -1) {};
  \draw[->] (0) to (1);
  \draw[->] (1) to (2);
  \draw[->] (0) to (2);
\end{tikzpicture}
\begin{tikzpicture}[baseline=(current bounding box.north)]
  \node[circle, draw] (0) at (0, 0) {};
  \node[circle, fill] (1) at (-0.66, -1) {};
  \node[circle, draw] (2) at (0.66, -1) {};
  \draw[->] (0) to (1);
  \draw[->] (1) to (2);
  \draw[->] (0) to (2);
\end{tikzpicture}
\begin{tikzpicture}[baseline=(current bounding box.north)]
  \node[circle, draw] (0) at (0, 0) {};
  \node[circle, draw] (1) at (-0.66, -1) {};
  \node[circle, fill] (2) at (0.66, -1) {};
  \draw[->] (0) to (1);
  \draw[->] (1) to (2);
  \draw[->] (0) to (2);
\end{tikzpicture}
\begin{tikzpicture}[baseline=(current bounding box.north)]
  \node[circle, fill] (0) at (0, 0) {};
  \node[circle, draw] (1) at (-0.66, -1) {};
  \node[circle, draw] (2) at (0.66, -1) {};
  \draw[->] (0) to (1);
  \draw[->] (1) to (2);
  \draw[->] (2) to (0);
\end{tikzpicture}
\caption{Homomorphism queries used in the local transitivity experiment}
\label{fig:lt-hom-queries}
\end{figure}
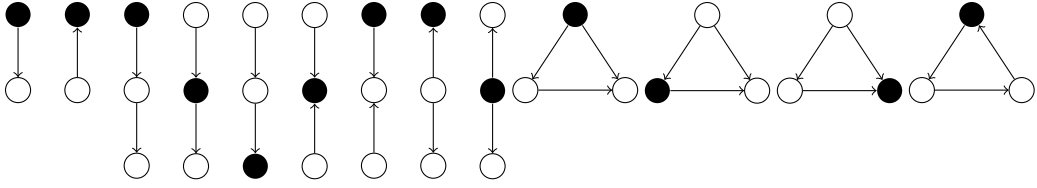

\subsection{Training} 

For training, we split the dataset into training, validation, and test set using a standard 60\%/20\%/20\% split. We train all models using the Adam optimizer on the training set as a single batch, and employ binary cross entropy with logits as a loss function. The DHN variants we train for 1000 epochs with a learning rate of 0.0003. Other GNN variants we train for 10000 epochs with a learning rate of 0.001, as we observed slower convergence for these models as for DHNs.
We do not employ regularization measures such as weight decay.

For GCN, GraphSage, and GIN we use the provided implementations in PyTorch Geometric 2.7.0,
also using LayerNorm and leaky ReLU. For GIN we set the parameter $\varepsilon$ to be trainable.

All experiments were performed on the CPU of a 2021 MacBook Pro with M1 Pro chip. The longest training run took 83 minutes.

\subsection{Evaluation}

We compute F1-score and AUROC of the trained models  on the validation and test sets using available functions in scikit-learn 1.6.1. Due to some observed variance in results for all architectures, we perform 3 training runs each, and select the model with the highest F1-score on the validation set. The F1-score and AUROC of this selected model on the test set are then reported in Tables~\ref{tab:exp-results-lt} and~\ref{tab:exp-results-sp}.

To support that our results are statistically significant, we additionally report mean and standard error on the validation set over the three training runs.

On the local transitivity dataset:
\begin{center}
\begin{tabular}{lrr}\toprule
        Architecture & Mean F1 (SE) & Mean AUROC (SE)\\\midrule
            \sumagg-DHN & 0.96 (0.03)& 0.99 (0.01) \\
            \maxagg-DHN & 0.83 (0.02) & 0.93 (0.01) \\
            GCN & 0.76 (0.01) & 0.86 (0.01) \\
            GraphSAGE & 0.79 (0.01) & 0.91 (0.01) \\
            GIN & 0.93 (0.01) & 0.98 (0.00) \\
            \bottomrule
\end{tabular}
\end{center}

And on the sun dataset:
\begin{center}
\begin{tabular}{lrr}\toprule
        Architecture & Mean F1 (SE) & Mean AUROC (SE) \\\midrule
            \sumagg-DHN & 1.00 (0.00) & 1.00 (0.00) \\
            GIN with hom. counts & 0.77 (0.06) & 0.82 (0.06) \\
            GCN & 0.78 (0.01) & 0.85 (0,01) \\
            GraphSAGE & 0.64 (0.02) & 0.50 (0.00) \\
            GIN & 0.63 (0.15) & 0.65 (0.19) \\
            \bottomrule
 \end{tabular}
 \end{center}
 

\end{document}